\documentclass[aps,groupedaddress,superscriptaddress,amsmath,amssymbol,stfloats,nofootinbib]{revtex4}

\usepackage[latin1]{inputenc}
\usepackage{amsthm}
\usepackage{amssymb}
\usepackage{graphicx}
\usepackage{amsmath}
\usepackage{bbold}
\usepackage{bbm}
\usepackage{nonfloat}
\usepackage{color}
\usepackage[pdftex]{hyperref}
\usepackage{braket}
\usepackage{graphicx}
\usepackage{dsfont}
\usepackage{mathdots}
\usepackage{wasysym}

\DeclareGraphicsExtensions{.jpeg, .png, .pdf}

\newtheorem{thm}{Theorem}
\newtheorem*{thm*}{Theorem}
\newtheorem{prop}[thm]{Proposition}
\newtheorem*{prop*}{Proposition}
\newtheorem{lemma}[thm]{Lemma}
\newtheorem*{lemma*}{Lemma}
\newtheorem{cor}[thm]{Corollary}
\newtheorem*{cor*}{Corollary}

\newtheorem*{cj*}{Conjecture}
\newtheorem{Def}[thm]{Definition}
\newtheorem*{Def*}{Definition}
\newtheorem{iss}{Issue}

\theoremstyle{definition}
\newtheorem*{rem}{Remark}
\newtheorem*{note}{Note}

\newtheorem{ex}[thm]{Example}
\newtheorem*{ex*}{Example}

\newcommand{\bq}{\begin{equation}}
\newcommand{\eq}{\end{equation}}
\newcommand{\bqq}{\begin{equation*}}
\newcommand{\eqq}{\end{equation*}}
\newcommand{\bal}{\begin{align}}
\newcommand{\eal}{\end{align}}
\newcommand{\ball}{\begin{align*}}
\newcommand{\eall}{\end{align*}}
\newcommand{\Tr}{\text{Tr}\,}
\newcommand{\tmin}{\ensuremath \! \raisebox{2pt}{$\underset{\begin{array}{c} \vspace{-3.7ex} \\ \text{\scriptsize min} \end{array}}{\otimes}$}\! }
\newcommand{\tmax}{\ensuremath \! \raisebox{2pt}{$\underset{\begin{array}{c} \vspace{-3.7ex} \\ \text{\scriptsize max} \end{array}}{\otimes}$}\! }
\newcommand{\tminit}{\ensuremath \! \raisebox{2pt}{$\underset{\begin{array}{c} \vspace{-3.7ex} \\ \text{\scriptsize \emph{min}} \end{array}}{\otimes}$}\! }

\newcommand{\tminfoot}{\ensuremath \! \raisebox{2.1pt}{$\scriptstyle\underset{\begin{array}{c} \vspace{-4.1ex} \\ \text{\fontsize{2}{4}\selectfont min} \end{array}}{\otimes}$}\!}
\newcommand{\tminitfoot}{\ensuremath \! \raisebox{2.1pt}{$\scriptstyle\underset{\begin{array}{c} \vspace{-4.1ex} \\ \text{\fontsize{2}{4}\selectfont \emph{min}} \end{array}}{\otimes}$}\!}

\newcommand*{\coloneqq}{\mathrel{\vcenter{\baselineskip0.5ex \lineskiplimit0pt \hbox{\scriptsize.}\hbox{\scriptsize.}}} =}

\makeatletter
\newcommand*\rel@kern[1]{\kern#1\dimexpr\macc@kerna}
\newcommand*\widebar[1]{%
  \begingroup
  \def\mathaccent##1##2{%
    \rel@kern{0.8}%
   \overline{\rel@kern{-0.8}\raisebox{0pt}[1.1\height]{$\m@th\macc@nucleus$}\rel@kern{0.2}}%
    \rel@kern{-0.2}%
  }%
  \macc@depth\@ne
  \let\math@bgroup\@empty \let\math@egroup\macc@set@skewchar
  \mathsurround\z@ \frozen@everymath{\mathgroup\macc@group\relax}%
  \macc@set@skewchar\relax
  \let\mathaccentV\macc@nested@a
  \macc@nested@a\relax111{#1}%
  \endgroup
}
\makeatother
                     
\begin{document}

\title{Ultimate data hiding in quantum mechanics and beyond}
%\date{December 3, 2017}

\author{Ludovico Lami}
\affiliation{F\'{i}sica Te\`{o}rica: Informaci\'{o} i Fen\`{o}mens Qu\`{a}ntics, Departament de F\'{i}sica, Universitat Aut\`{o}noma de Barcelona, 08193 Bellaterra, Barcelona, Spain.}

\author{Carlos Palazuelos}
\affiliation{Departamento de An\'{a}lisis y Matem\'{a}tica Aplicada, Universidad Complutense de Madrid, 28040 Madrid, Spain.}
\affiliation{Instituto de Ciencias Matem\'{a}ticas, C/ Nicol\'{a}s Cabrera 13-15, 28049 Madrid, Spain.}

\author{Andreas Winter}
\affiliation{F\'{i}sica Te\`{o}rica: Informaci\'{o} i Fen\`{o}mens Qu\`{a}ntics, Departament de F\'{i}sica, Universitat Aut\`{o}noma de Barcelona, ES-08193 Bellaterra, Barcelona, Spain.}
\affiliation{ICREA -- Instituci\'o Catalana de Recerca i Estudis Avan\c{c}ats, Pg. Lluis Companys 23, 08010 Barcelona, Spain}

\begin{abstract}
The phenomenon of data hiding, i.e. the existence of pairs of states of a bipartite system that are perfectly distinguishable via general entangled measurements yet almost indistinguishable under LOCC, is a distinctive signature of nonclassicality. The relevant figure of merit is the maximal ratio (called data hiding ratio) between the distinguishability norms associated with the two sets of measurements we are comparing, typically all measurements vs LOCC protocols. For a bipartite $n\times n$ quantum system, it is known that the data hiding ratio scales as $n$, i.e. the square root of the real dimension of the local state space of density matrices. We show that for bipartite $n_A\times n_B$ systems the maximum data hiding ratio against LOCC protocols is $\Theta\left(\min\{n_A,n_B\}\right)$. This scaling is better than the previously obtained upper bounds $O\left(\sqrt{n_A n_B}\right)$ and $\min\{n_A^2, n_B^2\}$, and moreover our intuitive argument yields constants close to optimal.

In this paper, we investigate data hiding in the more general context of general probabilistic theories (GPTs), an axiomatic framework for physical theories encompassing only the most basic requirements about the predictive power of the theory. The main result of the paper is the determination of the maximal data hiding ratio obtainable in an arbitrary GPT, which is shown to scale linearly in the minimum of the local dimensions. We exhibit an explicit model achieving this bound up to additive constants, finding that the quantum mechanical data hiding ratio is only of the order of the square root of the maximal one. Our proof rests crucially on an unexpected link between data hiding and the theory of projective and injective tensor products of Banach spaces. Finally, we develop a body of techniques to compute data hiding ratios for a variety of restricted classes of GPTs that support further symmetries.
\end{abstract}

%Despite being well-understood in quantum mechanics, data hiding has never been addressed from a more fundamental point of view.

\maketitle

\section{Introduction}

Quantum information one may truly define as an ingenious investigation into the operational properties that differentiate classical systems from quantum systems. We believe it important from a fundamental perspective, to extend this definition so as to encompass all the properties discriminating classical from generically \emph{non-classical} systems~\cite{lamiatesi}. A useful framework for treating physical theories from a more general point of view, without losing the operational interpretation associated with them, is known as general probabilistic theories (GPTs, for short). Intuitively, one can think of GPTs as theories that are analogous to quantum mechanics in many ways, but with the fundamental difference that the cone of (unnormalised) states is no longer that of positive semidefinite matrices, but rather a generic convex cone in a finite-dimensional real vector space. In the rich landscape provided by GPTs, we investigate one particular phenomenon appearing in non-classical theories, \emph{data hiding}~\cite{dh original 1, dh original 2}. 

In quantum mechanics, data hiding is usually intended as the existence of pairs of states of a bipartite system that are perfectly distinguishable with global measurements yet almost indistinguishable when only protocols involving local operations and classical communication (LOCC) are allowed. We extend the relevant definitions to encompass more general form of data hiding in arbitrary GPTs (Definition~\ref{dh}). In this context, the effectiveness of discriminating protocols is measured by the minimal probability of error $P_e^{\mathcal{M}}(\rho,\sigma; p)$ in the task of distinguishing the two states $\rho,\sigma$ with a priori probabilities $p,1-p$ respectively, when only operations from the class $\mathcal{M}$ are available. For instance, the archetypical example of a pair of states exhibiting data hiding is given by the normalised projectors onto the symmetric and antisymmetric subspace in $\mathds{C}^{n}\otimes\mathds{C}^{n}$, denoted by $\rho_{\mathcal{S}}$ and $\rho_{A}$, respectively. While $P_e^{\text{ALL}}\left(\rho_{S}, \sigma_{\mathcal{A}}, \frac12\right) = 0$, because the two states have orthogonal support, it can be shown that $P_e^{\text{LOCC}}\left(\rho_{S}, \sigma_{\mathcal{A}}, \frac12\right) = \frac{2}{n+1}$~\cite{dh original 1, dh original 2}. 

For these discussions, a more convenient quantity is the \emph{distinguishability norm} associated with $\mathcal{M}$, denoted by $\|\cdot\|_{\mathcal{M}}$ and defined by $P_{e}^{\mathcal{M}}(\rho,\sigma; p)\, =\, \frac12 \left( 1 - \| p\rho-(1-p)\sigma\|_{\mathcal{M}}  \right)$, as detailed in Definition~\ref{def d norm} and Lemma~\ref{discr GPT}; it quantifies the advantage of making an observation over pure guessing (i.e. the prior information). It is immediately obvious that higher values of $\| p\rho-(1-p)\sigma\|_{\mathcal{M}}$ correspond to an increased discriminating power of the set $\mathcal{M}$. Thus, the central object of our investigation is a quantity that we name \emph{data hiding ratio}, which depends on the GPT as well as on the restricted set of measurements we consider, and is given by
\begin{equation*}
R(\mathcal{M}) = \max \frac{\| p\rho-(1-p)\sigma\|_{\text{ALL}}}{\| p\rho-(1-p)\sigma\|_{\mathcal{M}}},
\end{equation*}
where $\|\cdot\|_{\text{ALL}}$ denotes the norm associated with the whole set of possible measurements and the maximisation ranges over all pairs of states $\rho,\sigma$ and a priori probabilities $p$ (Definition~\ref{dh} and Proposition~\ref{dh ratio}).

Our first interest is the determination of the value of $R(\mathcal{M})$ for relevant classes of models and physically or operationally significant restricted sets of protocols $\mathcal{M}$. Concretely, we focus on the case when the system under examination is bipartite, and $\mathcal{M}$ is the set of operations that are implementable if various kind of \emph{locality constraints} are imposed (Definition~\ref{locally constr}). We will look at local operations assisted by one-way classical communication ($\text{LOCC}_\rightarrow$) or the generally broader set of measurements with separable effects (called separable measurements and denoted by $\text{SEP}$). Throughout the rest of this section, $\mathcal{M}$ will always denote one of these locally restricted sets of measurements.

Our interest in this kind of questions has been spurred by a result in~\cite{VV dh}, stating that on a finite-dimensional quantum mechanical system $\mathds{C}^{n_A}\otimes\mathds{C}^{n_B}$ the data hiding ratio against LOCC protocols satisfies $\Omega\left(n\right) \leq R_{\text{QM}}(\text{LOCC}) \leq O \left( \sqrt{n_A n_B} \right)$, where $n = \min\{n_A, n_B\}$. In particular, when $n_A=n_B$ one obtains $R_{\text{QM}}(\text{LOCC})=\Theta(n)$. Our first contribution is an intuitive argument, using the quantum teleportation protocol, to show that $R_{\text{QM}}(\text{LOCC})=\Theta(n)$ holds in fact for all $n_A, n_B$ (Theorem~\ref{dh QM}). Note that the local real dimensions of the cones of states (i.e. the cones of positive semidefinite matrices) satisfy $d_A=n_A^2$ and $d_B=n_B^2$, so that one could write $R_{\text{QM}}(\text{LOCC})= \min\{ \sqrt{d_A}, \sqrt{d_B} \}$.
Along the same line of thought, we compute the data hiding ratios against locally constrained sets of measurements for another significant GPT called spherical model (Example~\ref{ex sph}). As the name suggests, the state space of such a model is a Euclidean ball. We prove that for fixed local real dimensions, the data hiding ratio displayed by the spherical model is quadratically larger than the quantum mechanical one, i.e. $R_{\text{Sph}}(\mathcal{M}) = \Theta \left( \min\{ d_A, d_B \} \right)$.

Our second line of investigation aims at establishing an \emph{optimal, universal upper bound} for $R(\mathcal{M})$ that depends only on the local real dimensions $d_A, d_B$, and where $\mathcal{M}$ is a locally constrained set of measurements, as usual. The answer to this question is the content of the main result of the present paper, Theorem~\ref{thm univ}, which states that $R(\mathcal{M})\leq \min\{d_A, d_B\}$ holds for all bipartite GPTs of local dimensions $d_A, d_B$. Since we have seen that such a scaling characterises spherical models, we deduce that $\min\{d_A, d_B\}$ is the optimal universal upper bound on the data hiding ratio against locally constrained sets of measurements. This answers our fundamental question on the ultimate effectiveness of data hiding when the local systems have bounded size.

The rest of the paper is organised as follows. In Section~\ref{sec GPT} we set the notation and provide a brief introduction to the main feature of the GPT framework. Section~\ref{sec dh GPTs} is devoted to giving precise definitions of what we mean by data hiding in the GPT setting. Then, throughout Section~\ref{sec ex} we show this machinery at work by examining some examples of significant physical theories, including quantum mechanics. In Section~\ref{sec univ} we present the main result of the paper, i.e. the determination of the optimal universal upper bound on data hiding ratios against locally constrained sets of measurements (Theorem~\ref{thm univ}). Finally, Section~\ref{sec special} presents a body of techniques to compute data hiding ratios in specific classes of GPT models satisfying some further assumptions.

\section{General probabilistic theories} \label{sec GPT}

That branch of the study of cones, which turns into the study of ordered Banach spaces, gives rise to an axiomatic framework to describe probabilistic physical models. Since it has been initially proposed as a way to generalise quantum probability rules~\cite{MACKEY, Ellis dual base, Davies 1970}, this framework, now widely known under the name of general probabilistic theories (GPTs), has received growing attention and many important refinements were put forward~\cite{Edwards 1, Edwards 2, LUDWIG, tensor rule 1, tensor rule 2}. The main advantage of GPTs is that within their realm a rich variety of physical theories can be encompassed. Prominent examples include classical probability theory and quantum mechanics as well as more abstract objects such as Popescu-Rohrlich boxes~\cite{PR boxes, Barrett original}. Nowadays, many intriguing features of quantum information processing have been investigated in the more general context provided by GPTs~\cite{Barnum no-broad, telep in GPT, nonloc polygon, Barnum steering, Pfister no disturbance, ent therm GPT}.

For the sake of completeness, throughout this section we provide an overview of the mathematical machinery of GPTs. For further details we refer the reader to more complete reviews that can be found in the literature~\cite{telep in GPT, Pfister Master, Barnum review, lamiatesi}. In conformity with these authors, we adopt the so-called {\it abstract state space} formalism.

\subsection{Single systems}

We shall describe a physical system by means of a {\it state space} $\Omega$, and measurement outcomes (called {\it effects}) by functions $e:\Omega\rightarrow[0,1]$, were $e(\omega)$ stands for the probability of getting the outcome $e$ on the state $\omega\in\Omega$. If we represent the probabilistic process of preparing the state $\omega_{1}\in\Omega$ with probability $p$ and the state $\omega_{2}\in\Omega$ with probability $1-p$ as $p\,\omega_{1}+(1-p)\omega_{2}$, and we require this to be a physically allowable state, we obtain that the state space must be convex. The compatibility with the rule determining how the probabilities of measurement outcomes are formed forces $e$ to be a convex-linear functional. A special effect is the {\it unit effect} $u$, defined by $u(\omega)\equiv 1$ for all $\omega\in\Omega$. Following a common convention, the set of all effects will be denoted with $[0,u]$. We shall see in a moment that this is really an interval with respect to a natural ordering on the dual space. In this context, a measurement (usually called positive operator-valued measurement or POVM in quantum mechanics) is a {\it finite} collection of effects $(e_{i})_{i\in I}\subset [0,u]$ such that $\sum_{i\in I} e_{i}=u$. We shall denote the set of all measurements in a GPT as $\mathbf{M}$.

A more convenient picture can be obtained by exploiting the natural embedding of $\Omega$ in the dual of the vector space $\mathcal{A}_{\Omega}$ of affine functionals $f:\Omega\rightarrow \mathds{R}$. Namely, to a state $\omega\in\Omega$ we can associate $\hat{\omega}\in \mathcal{A}_{\Omega}^{*}$ whose action is given by $\hat{\omega}(f)=f(\omega)$, for all $f\in \mathcal{A}_{\Omega}$. A common technical assumption we will make everywhere in the paper is that $\mathcal{A}_{\Omega}$ is {\it finite-dimensional}. In this way, $\Omega$ is endowed with a canonical topology we do not need to specify. As customary, we will consider only compact state spaces. In what follows, we will adopt the shorthand $\mathcal{A}_{\Omega}^{*}=V$ and hence $\mathcal{A}_{\Omega}=V^{*}$. We stress that while the assumptions of finite dimension and compactness can be regarded as purely technical, the convexity of the state space plays a decisive role in the physical interpretation.

We see that the vector space $V^{*}$ can be equipped with a natural ordering: if $f,g\in V^{*}$, we say that $f\leq g$ whenever $f(\omega)\leq g(\omega)$ for all $\omega\in \Omega$. Observe that the set $[0,u]$ of all effects coincides with the interval $\{f\in V^*: 0\leq f\leq u\}$, as our notation suggested. The ordering on $V^*$, in turn, induces a dual ordering on the space $V$: if $x,y\in V$, we declare $x\leq y$ whenever $f(x)\leq f(y)$ for all positive functionals $V^{*}\ni f\geq 0$. It is easy to see that the positive cone $C=V_{+}\subset V$ coincides with the set of non-negative multiples of states. Observe that $V,V^{*}$ equipped with these translationally invariant, positively homogeneous order relations become ordered vector spaces. The cones of positive elements $C,C^{*}$ are dual to each other and both closed, convex, salient and generating. Furthermore, we notice that the unit functional $u$ belongs to the interior of the cone $C^{*}$, which we signify by calling it it {\it strictly positive}. We remind the reader that given a cone $K\subset V$, its dual $K^{*}$ is defined as $K^{*}=\{f\in V^{*}:\ f(a)\geq 0\ \forall\ a\in K\}\subset V^{*}$. Moreover, $K$ is said to be salient if $K\cap (-K)=\{0\}$, and generating if $\text{span}(K)=V$. These two latter notions are elementarily seen to be dual to each other, meaning that a closed cone is salient iff the dual is generating, and vice versa. From the above discussion we abstract the following definition.

\vspace{2ex}
\begin{Def}[General probabilistic theories~\cite{Ellis dual base, Davies 1970}]
A \emph{general probabilistic theory} (GPT) is a triple $(V, C, u)$ consisting of a real, finite-dimensional vector space $V$ ordered by a closed, convex, salient and generating cone $C$ and of a strictly positive element $u\in \text{\emph{int}}(C^{*})$. We call $d=\dim V$ the dimension of the GPT.
\end{Def}

\vspace{0ex}
\begin{note}
All vector spaces in the present paper are from now on understood to be real and finite-dimensional.
\end{note}

Let us note in passing that in~\cite{Davies 1970} GPTs are called \emph{state spaces}. However, here we reserve this latter term for the set $\Omega$ of normalised states. As it turns out, in every GPT there is a natural notion of norm induced on the vector space by the unit in the dual.

\vspace{2ex}
\begin{Def}[Base norm] \label{def base}
Let $(V,C,u)$ be a GPT. The \emph{base norm} $\|\cdot\|$ on $V$ is defined as
\begin{equation}
\|x\|\, \coloneqq\, \max_{f\in [-u,u]} f(x)\, =\, \max_{e\in [0,u]} \{|e(x)| + |(u-e)(x)|\}\, ,
\label{base norm}
\end{equation}
where $[-u,u]=\{f\in V^{*}:\ -u\leq f\leq u\}$. In other words, the dual base norm has $[-u,u]$ as the unit ball.
\end{Def}

\begin{note}
Whenever there is room for ambiguity, we will denote by a subscript the GPT to which the base norm refers.
\end{note}

Observe that if $a\geq 0$ then $\|a\|=u(a)$, i.e. its base norm coincides with the expectation value of the unit effect. For general $x\in V$, we will find useful a dual formula for the base norm as given in the following lemma.

\vspace{2ex}
\begin{lemma}\emph{\cite{Krein dual base, Ellis dual base}} \label{dual base}
Let $(V,C,u)$ be a GPT. Then its base norm $\|\cdot\|$ as defined in~\eqref{base norm} is also given by
\begin{equation}
\|x\|\, =\, \min\left\{ u(x_{+}) + u(x_{-})\ :\ x_{\pm}\geq 0,\ x=x_{+}-x_{-} \right\} .
\label{dual base eq}
\end{equation}
\end{lemma}

\begin{proof}
For the original proof we refer the reader to~\cite{Krein dual base, Ellis dual base}. Alternatively, in finite dimension one can observe that~\eqref{dual base eq} is exactly the dual of the convex program in~\eqref{base norm} (first equality). Since $f=0$ is a strictly feasible point of the primal problem, Slater's condition~\cite{B&V} ensures that the two values coincide, concluding the proof.
\end{proof}

\subsection{Bipartite systems}

A physical theory modelled by a GPT should encompass a way to build a bipartite system $AB$ out of two single systems $A$ and $B$. Throughout this section, we will review the basic physical requirements to be imposed on such a construction. First, performing a local measurement $(e_{A},e_{B})$ should be allowed for all effects $e_{A}\in[0,u_{A}]$ and $e_{B}\in [0,u_{B}]$, no matter what the joint state of the system is. Secondly, preparing a separate state $(\omega_{A},\omega_{B})$ should be possible for all local states $\omega_{A}\in \Omega_{A}$ and $\omega_{B}\in \Omega_{B}$. Finally, it is reasonable to postulate the {\it local tomography principle}: bipartite states are fully determined by the statistics resulting from local measurements. Under these assumptions, it can be shown~\cite{tensor rule 1, tensor rule 2} that the joint vector space $V_{AB}$ can be identified with the tensor product of the two local spaces, i.e. 
\begin{equation}
V_{AB}\, \simeq\, V_{A}\otimes V_{B}\, .
\label{tensor spaces}
\end{equation}
Similarly, the joint unit effect can be taken as the product of the two local unit effects, $u_{AB}=u_{A}\otimes u_{B}$. Within this framework, the non-signalling principle is automatically implemented by the mathematically formalism: if $AB$ is in a state $\omega_{AB}$, whatever operation is carried out on the local party $B$, the reduced state of the system $A$ will be given by $u_B(\omega_{AB})$, implicitly defined by the equations $f\left( u_B(\omega_{AB})\right)\coloneqq (f\otimes u_B)(\omega_{AB})$ for all $f\in V_A^*$.

Remarkably, the joint cone of positive elements $C_{AB}$ is not fully determined by the above axioms. Instead, it can be shown that
\begin{equation}
C_{A}\tmin C_{B}\, \subseteq\, C_{AB}\, \subseteq\, C_{A}\tmax  C_{B}\, .
\label{CAB bound}
\end{equation}
Here, the `lower bound' $C_{A}\tmin C_{B}$, called {\it minimal tensor product}, is given by
\begin{equation}
C_{A}\tmin C_{B}\, \coloneqq \, \text{conv}\left( C_{A}\otimes C_{B} \right) ,
\label{minimal}
\end{equation}
where $C_{A}\otimes C_{B}=\left\{ a\otimes b:\ a\in C_{A},\ b\in C_{B} \right\}$, and $\text{conv}$ denotes the convex hull. States in the minimal tensor product are conventionally called {\it separable}. We should mention here that the above definition of minimal tensor product yields automatically a closed convex cone. To see that this, it suffices to show that the corresponding state space $\text{conv}(\Omega_{A}\otimes \Omega_{B})$ is closed. However, in the case considered in this work this is particularly easy, since the fact that $\Omega_{A}$ and $\Omega_{B}$ are compact sets immediately implies that the set $\text{conv}(\Omega_{A}\otimes \Omega_{B})$ is compact too~\cite[Theorem 3.25(b)]{RUDIN}. We refer to~\cite[Exercise 4.14]{AUBRUN} for a more general statement where compactness is not assumed.

The complementary `upper bound' in~\eqref{CAB bound}, called {\it maximal tensor product}, can be alternatively defined as
\begin{equation}
C_{A}\tmax  C_{B}\, \coloneqq\, \{ Z\in V_{A}\otimes V_{B}:\ (f\otimes g)(Z)\geq 0\ \forall\ f\in C_{A}^{*},\, g\in C_{B}^{*} \}\, =\, \Big( C_{A}^{*}\tmin C_{B}^{*}\Big)^{*} .
\label{maximal}
\end{equation}
Observe that also~\eqref{maximal} identifies a closed convex cone. The original definitions as well as the (elementary) proof of the above equivalence were given by Namioka and Phelps in~\cite{NP}. For a more comprehensive introduction to the theory of tensor product of cones, we refer the reader to~\cite{Mulansky}.

In the following, we will see how the two-sided bound~\eqref{CAB bound} applies in many concrete cases. With a slight abuse of notation, given two GPTs $A=(V_{A},C_{A},u_{A}),\, B=(V_{B},C_{B},u_{B})$ we will refer to the composites
\begin{equation}
A\tmin B\, \coloneqq\, \Big( V_{A}\otimes V_{B},\, C_{A}\tmin C_{B},\, u_{A}\otimes u_{B}\Big) ,\qquad A\tmax  B\, \coloneqq\, \Big( V_{A}\otimes V_{B},\, C_{A}\tmax C_{B},\, u_{A}\otimes u_{B}\Big)
\label{min max theories}
\end{equation}
as the minimal and maximal tensor product of the GPTs $A$ and $B$, respectively. 

Now, let us devote some time to examine a particularly simple instance of a GPT, namely classical probability theory with a finite alphabet. Although it is trivial in some respects, its understanding is crucial in appreciating more complicated examples such as quantum theory (Example~\ref{ex QM}) or spherical models (Example~\ref{ex sph}), not to mention even more exotic GPTs to be treated later on (Example~\ref{cubic}).

\vspace{2ex}
\begin{ex}[Classical probability theory] \label{ex class}
The state space of classical probability theory is the set of probability distributions over a finite alphabet. The corresponding GPT can be defined as the triple
\begin{equation}
\text{Cl}_{d}\coloneqq \left(\mathds{R}^{d},\, \mathds{R}^{d}_{+},\, u\right)\, ,
\label{classical}
\end{equation}
where $\mathds{R}^{d}_{+}\coloneqq\{ x\in \mathds{R}^{d}:\, x_{i}\geq 0\ \forall\, i=1,\ldots, d\}$ is the positive orthant, and the unit effect acts as $u(y)=\sum_{i=1}^{d}y_{i}$ for all $y\in \mathds{R}^{d}$. The base norm associated with classical probability theory is easily seen to be the $l_{1}$-norm $|x|_{1}=\sum_{i=1}^{d}|x_{i}|$. Since classical cones are simplicial, it can be seen that the lower and upper bound in~\eqref{CAB bound} coincide and the composition rule is trivial: if $A=\text{Cl}_{d_{A}}$ or $B=\text{Cl}_{d_{B}}$, then necessarily $C_A\tmin C_B=C_A\tmax  C_B$.
\end{ex}

\section{Data hiding in GPTs} \label{sec dh GPTs}

\subsection{State discriminiation}

The binary distinguishability problem consists of choosing secretly one of the two states $\rho,\sigma\in \Omega$ (where $\Omega$ is a generic state space) with known a priori probabilities $p, 1-p$ and handing it over to an agent, whose task is to discriminate between the two alternatives. Naturally, the larger the set of measurements the agent has at his disposal, the lower the associated probability of error will be (for fixed states and a priori probabilities). In general, it will make sense to consider measurements that are at least {\it informationally complete}, meaning that from their complete statistics the full state can be reconstructed unambiguously. A formal definition is below.

\vspace{2ex}
\begin{Def} \label{info complete}
Let $(V,C,u)$ be a GPT. Then a measurement $(e_{i})_{i\in I}\in \mathbf{M}$ (i.e. finite family of effects $(e_{i})_{i\in I}\subset[0,u]$ such that $\sum_{i\in I} e_{i}=u$) is said to be \emph{informationally complete} if $\text{\emph{span}}\{e_{i}:\, i\in I\}=V^{*}$. A set $\{\mu_t\}_{t\in T}$ made of measurements $\mu_t=\big(e_i^{(t)}\big)_{i\in I_t}$ is deemed informationally complete if $\text{\emph{span}}\big\{e_{i}^{(t)}:\, t\in T, i\in I_t\big\}=V^{*}$.
\end{Def}

\vspace{2ex}
If $\mathcal{M}\subseteq \mathbf{M}$ is a set of measurements in an arbitrary GPT, we can define an associated norm by translating to GPTs the analogous definition in~\cite{VV dh}.

\vspace{2ex}
\begin{Def} \label{def d norm}
Let $\mathcal{M}\subseteq\mathbf{M}$ be an informationally complete set of measurements in a GPT $(V,C,u)$. The associated \emph{distinguishability norm} $\|\cdot\|_{\mathcal{M}}$ is a norm on $V$ given by
\begin{equation}
\|x\|_{\mathcal{M}} \coloneqq \sup_{(e_{i})_{i\in I}\in \mathcal{M}} \sum_{i} |e_{i}(x)|
\label{d norm}
\end{equation}
for all $x\in V$.
\end{Def}

\vspace{2ex}
As is easy to see, the function defined in~\eqref{d norm} is truly a norm on $V$ thanks to the informational completeness of $\mathcal{M}$. Among its elementary properties, we note the following: (i) the identity $\|a\|_{\mathcal{M}}=u(a)$, valid on positive states $a\geq 0$ (independently of $\mathcal{M}$); (ii) the general bound $\|x\|_{\mathcal{M}} \geq |u(x)|$, an elementary consequence of the definition~\eqref{d norm}; (iii) the monotonicity of $\|\cdot\|_\mathcal{M}$ in $\mathcal{M}$, with the partial order defined by the inclusion; and (iv) the fact that $\|\cdot\|_{\mathbf{M}}$ coincides with the base norm, as seen from the second equality in~\eqref{base norm}. The above formal definition becomes relevant by virtue of its link to the operational task of state discrimination as given by the following lemma, totally analogous to~\cite[Theorem 5]{VV dh}.

\vspace{2ex}
\begin{lemma} \label{discr GPT}
Let $(V,C,u)$ be a GPT with state space $\Omega$, and consider an informationally complete set of measurements $\mathcal{M}\subseteq \mathbf{M}$. Then the lowest probability of error for discriminating between two states $\rho,\sigma\in \Omega$ with a priori probabilities $p,1-p$, respectively, is given by
\begin{equation}
P_{e}^{\mathcal{M}}(\rho,\sigma; p)\, =\, \frac12 \left( 1 - \| p\rho-(1-p)\sigma\|_{\mathcal{M}}  \right) , \label{pr error}
\end{equation}
where $\|\cdot\|_{\mathcal{M}}$ is the distinguishability norm given by~\eqref{d norm}.
\end{lemma}

\begin{proof}
This goes in complete analogy with the corresponding argument for quantum mechanics~\cite{Helstrom original, VV dh, state discr GPT}, but we explain it here for the sake of completeness. Without loss of generality, we can assume that the protocol consists of measuring the state with a measurement $(e_{i})_{i\in I}$ and performing a (possibly probabilistic) post-processing of the classical outcome $i$. Assume that the outcome $i$ yields $\rho$ or $\sigma$ as final answers with probabilities $q_{i}$ and $1-q_{i}$, respectively. Then the probability of error is given by
\begin{equation*}
P_{e}\, =\, p \sum_{i\in I} e_{i}(\rho) (1-q_{i}) + (1-p) \sum_{i} e_{i}(\sigma) q_{i}\, =\, p - \sum_{i} q_{i} e_{i}\left(p\rho-(1-p)\sigma\right) ,
\end{equation*}
where we employed the normalisation relation $\sum_{i\in I} e_{i}=u$. Minimising over all probabilities $q_{i}$ one obtains
\begin{equation*}
P_{e}\, =\, \frac12 \left( 1 - \sum_{i\in I} |e_{i}\left(p\rho-(1-p)\sigma\right)| \right) ,
\end{equation*}
and finally~\eqref{pr error} after a minimisation over all measurements $(e_{i})_{i\in I}\in \mathcal{M}$.
\end{proof}

\vspace{2ex}
The analogy with quantum mechanics goes much beyond this. In fact, all the results of~\cite[Section 2]{VV dh} (with the exception of Proposition 8 there) carry over to GPTs. In translating the statements one has just to remember that the quantum mechanical trace norm $\|\cdot\|_{1}$ becomes the base norm in the GPT framework, and that similarly the identity becomes the unit effect. For details on the interpretation of quantum mechanics as a GPT, we refer the reader to Example~\ref{ex QM}. Here we limit ourselves to provide a formulation of~\cite[Theorem 4]{VV dh} for GPTs. In what follows, for a set of measurements $\mathcal{M}\subseteq\mathbf{M}$ we will denote by $\langle \mathcal{M} \rangle$ the set generated by $\mathcal{M}$ via {\it coarse graining}, i.e. by a posteriori declaring some of the outcomes of a measurement in $\mathcal{M}$ as the same. In formula,
\begin{equation}
\langle \mathcal{M} \rangle\, \coloneqq\, \bigg\{ (e_{j})_{j\in J}:\ \exists\ I\ \text{finite},\ \{I_{j}\}_{j\in J}\ \text{partition of}\ I,\ (e'_{i})_{i\in I}\in\mathcal{M}:\ e_{j}=\sum_{i\in I_{j}} e'_{i}\ \, \forall\, j\in J \bigg\} .
\label{coarse}
\end{equation} 

\vspace{0ex}
\begin{lemma} \label{unit ball dual d norm}
The unit ball of the dual to the distinguishability norm~\eqref{d norm} is given by
\begin{equation}
B_{\|\cdot\|_{\mathcal{M},*}}\, =\, \overline{\text{\emph{conv}} \big\{ 2e-u:\ ( e, u-e )\in \langle \mathcal{M}\rangle \big\}}\, .
\label{unit ball dual d norm eq}
\end{equation}
Equivalently, $\|\cdot\|_{\mathcal{M}}$ can be computed as
\bq
\|x\|_{\mathcal{M}}\, =\, \sup\left\{ f(x):\ \left(\frac{u+f}{2},\, \frac{u-f}{2} \right) \in \langle \mathcal{M}\rangle \right\}\, .
\label{d norm altern}
\eq
Consequently, there is a one-to-one correspondence between distinguishability norms~\eqref{d norm} and closed symmetric convex bodies $K$ such that $\pm u\in K\subseteq [-u,u]$.
\end{lemma}

\begin{proof}
See~\cite{VV dh}.
\end{proof}

\vspace{1ex}
\begin{rem}
From Lemma~\ref{unit ball dual d norm} it follows in particular that $\|\cdot\|_{\mathcal{M}}$ depends only on the set of measurements generated by $\mathcal{M}$ via coarse graining, and in fact only on the right-hand side of~\eqref{unit ball dual d norm eq}.
\end{rem}

At this point, the reader should be familiar enough with the body of techniques we have discussed so far, to be able to work out the translations of the other results in~\cite[Section 2]{VV dh} by herself. As for us, instead of reporting them, we believe it more appropriate to tell the antecedent story of data hiding in quantum mechanics, and to show in detail how the definitions can be generalised, as to encompass arbitrary GPTs.

\subsection{Data hiding and statement of the problem}

Throughout this section, a generalisation of the concept of data hiding against LOCC measurements in quantum mechanics as originally conceived in~\cite{dh original 1, dh original 2} is discussed. On the one hand, we will extend this notion to an arbitrary GPT, and on the other hand we will allow for data hiding against an arbitrary set of measurements, without any a priori assumption on its nature. We give the following definition.

\vspace{2ex}
\begin{Def} \label{dh}
Let $(V,C,u)$ be a GPT. For an informationally complete set of measurements $\mathcal{M} \subseteq \mathbf{M}$, we say that there is \emph{data hiding against $\mathcal{M}$ with efficiency $R\geq 1$} if there are two normalised states $\rho,\sigma\in\Omega$ and a real number $p\in [0,1]$ such that the probability of error defined in~\eqref{pr error} satisfies
\begin{equation}
P_{e}^{\mathbf{M}}(\rho,\sigma; p) = 0\, ,\qquad P_{e}^{\mathcal{M}}(\rho,\sigma; p) = \frac12 \left( 1-\frac1R\right) .
\label{dh eq}
\end{equation}
The highest data hiding efficiency against $\mathcal{M}$ is called \emph{data hiding ratio against $\mathcal{M}$} and will be denoted by $R(\mathcal{M})$.
\end{Def}

\vspace{1ex}
\begin{rem}
A priori, the only meaningful way to define $R(\mathcal{M})$ is as the supremum of all achievable data hiding efficiencies. However, we will see in a moment that this supremum is actually a maximum (Proposition~\ref{dh ratio}).
\end{rem}

In the above Definition~\ref{dh}, we chose not to restrict ourselves to equiprobable pairs of states. There are several reasons that justify this choice. On the one hand, it can be shown that any pair that exhibits data hiding with high efficiency is approximately equiprobable, the approximation becoming better and better for higher efficiencies. On the other hand, even considering only the case of exact equality $p=1/2$ from the start, the obtained data hiding ratio does not differ by more than a factor of two (additive constants apart) from that we have defined here. We devote Appendix~\ref{app equi} to exploring the consequences of restricting the definition of data hiding to the equiprobable case.

We now go back to the investigation of data hiding in the sense of Definition~\ref{dh}.
An elementary yet fruitful observation is that if a set of measurements exhibits data hiding with high efficiency, then the associated distinguishability norm -- given by~\eqref{d norm} -- has to be very different from the base norm. This is again a straightforward consequence of Lemma~\ref{discr GPT}. The following result establishes the converse, i.e. that if the two norms are very different on some vectors then there is highly efficient data hiding.

\vspace{2ex}
\begin{prop} \label{dh ratio}
For an informationally complete set of measurements $\mathcal{M}\subseteq\mathbf{M}$ in an arbitrary GPT, the data hiding ratio $R(\mathcal{M})$ is given by
\begin{equation}
R(\mathcal{M})\, =\, \max_{0\neq x\in V} \frac{\|x\|}{\|x\|_{\mathcal{M}}}\, .
\label{dh ratio eq}
\end{equation}
Adopting the terminology of~\cite{VV dh}, we can rephrase~\eqref{dh ratio eq} by saying that $R(\mathcal{M})$ is the constant of domination of $\|\cdot\|_{\mathcal{M}}$ on $\|\cdot\|$, i.e. the smallest $k\in \mathds{R}$ such that $\|\cdot\|\leq k\|\cdot\|_{\mathcal{M}}$.
\end{prop}

\begin{proof}
Let the GPT $(V,C,u)$ have state space $\Omega$. Then, from Definition~\ref{dh} and from~\eqref{pr error} we see that
\bqq
R(\mathcal{M})\, =\, \sup \left\{ \|p\rho-(1-p)\sigma\|_{\mathcal{M}}^{-1}:\ \rho,\sigma\in\Omega,\ 0\leq p\leq 1,\ \|p\rho-(1-p)\sigma\|=1 \right\} .
\eqq
The crucial observation is that the set $K$ of vectors $x\in V$ that can be represented as $x=p\rho-(1-p)\sigma$ for appropriate $\rho,\sigma\in\Omega$ and $p\in[0,1]$ coincides with the unit ball of the base norm, i.e. $K=B_{\|\cdot\|}$. In order to prove this, start by observing that $\|p\rho-(1-p)\sigma\|\leq p\|\rho\|+(1-p)\|\sigma\|=p+(1-p)=1$, so $\|x\|\leq 1$ is a necessary condition for $x$ to belong to $K$. To see that it is also sufficient, notice that $K$ is convex by construction and $0\in K$, and therefore it suffices to consider the case $\|x\|=1$. With this hypothesis, Lemma~\ref{dual base} yields a decomposition $x=x_{+}-x_{-}$ such that $x_{\pm}\geq 0$ and $u(x_{+})+u(x_{-})=1$, from which it follows that $x_{+}=p\rho$ and $x_{-}=(1-p)\sigma$ for $p=u(x_{+})\in [0,1]$ and $\rho,\sigma\in \Omega$ normalised states. Thanks to this observation, we rewrite the above representation of $R(\mathcal{M})$ as
\bqq
R(\mathcal{M})\, =\, \sup\left\{\|x\|_{\mathcal{M}}^{-1}:\, \|x\|=1\right\}\, =\, \max_{x\neq 0}\, \frac{\|x\|}{\|x\|_{\mathcal{M}}}\, .
\eqq
For the last step, we used the positive homogeneity of the norms, and we converted the supremum over the (compact) unit ball of the base norm into a maximum.
\end{proof}

\vspace{2ex}
Some elementary properties of the data hiding ratio are as follows.

\vspace{2ex}
\begin{lemma} \label{elem dh}
Let $\mathcal{M}$ be an informationally complete set of measurements in an arbitrary GPT. Then:
\begin{itemize}
\item the data hiding ratio given in Definition~\ref{dh} satisfies $R(\mathcal{M})\geq 1$, with equality iff the set on the left hand side of~\eqref{unit ball dual d norm eq} coincides with the full interval $[-u,u]$;
\item $R(\mathcal{M})$ is monotonically non-increasing as a function of $\mathcal{M}$, where the partial order on sets of measurements is the one given by inclusion.
\end{itemize}
\end{lemma}

\vspace{2ex}
As expected, not much can be said about data hiding ratios for a single system and when the sets of measurements are completely arbitrary. In fact, it is easy to show that already in a classical GPT (Example~\ref{ex class}) such as $(\mathds{R}^2, \mathds{R}^2_+, u)$ with $u=(1,1)$ such that $u(x,y)\coloneqq x+y$, for the particular case when $\mathcal{M}$ is made of just one measurement $\big((\varepsilon,0),\, (1-\varepsilon, 1)\big)$ we have $R(\mathcal{M})=(2-\varepsilon)/\varepsilon$, so that the data hiding ratio can even be unbounded in a system of fixed dimension.

The situation changes dramatically when we consider bipartite systems, as originally done in the context of data hiding. Let us discuss first bipartite quantum systems. Here, some restricted sets of measurements come into play quite naturally as deriving from operational constraints. Examples of such sets include local operations (LO), local operations assisted by shared randomness (LOSR) or one-way classical communication ($\text{LOCC}_{\rightarrow}$), and general LOCC protocols when both communication directions are allowed. It is also convenient to introduce mathematical relaxations of these classes. For instance, consider the set of measurements $(e_{i})_{i\in I}$ such that $E_{i}$ is a separable positive operators for all $i\in I$ (i.e. it belongs to the quantum mechanical equivalent of~\eqref{minimal} as given by~\eqref{separable}). We call these measurements {\it separable}, and denote them collectively as SEP. It is easy to see that $\text{LOCC}\subseteq\text{SEP}$, and less trivially the inclusion can be shown to be strict.

Until now we have discussed only quantum theory. Perhaps surprisingly, it turns out that all these restricted classes of measurements can be defined in any bipartite arbitrary GPT $AB=(V_A\otimes V_B,\, C_{AB},\, u_A\otimes u_B)$ constructed out of two local theories $A$ and $B$ in such a way that the constraints~\eqref{CAB bound} are met. Before we provide general definitions below, let us briefly discuss how to add {\it dynamical prescriptions} to our structure. The purpose of these rules is to specify how states transform after a measurement, extending the picture we have been describing so far, mostly orientated towards the outcomes and their probabilities. In practice, we will not make use of these prescriptions, and in fact our results are totally independent of any assumption concerning them beyond the mere consistency with the operational interpretation of the theory. However, this apparatus is needed to define a generic LOCC protocol, which requires multiple, interactive rounds of operations on the same systems.

Following~\cite{Davies 1970}, we can define \emph{instruments} on one of the two system, say $A=(V_A, C_A, u_A)$, as collections $(\phi_i)_{i\in I}$ of linear maps $\phi_i: V_A\rightarrow V_A$ that are \emph{completely positive}, i.e. satisfy $\left((\phi_i)_A \otimes I_B\right)(C_{AB})\subseteq C_{AB}$, and sum up to a normalisation-preserving map, i.e. $\sum_{i\in I} \phi_{i}^*(u_A) = u_A$, with $\phi_i^*: V_A^*\rightarrow V_A^*$ being the dual (or transpose) of $\phi_i$. A totally analogous definition can be given for instruments on the $B$ system. In the operational interpretation of the theory, an instrument describes a non-destructive measurement, with $\phi_i(\rho)$ representing the unnormalised post-measurement state when the outcome $i$ has been recorded on the initial state $\rho$, and the normalisation coefficient $u\left(\phi_{i}(\rho) \right) = \left(\phi_{i}^*(u)\right)(\rho)$ being the probability that the process yields the outcome $i$ (accordingly, observe that $\left(\phi_i^*(u)\right)_{i\in I}$ is a valid measurement in the GPT sense). 

With the concept of instrument at hand, in order to define LOCC protocols we can follow the steps described in~\cite[Section 2.2]{LOCC}. We will not repeat the construction here since it is totally analogous to the quantum mechanical one, once the concept of instrument in GPTs has been clarified.

\vspace{2ex}
\begin{Def} \label{locally constr}
Let $A=(V_{A}, C_{A}, u_{A})$ and $B=(V_{B}, C_{B}, u_{B})$ be two GPTs, and let the composite system \mbox{$AB=(V_A\otimes V_B,\, C_{AB},\, u_A\otimes u_B)$} satisfy~\eqref{CAB bound}. Then local operations (LO), local operations assisted by one-way classical communication ($\text{LOCC}_{\rightarrow}$) or two-way classical communication ($\text{LOCC}$), and separable measurements ($\text{SEP}$) are subsets of the set $\mathbf{M}_{AB}$ of all measurements on $AB$ given by:
\begin{align}
\text{\emph{LO}}\, &\coloneqq\, \left\langle \left\{ (e_{i}\otimes f_{j})_{(i,j)\in I\times J}:\ (e_{i})_{i\in I}\in\mathbf{M}_{A},\ (f_{j})_{j\in J}\in\mathbf{M}_{B} \right\}\right\rangle , \label{LO} \\[1ex]
\text{\emph{LOCC}}_{\rightarrow}\, &\coloneqq\, \left\langle\left\{ (e_{i}\otimes f_{j}^{(i)})_{(i,j)\in I\times J}:\ (e_{i})_{i\in I}\in\mathbf{M}_{A},\ (f_{j}^{(i)})_{j\in J}\in\mathbf{M}_{B}\ \forall\ i\in I \right\}\right\rangle , \label{1-way LOCC} \\[1ex]
\text{\emph{LOCC}}\, &\coloneqq\, \left\{ \left( \Phi_i^*(u_A\otimes u_B) \right)_{i\in I}:\ \text{$\left( \Phi_i \right)_{i\in I}$ LOCC instrument on $AB$} \right\} \\[1ex]
\text{\emph{SEP}}\, &\coloneqq\, \left\{ (E_i)_{i\in I} \in \mathbf{M}_{AB}:\ E_i\in C_A^* \tminit C_B^*\ \forall\ i \right\} . \label{SEP}
\end{align}
Here, $\langle\cdot\rangle$ denotes coarse graining as defined by~\eqref{coarse}. The above sets will be collectively called \emph{locally constrained sets of measurements}.
\end{Def}

\vspace{2ex}
It is easy to verify that
\begin{equation}
\text{LO}\, \subseteq\, \text{LOCC}_{\rightarrow}\, \subseteq\, \text{LOCC}\, \subseteq\, \text{SEP}\, .
\label{chain M}
\end{equation}
The last inclusion is slightly less trivial than the others, but its proof follows closely the quantum mechanical one. Namely, referring for details and nomenclature to~\cite{LOCC}, one can observe that: (i) one-way local instruments are separable, in the sense that each component is a positive sum of tensor products of completely positive maps; (ii) coarse-graining preserves separability; (iii) an instrument that is LOCC-linked to a separable one is again separable; (iv) separability is preserved under limits; and finally (v) if $(\Phi_i)_{i\in I}$ (acting on $AB$) is separable as an instrument, $\left( \Phi_i(u_A\otimes u_B) \right)_{i\in I}$ is separable as a measurement.

\vspace{1ex}
\begin{rem}
From an operational point of view, there is at least another notable set of locally constrained measurements that is worth mentioning. This is the set of local operations assisted by shared randomness, formally defined as
\begin{equation}
\text{LOSR}\, \coloneqq\, \left\langle\left\{ \big(p_{i}\, e^{(i)}_{j}\otimes f^{(i)}_{k}\big)_{(i,j,k)\in I\times J\times K}:\ I\ \text{finite},\ p\in\mathcal{P}(I),\ \big(e^{(i)}_{j}\big)_{j\in J}\in\mathbf{M}_{A},\ \big(f_{k}^{(i)}\big)_{k\in K}\in\mathbf{M}_{B}\ \forall\ i\in I \right\}\right\rangle , \label{LOSR}
\end{equation}
where $\mathcal{P}(I)$ stands for the the set of probability distributions on a finite alphabet $I$. It is very easy to see that $\text{LOSR}$ lies between $\text{LO}$ and $\text{LOCC}_\rightarrow$, i.e.
\begin{equation}
\text{LO}\, \subseteq\, \text{LOSR}\, \subseteq\, \text{LOCC}_{\rightarrow}\, .
\end{equation}
The reason why we did not include this additional set in the above definition is that elements of $\text{LOSR}$ are convex combinations of local measurements, hence it is easy to check that
\begin{equation}
\|\cdot\|_{\text{LOSR}}\, =\, \|\cdot\|_{\text{LO}}\, .
\label{LOSR = LO norm}
\end{equation}
In other words, once they are restricted to local measurements, providing the parties with additional shared randomness does not enhance their distinguishability power.
\end{rem}

\vspace{1ex}
\begin{rem}
Let us stress here that $\text{LOCC}$ is the only locally constrained set of measurements that depends explicitly on the choice of the positive cone $C_{AB}$ of the bipartite system. In fact, it is easy to realise that only the structure of the local GPTs matters in~\eqref{LO},~\eqref{1-way LOCC}, and~\eqref{SEP} (as well as in~\eqref{LOSR}, incidentally). Following the discussion before Definition~\ref{locally constr}, we see that the dependence of $\text{LOCC}$ on $C_{AB}$ is hidden inside the concept of completely positive map, which is in turn necessary to define local instruments.
\end{rem}

An important feature of locally constrained sets of measurements is informational completeness. This is a consequence of the {\it local tomography principle}, valid in an arbitrary bipartite GPT and stating that the statistics of local measurements contain enough information to determine the global state completely. In order to prove this elementary fact, consider a local measurement $(e_{i}\otimes f_{j})_{(i,j)\in I\times J}\in \text{LO}$ such that $\text{span}\{e_{i}\}_{i\in I}=V_{A}^{*}$ and $\text{span}\{f_{j}\}_{j\in J}=V_{B}^{*}$. Then, obviously, $\text{span}\{e_{i}\otimes f_{j}\}_{(i,j)\in I\times J}=V_{A}^{*}\otimes V_{B}^{*}$, which yields the claim.

As we said, our primary interest lies in understanding data hiding against those restricted sets of measurements, whose corresponding constraints have an operational nature. In this context, the above locally constrained sets of measurements are thus excellent candidates, and in fact the rest of the present paper is devoted to the study of the data hiding ratios $R(\mathcal{M})$, with $\mathcal{M}=\text{LO},\, \text{LOCC}_\rightarrow,\, \text{LOCC},\, \text{SEP}$.
As an preliminary observation, note that~\eqref{chain M} and Lemma~\ref{elem dh} imply that
\bq
R(\text{SEP})\, \leq\, R(\text{LOCC})\, \leq\, R(\text{LOCC}_{\rightarrow})\, \leq\, R(\text{LO})
\label{chain dh}
\eq
for all fixed GPTs.
Among the many questions one could ask at this point, one seems particularly relevant to us. Namely, we can wonder, how the best data hiding ratio {\it scales} with the dimensions of the local GPTs. To be more precise, we give the following definition.

\vspace{2ex}
\begin{Def}[Ultimate data hiding ratio] \label{univ dh ratio}
For a locally constrained set of measurements $\mathcal{M}$, the \emph{ultimate data hiding ratio against $\mathcal{M}$} for fixed local dimensions, denoted by $R_{\mathcal{M}}(d_{A},d_{B})$, is the supremum over all data hiding ratios $R(\mathcal{M})$ achieved by composite GPTs that satisfy~\eqref{CAB bound} and have local dimensions $d_{A},d_{B}$.
\end{Def}

\vspace{2ex}
With this concept in hand, we are ready to formulate the question lying at the heart of our investigation, namely, {\it what is the scaling of the ultimate ratio $R_\mathcal{M}(d_{A}, d_{B})$ with the local dimensions $d_{A},d_{B}$?} Clearly, thanks to the chain of inequalities~\eqref{chain dh}, we find
\bq
R_{\text{SEP}}(d_{A},d_{B})\, \leq\, R_{\text{LOCC}}(d_A,d_B)\, \leq\, R_{\text{LOCC}_{\rightarrow}}(d_{A},d_{B})\, \leq\, R_{\text{LO}}(d_{A},d_{B})\, .
\label{chain univ dh}
\eq

We stress that the supremum in Definition~\ref{univ dh ratio} has to be taken over {\it all} local GPTs of the given dimensions, and over all {\it composition rules} to join the system (i.e. among all the global cones respecting the bounds~\eqref{CAB bound}). Now, we will show that at least this latter maximisation can be carried out explicitly when $\mathcal{M}$ is a locally constrained set of measurements different from $\text{LOCC}$, the optimal composite being always given by the {\it minimal} tensor product. To see why, notice that the exclusion of $\text{LOCC}$ implies that for a fixed $X\in V_A\otimes V_B$ only the global base norm $\|X\|$ depends on the composition rule we chose. On the contrary, the locally constrained norm $\|\cdot\|_{\mathcal{M}}$ ($\mathcal{M}\neq \text{LOCC}$) will depend on the local structure only, as already observed. Then, maximising the ratio between the former and the latter amounts to maximising the global base norm. In order to do so, a large set of global effects and thus a small set of states are required, and according to~\eqref{CAB bound} the smallest possible positive cone in a bipartite system is given by the minimal tensor product.

The above reasoning is perhaps not obvious from Definition~\ref{dh} alone, because restricting the set of available bipartite states gives less freedom in choosing the data hiding pair. However, this restriction plays no role once Proposition~\ref{dh ratio} is available. This way around the problem is made possible by the fact that any difference of two normalised states can be thought of as a positive multiple of the difference of two separable states, the multiplication coefficient being given by the base norm induced by the minimal tensor product. We summarise this whole discussion stating the following result, whose proof reproduces the informal argument presented above.

\vspace{2ex}
\begin{prop} \label{min is optimal}
Given two local GPTs $A=(V_A,C_A,u_A)$, $B=(V_B,C_B,u_B)$, and a locally constrained set of measurements $\mathcal{M}\neq \text{\emph{LOCC}}$, the maximal data hiding ratio against $\mathcal{M}$ is achieved when the bipartite GPT is constructed according to the minimal tensor product, i.e. $AB=A\tminit B$.
\end{prop}

\begin{proof}
Let us denote by $\|\cdot\|_{AB}$ the global base norm, whose dependence from the choice of the GPT $AB$ (i.e. of the bipartite cone $C_{AB}$) has been made explicit. It is always implicitly assumed that $C_{AB}$ obeys the two-sided bound in~\eqref{CAB bound}.
As is easy to verify, for all fixed $x\in V_{A}\otimes V_B$ the norm $\|x\|_{AB}$ is monotonically non-increasing as a function of $C_{AB}$, in the sense that
\begin{equation*}
C_{AB}\subseteq \widetilde{C}_{AB}\quad\Longrightarrow\quad \|x\|_{AB}\geq \|x\|_{\widetilde{AB}}\, .
\end{equation*}
To see why, it suffices to go back to the definition of base norm and to observe that: (i) since taking the dual reverses the inclusions, if $C_{AB}\subseteq \widetilde{C}_{AB}$ then $\widetilde{C}_{AB}^*\subseteq C_{AB}^*$; (ii) from this it follows that $[0,u_{AB}]_{\widetilde{C}_{AB}^*}\subseteq [0,u_{AB}]_{C^*_{AB}}$, where $[a,b]_{K} \coloneqq \{c:\ b-c, c-a\in K \}$ denotes the interval according to the ordering determined by a cone $K$; and finally (iii) the base norm can be written as a maximisation over the interval $[0,u]_{C^*}$, as detailed in~\eqref{base norm}, hence one finds $\|x\|_{\widetilde{AB}}\leq \|x\|_{AB}$, as claimed.

An immediate corollary of this inequality is that
\bq
\|x\|_{AB} \leq\, \|x\|_{A \tminfoot B}
\label{Cmin largest norm}
\eq
for all $x\in V_A \otimes V_B$ and for all admissible composites $AB$ whose corresponding cones $C_{AB}$ satisfy~\eqref{CAB bound}.

Now that dependence of the global base norm on the choice of the positive cone has been addressed, we can turn our attention to the other object appearing in the formula~\eqref{dh ratio eq} for computing data hiding ratios, i.e. the distinguishability norm $\|\cdot\|_{\mathcal{M}}$. We already observed that for the locally constrained sets of measurements $\mathcal{M}= \text{LO},\,\text{LOCC}_\rightarrow,\, \text{SEP}$, the distinguishability norm $\|\cdot\|_{\mathcal{M}}$ does not actually depend on $C_{AB}$. Therefore, 
\begin{equation*}
R_{AB}(\mathcal{M})\, =\, \max_{0\neq x\in V_A\otimes V_B} \|x\|_{AB}\big/ \|x\|_{\mathcal{M}}\, \leq\, \max_{0\neq x\in V_A\otimes V_B} \|x\|_{A \tminfoot B} \Big/ \|x\|_{\mathcal{M}}\, =\, R_{A \tminfoot B}(\mathcal{M})\, ,
\end{equation*}
where the dependence of the data hiding ratio on the GPT has been made explicit. The data hiding ratio $R(\mathcal{M})$ is thus maximised by the minimal tensor product $A\tmin B$, as claimed.
\end{proof}

\vspace{1ex}
\begin{rem}
Some intuitive understanding of Proposition~\ref{min is optimal} can be gained by looking at the opposite case, i.e. when the composite system is formed via the maximal tensor product. When $AB=A \tmax  B$, in fact, the global base norm coincides with the separability norm. This is ensured by the fact that every allowed effect within this theory is automatically separable.
\end{rem}

\section{Examples: quantum mechanics and spherical model} \label{sec ex}

Throughout this section, we will investigate from the point of view of data hiding two well-known examples of GPTs, namely quantum mechanics and the so-called spherical model. Besides seeing all the mathematical machinery of GPTs in action in some concrete case, the purpose of doing so is twofold. First of all, the GPTs we chose to look into, especially quantum mechanics, are interesting on their own and deserve a complete solution. Second of all, computing the data hiding ratio against a locally constrained set of measurements for a specific case still yields a general lower bound on the maximal value achievable for fixed local dimensions.

\vspace{2ex}
\begin{ex}\label{ex QM}
Let us start by describing $n$-level quantum mechanics as a GPT. As is well-known, the cone of states is composed of the positive semidefinite $n\times n$ matrices (collectively denoted by $\text{PSD}_{n}$), embedded in the real space of hermitian matrices (called $\mathcal{H}_{n}$) whose real dimension is $d=n^2$. Since the density matrices are the positive matrices with trace one, the unit effect is easily seen to coincide with the trace. Therefore, we will write symbolically
\begin{equation}
\text{QM}_{n} \coloneqq \left( \mathcal{H}_{n},\, \text{PSD}_{n},\, \Tr \right) ,
\label{quantum}
\end{equation}
remembering that
\begin{equation}
\dim \text{QM}_{n}\, =\, n^{2}\, .
\label{dim quantum}
\end{equation}
Observe that the positive semidefinite cone is self-dual, i.e. $\text{PSD}_{n}^{*}=\text{PSD}_{n}$. The base norm in quantum mechanics can be proved to coincide with the trace norm $\|X\|_{1} = \Tr |X| =\sum_{i=1}^{n}|\lambda_{i}(X)|$, where $\lambda_{i}(X)$ are the eigenvalues of $X\in\mathcal{H}_{n}$ and the last equality holds because $X$ is hermitian.

Now, let us discuss the composition rules for bipartite systems. Using the definitions~\eqref{minimal} and~\eqref{maximal}, we see that
\begin{align}
\text{PSD}_{n_{A}}\tmin \text{PSD}_{n_{B}}\, &=\, \Big\{ \sum\nolimits_{i\in I} P_{i}\otimes Q_{i} :\ I\ \text{finite},\ P_{i}\in\text{PSD}_{n_{A}},\, Q_{i}\in \text{PSD}_{n_{B}}\ \forall\ i\in I \Big\}\, , \label{separable} \\[1ex]
\text{PSD}_{n_{A}}\tmax  \text{PSD}_{n_{B}}\, &=\, \left\{ W\in \mathcal{H}_{n_{A}n_{B}}:\ \Tr[ P\otimes Q\, W ]\geq 0\quad \forall\ P\in\text{PSD}_{n_{A}},\, Q\in \text{PSD}_{n_{B}} \right\} . \label{witnesses} \end{align}
In quantum information, elements of~\eqref{witnesses} are variously called entanglement witnesses, separability witnesses or block-positive operators. This latter name comes from the fact that in~\eqref{witnesses} we can restrict $P$ and $Q$ to be pure states (i.e. rank-one projectors), and the defining condition for belonging to the set amounts to impose the positivity of the diagonal block(s) in all product bases. Interestingly enough, Nature has a preferred choice for the cone of bipartite states, which is neither the maximal nor the minimal tensor product. Instead, if $A=\text{QM}_{n_{A}}$ and $B=\text{QM}_{n_{B}}$ then $AB=\text{QM}_{n_{A}n_{B}}$, i.e.
\begin{equation}
\text{PSD}_{n_{A}}\tmin \text{PSD}_{n_{B}}\, \subsetneq\, C_{AB}\, =\, \text{PSD}_{n_{A}n_{B}}\, \subsetneq\, \text{PSD}_{n_{A}}\tmax  \text{PSD}_{n_{B}}\, .
\label{cone bipartite quantum}
\end{equation} 
While this is clearly not optimal in the sense of data hiding (because of Proposition~\ref{min is optimal}), it deserves special attention because of its prime importance in physics. Toward the end of this example, we will also look into a modified version of quantum mechanics designed to encompass the minimal tensor product rule suggested by Proposition~\ref{min is optimal}.

The computation of data hiding ratios in quantum mechanics has been the subject of many papers, whose main results we summarise briefly. The original example of a data hiding pair involves the normalised projectors onto the symmetric and antisymmetric subspace in $\mathds{C}^{n}\otimes\mathds{C}^{n}$, denoted by $\rho_{S}$ and $\rho_{A}$, respectively~\cite{dh original 1, dh original 2}. While $\|\rho_{S}-\rho_{A}\|_{1}=2$ because the two states have orthogonal support, it can be shown that $\|\rho_{S}-\rho_{A}\|_{\text{LOCC}}=2/(n+1)$~\cite{dh original 1, dh original 2}. The fact that the two states are mixed is crucial for this construction to work, as it can be shown that for pure states trace norm and LOCC norm always coincide~\cite{no dh pure 1, no dh pure 2}. In general, from~\cite{VV dh Chernoff, VV dh} it is known that
\bq
\frac{\min\{n_A, n_B\}+1}{2}\, \leq\, R_{QM}(\text{SEP})\, \leq\, R_{QM}(\text{LOCC})\, \leq\, R_{QM}(\text{LO})\, \leq\, \sqrt{153\, n_A n_B}\, , \label{VV bound 1}
\eq
while for $R_{QM}(\text{SEP})$ the tighter bound
\bq
\frac{\min\{n_A, n_B\}+1}{2}\, \leq\, R_{QM}(\text{SEP})\, \leq\, \sqrt{n_A n_B} \label{VV bound 2}
\eq
is available. As shown in~\cite[Corollary 17]{VV dh}, the above relations solve the problem of determining the optimal scaling in $n$ of all the data hiding ratios against locally constrained measurements when the subsystems have equal dimensions $n_A=n_B=n$. Instead, a problem arises when $n_A$ and $n_B$ are very different and thus $\min\{n_A,n_B\} \ll \sqrt{n_A n_B}$. In this case, the leftmost and rightmost side of~\eqref{VV bound 1} are no longer of the same order of magnitude, and an alternative argument has to be designed.

This scaling problem is somehow mitigated by~\cite[Lemma 20]{Brandao area law}, which implies that
\bq
R_{QM}(\text{LOCC})\, \leq\, R_{QM} (\text{LO})\, \leq\, \min\{n_A^2,\, n_B^2\}\, .
\eq
Although this upper bound behaves better than that in~\eqref{VV bound 1} when $n_A$ and $n_B$ are very different from each other, its quadratic nature prevents us from determining -- for instance -- the exact scaling of the operationally relevant data hiding ratio against LOCC protocols.

Here we provide a simple reasoning that shows that in fact $O(\min\{n_A,n_B\})$ is still an upper bound for $R_{QM}(\text{LOCC})$ (and hence for $R_{QM}(\text{SEP})$, too). Furthermore, our reasoning yields much better constants for both the leftmost and the rightmost side of~\eqref{VV bound 1} (where $R_{QM}(\text{LO})$ is excluded, though). In fact, these constants are so close to being optimal that we are even able to compute $R_{QM}(\text{SEP})$ {\it exactly} when $n_A=n_B=n$.

\vspace{2ex}
\begin{thm}[Teleportation argument] \label{dh QM}
For a bipartite quantum mechanical system with Hilbert space $\mathds{C}^{n_A}\otimes \mathds{C}^{n_B}$, define $n\coloneqq\min\{n_A,n_B\}$. Then the data hiding ratios against separable and  $\text{LOCC}_{(\rightarrow)}$ protocols satisfy
\bq
n\, \leq\, R_{QM}(\text{\emph{SEP}})\, \leq\, R_{QM}(\text{\emph{LOCC}})\, \leq\, R_{QM}(\text{\emph{LOCC}}_{\rightarrow})\, \leq\, 2n-1\, ,
\eq
where the communication direction in $\text{LOCC}_{\rightarrow}$ is from the smaller to the larger subsystem. Moreover, if $n_A=n_B=n$ then $R_{\text{\emph{QM}}}(\text{\emph{SEP}})=n$.
\end{thm}

\begin{proof} Let us assume without loss of generality that $n=n_{A}\leq n_{B}$, and that classical communication goes from $A$ to $B$. We start by reminding the reader that the maximally entangled state $\ket{\Phi}=\frac{1}{\sqrt{n}}\,\sum_{i=1}^{n} \ket{ii}\in\mathds{C}^{n}\otimes \mathds{C}^{n}$, whose corresponding rank-one projector we denote by $\Phi$, has the property that there is a separable state $\sigma$ such that $\frac1n \Phi+\frac{n-1}{n}\, \sigma$ is again separable (in the language of~\cite{VidalTarrach}, $\Phi$ has entanglement robustness $r(\Phi)=n-1$). For instance, it is not difficult to see that $\sigma = \frac{\mathds{1} - \Phi}{n^{2}-1}$ satisfies all the requirements. This follows, for instance, from the characterisation of the separability region for isotropic states given in~\cite{Horodecki97}.

Now, since we can always produce any separable state with $\text{LOCC}_{\rightarrow}$ operations, we are free to evaluate the $\text{LOCC}_{\rightarrow}$ norm on $X_{AB} \otimes \left( \frac1n \Phi+\frac{n-1}{n}\, \sigma \right)_{A'B'}$ instead of $X_{AB}$. Here, the systems $A',B'$ have dimension $n_{A'}=n_{B'}=n_A=n$, and the operations are $\text{LOCC}_{\rightarrow}$ with respect to the splitting $AA'|BB'$.

Now, we are ready to apply the quantum teleportation protocol from $A$ to $B$~\cite{teleportation}. This is an $\text{LOCC}_{\rightarrow}$ operation $\tau$ mapping states of the system $AA'BB'$ to states of $B'B$, which can be defined as follows.
For $p,q=0,\ldots,n-1$, introduce the unitary matrices
\bq
\mathbf{X}(p)\,\coloneqq\, \sum_{k=1}^n \ket{k \oplus p}\!\!\bra{k}\, ,\qquad \mathbf{Z}(q)\, \coloneqq\, \sum_{k=1}^n e^{2qk\pi i/n} \ket{k}\!\!\bra{k}\, ,\qquad U(p,q)\coloneqq \mathbf{X}(p) \mathbf{Z}(q)\, ,
\label{HW}
\eq
where $\oplus$ denotes sum modulo $n$.
Then the teleportation $\tau$ is given by
\bq
\tau(X_{AA'BB'})\, \coloneqq\, \sum_{p,q=0}^{n-1} U(p,q)_{B'}\ \text{Tr}_{AA'}\,\big[ X_{AA'BB'}\, U(p,q)_A \Phi_{AA'} U(p,q)_A^\dag \big]\ U(p,q)_{B'}^{\dag}\, .
\label{telep}
\eq
Most notably, observe that $\tau\left(X_{AB} \otimes \Phi_{A'B'}\right)=X_{B'B}$ (meaning that the same operator $X$ is written in the registers $B'\simeq A$ and $B$). Now, on the one hand, after the protocol has been performed, the local constraint plays no role any more, and any measurement can be applied to $B'B$, showing that $\|X_{B'B}\|_{\text{LOCC}_{\rightarrow}} = \|X_{AB}\|_{1}$. On the other hand, $\tau(X_{AB}\otimes \sigma_{A'B'})$ is obtained from $X_{AB}$ via an $\text{LOCC}_{\rightarrow}$ protocol, hence $\|\tau (X_{AB} \otimes \sigma_{A'B'})\|_{\text{LOCC}_{\rightarrow}}\leq\|X_{AB}\|_{\text{LOCC}_{\rightarrow}}$.
Putting all together, we obtain the following chain of inequalities:
\begin{align}
\|X_{AB}\|_{\text{LOCC}_{\rightarrow}}\, &=\, \left\| X_{AB} \otimes \left( \frac1n \Phi+\frac{n-1}{n}\, \sigma \right)_{A'B'} \right\|_{\text{LOCC}_{\rightarrow}}
\label{telep 1} \\[0.5ex]
&\geq\, \left\| \tau\left(X_{AB} \otimes \left( \frac1n \Phi+\frac{n-1}{n}\, \sigma \right)_{A'B'}\right) \right\|_{\text{LOCC}_{\rightarrow}} \label{telep 2} \\[0.5ex]
&=\, \left\| \frac1n\, X_{B'B} +\frac{n-1}{n}\, \tau(X_{AB}\otimes \sigma_{A'B'}) \right\|_{\text{LOCC}_{\rightarrow}} \label{telep 3} \\
&\geq\, \frac1n\, \|X_{B'B}\|_{\text{LOCC}_{\rightarrow}} - \frac{n-1}{n}\, \|\tau(X_{AB}\otimes \sigma_{A'B'}) \|_{\text{LOCC}_{\rightarrow}} \label{telep 4} \\[0.5ex]
&\geq\, \frac1n\, \|X_{AB}\|_{1} - \frac{n-1}{n}\, \|X_{AB} \|_{\text{LOCC}_{\rightarrow}}\, . \label{telep 5}
\end{align}
We conclude that
\bq
\|X_{AB}\|_{\text{LOCC}_{\rightarrow}}\, \geq\,  \frac{1}{2n-1}\, \left\| X_{AB} \right\|_1\, , \label{bound telep}
\eq
enforcing $R_{QM}(\text{LOCC}_{\rightarrow})\leq 2n-1$ in view of Proposition~\ref{dh ratio}.

In order to deduce the lower bound $R_{QM}(\text{SEP})$, we appeal to Werner states~\cite{Werner, Werner symmetry}. These can be thought of as convex combinations of the normalised projectors onto the symmetric and antisymmetric subspace of $\mathds{C}^{n}\otimes \mathds{C}^{n}$, denoted by $\rho_{S}$ and $\rho_{A}$, respectively. In terms of the `flip operator' $F$ defined by $F\ket{\alpha\beta} = \ket{\beta\alpha}$ for all $\ket{\alpha},\ket{\beta}\in\mathds{C}^{n}$, we have
\begin{equation}
\rho_{S}\, =\, \frac{\mathds{1}+F}{n(n+1)}\, ,\qquad \rho_{A}\, =\, \frac{\mathds{1}-F}{n(n-1)}\, .
\label{symm antisymm proj}
\end{equation}
Since $n=\min\{n_{A},n_{B}\}$, we can safely imagine to give one share of this bipartite system to $A$ and the other to $B$. We already saw how the preparation with equal a priori probabilities of the two extremal states is well-known to produce data hiding, as shown by the fact that $\|\rho_{S}-\rho_{A}\|_{1}=2$ but $\|\rho_{S}-\rho_{A}\|_{\text{SEP}}=\|\rho_{S}-\rho_{A}\|_{\text{LOCC}}=2/(n+1)$~\cite{VV dh, VV dh Chernoff}. Curiously, there is an optimised version of this construction with different weights that does not seem to have been considered before. Namely, via the same techniques it can be shown that
\bq
\left\| \frac{n+1}{n}\, \rho_{S} - \frac{n-1}{n}\, \rho_{A} \right\|_{1}\, =\, 2\, ,\qquad \left\| \frac{n+1}{n}\, \rho_{S} - \frac{n-1}{n}\, \rho_{A} \right\|_{\text{SEP}}\, =\, \frac{2}{n}\, . \label{bound Werner}
\eq
Since the proof of~\eqref{bound Werner} is just a variation of a standard calculation, we relegate it to Appendix~\ref{app Werner}. Thanks to Proposition~\ref{dh ratio}, this yields the lower bound in the claim. Finally, combining $R_{QM}(\text{SEP})\geq n$ with the upper bound in~\eqref{VV bound 2}, we see that when $n_{A}=n_{B}=n$ we must have $R_{QM}(\text{SEP})=n$.
\end{proof}

\vspace{1ex}
\begin{rem}
The fact that the upper bound for $R_{QM}(\text{LOCC})$ in Theorem~\ref{dh QM} scales only linearly in $\min\{n_{A},n_{B}\}$ is crucial in solving the data hiding problem in quantum mechanics (up to constants) for all pairs $(n_{A}, n_{B})$. To our knowledge, this complete solution was not known before. We find the simplicity of the above proof quite instructive on its own, but it does not seem like the teleportation argument can encompass the case of purely local measurements, by its very nature. Therefore, we must leave open the problem of finding the optimal scaling of $R_{QM}(\text{LO})$ in the general case.
\end{rem}

\vspace{1ex}
\begin{rem}
In terms of the real dimensions of the local spaces, given by $d_A=n_A^2,\, d_B=n_B^2$ according to~\eqref{dim quantum}, Theorem~\ref{dh QM} shows that the data hiding ratio against separable protocols scales as $\min\{\sqrt{d_{A}},\,\sqrt{d_{B}} \}$. Thus, we deduce a first estimate $R_{\text{SEP}}(d_{A},d_{B})\geq \min\{\sqrt{d_{A}},\,\sqrt{d_{B}}\}$ (valid when $\sqrt{d_{A}},\sqrt{d_{B}}$ are integers).
\end{rem}

\vspace{2ex}
In view of Proposition~\ref{min is optimal}, the reader might wonder, whether considering a modified version of quantum mechanics in which composite systems are obtained via the minimal tensor product~\eqref{separable} exhibits better data hiding properties. In what follows, we call such a theory {\it witness theory} (or {\it $W$-theory}, for short). In $W$-theory, the only allowed states of a multipartite system are fully separable, while the set of possible effects contains all entanglement witnesses (equivalently, all elements in the cone~\eqref{witnesses}).
Thus, the base norm of a bipartite operator $X_{AB}$ will be given by
\bq
\|X\|_W\, \coloneqq\, \|X\|_{\text{QM}_{n_A}\tminfoot \text{QM}_{n_B}} =\, \max\big\{\Tr XY:\ Y\in \mathcal{H}_{n_An_B},\ \big|\!\braket{\alpha\beta|Y|\alpha\beta}\!\big|\leq \braket{\alpha|\alpha}\! \braket{\beta|\beta}\ \, \forall\, \ket{\alpha}\in\mathds{C}^{n_A},\, \ket{\beta}\in\mathds{C}^{n_B} \big\}\, .
\label{W norm}
\eq

Let us consider the class of protocols that can be realised in the $W$-theory framework when two agents $A$ and $B$ are allowed to use just local operations and $A\rightarrow B$ classical communication. We denote this class of protocols by $\text{{\bf LWCC }}_\rightarrow$. Despite the fact that $\text{{\bf LWCC }}_\rightarrow$ constitutes a more general class of {\it protocols} than $\text{LOCC}_\rightarrow$ in standard quantum mechanics, it is not difficult to see that they are not more powerful than the latter within the context of state discrimination. This is a consequence of the fact that the {\it measurements} one can implement with $\text{{\bf LWCC}}_\rightarrow$ operations are -- up to coarse graining -- necessarily of the form $\big(E_i\otimes F_j^{(i)}\big)_{(i,j)\in I\times J}$ for some local measurements $(E_i)_{i\in I}$ on $A$ and $\big( F^{(i)}_j\big)_{j\in J}$ on $B$. Since local measurements in $W$-theory are the same as in standard quantum mechanics, the same measurement is also obtainable via $\text{LOCC}_\rightarrow$ operations. More generally, we saw already that the locally constrained sets of measurements defined in~\eqref{LO},~\eqref{1-way LOCC},~\eqref{SEP} do not depend on the composition rule we choose for assembling multipartite systems.

The above discussion allows us to write the identity
\bq
\|X\|_{\text{LOCC}_{\rightarrow}} =\, \|X\|_{\text{{\bf LWCC}}_{\rightarrow}}\, ,
\label{confusing}
\eq
where the right-hand side is defined through the usual formula~\eqref{d norm}, and it is understood that the corresponding set $\mathcal{M}$ includes in this case all those measurements that are implementable through an $\text{{\bf LWCC}}_{\rightarrow}$ protocol.

% Let us note that while the set of {\it protocols} one can perform in $W$-theory with local operations and forward classical communication only, denoted by $\text{LWCC}_{\rightarrow}$, is strictly larger than the corresponding set of $\text{LOCC}_{\rightarrow}$ protocols, the difference between the two disappears as long as plain measurements are concerned. In other words, the set defined in~\eqref{1-way LOCC} is the same for standard bipartite quantum mechanics and $W$-theory, since it depends only on the local structure of the theories. We will write symbolically

% As a side remark, we note that unlike in quantum theory, in witness theory it is not true anymore that for all dimensions $d=n^{2}$ there is a unique type of system up to linear isomorphisms. Instead, we will have a number $B_{k_{n}}$ of different systems, where $B_{k}$ is the so-called Bell number counting the number of partitions of a set of $k$ elements, and $k_{n}$ is the number of prime factors in $n$, each counted with its multiplicity. {\color{red} [We can also omit this latter observation. In case we keep it, we should perhaps say why all the different cases are inequivalent.]}

As it turns out, data hiding in $W$-theory is not much better than in quantum theory, in the sense that the scaling with the local dimensions is exactly the same.
The proof of this latter result constitutes another example of how the techniques used in~\cite{VV dh Chernoff, VV dh} seem not to be applicable in a more general scenario. In fact, the approach taken there relies on the inequality $\|\cdot\|_{\text{LOCC}_{\rightarrow}}\geq C \|\cdot\|_2$ (for $C$ universal constant), where $\|\cdot\|_2$ stands for the Hilbert-Schmidt norm. When the system under examination is $\mathds{C}^n\otimes \mathds{C}^n$, the elementary relation $\|\cdot\|_2\leq \frac{1}{n} \|\cdot\|_1$ yields $\|\cdot\|_{\text{LOCC}_{\rightarrow}}\geq \frac{C}{n} \|\cdot\|_1$, which is optimal up to a constant factor. However, it is not difficult to prove that the ratio between the norm $\|\cdot\|_W$ defined in~\eqref{W norm} and the Hilbert-Schmidt norm $\|\cdot\|_2$ can be asymptotically as large as $n^{3/2}$ (see Appendix~\ref{app W norm}). Therefore, the tighter inequality we can deduce by making use of the Hilbert-Schmidt norm in an intermediate step is $\|\cdot\|_{\text{LOCC}_\rightarrow}\geq C \|\cdot\|_2 \geq \frac{C}{n^{3/2}} \|\cdot\|_W$. As we will see in a moment, the scaling of the lower bound is not tight.

While a direct approach via the other techniques previously exploited in the literature does not lead to a satisfactory answer to the problem, the teleportation argument can be quickly adapted to compute {\it exactly} the data hiding ratios against $\text{SEP}$ or $\text{LOCC}_{\rightarrow}$ in $W$-theory.

\vspace{2ex}
\begin{prop} \label{dh W}
For a bipartite $W$-theory with local Hilbert spaces of dimensions $n_{A},n_{B}$, define $n\coloneqq\min\{n_A,n_B\}$. Then the data hiding ratios against separable and $\text{LOCC}_{\rightarrow}$ measurements are given by
\bq
R_{W}(\text{\emph{SEP}})\, =\, R_{W}(\text{\emph{LOCC}}_{\rightarrow})\, =\, 2n-1\, ,
\eq
where the communication direction in $\text{LOCC}_{\rightarrow}$ is from the smaller to the larger subsystem.
\end{prop}

\begin{proof}
We assume without loss of generality that $n_B\geq n_A=n$. Let us start by showing that the above data hiding ratios can be upper bounded by $2n-1$. The idea is that the argument in~\eqref{telep 1}-\eqref{telep 5} can be adapted to encompass also the case of $W$-theory, by replacing everywhere $\|\cdot\|_{\text{LOCC}_{\rightarrow}}$ with $\|\cdot\|_{\text{LWCC}_{\rightarrow}}$ and the trace norm $\|\cdot\|_{1}$ with the minimal tensor product base norm~\eqref{W norm}.
The first step consists in acknowledging the fact that we can compute the $\text{LOCC}_{\rightarrow}$ distinguishability norm by making use of more general $\text{{\bf LWCC}}_{\rightarrow}$ protocols that are available in $W$-theory, as expressed in~\eqref{confusing}.

Now, we choose a particular $\text{{\bf LWCC}}_\rightarrow$ protocol in order to lower bound the norm $\|\cdot\|_{\text{\textbf{LWCC}}_{\rightarrow}}$. Such a protocol resembles the one we devised for the proof of Theorem~\ref{dh QM}, with one important difference.

Since separable states can be created with local operations and shared randomness, we can safely start by supplying $A$ and $B$ with a separable isotropic state $\left( \frac1n \Phi+\frac{n-1}{n}\, \sigma \right)_{A'B'}$, defined on an ancillary system $A'B'$ with local dimension $n_{A'}=n_{B'}=n$. As usual, $\sigma$ is an appropriate normalised, separable state. Then, we perform the teleportation $\tau$ defined in~\eqref{telep}, which is an $\text{{\bf LWCC}}_\rightarrow$ (even $\text{LOCC}_\rightarrow$) operation with respect to the splitting $AA'|BB'$, where classical communication goes from $AA'$ to $BB'$. After applying the triangle inequality, we are left with two terms, that is, $\|X_{B'B}\|_{\text{{\bf LWCC}}_\rightarrow}$ and $\|\tau(X_{AB}\otimes \sigma_{A'B'})\|_{\text{{\bf LWCC}}_\rightarrow}$. The first one can be computed exactly, since the operator $X_{B'B}$ obtained after teleportation belongs to the local subsystem $BB'$, and therefore measuring any witness $Y_{B'B}$ satisfying the constraints in~\eqref{W norm} is an allowed $\text{{\bf LWCC}}_\rightarrow$ operation, leading to the equality $\|X_{B'B}\|_{\text{{\bf LWCC}}_\rightarrow}=\|X\|_W$. As for the second term, we observe that $\|\tau(X_{AB}\otimes \sigma_{A'B'})\|_{\text{{\bf LWCC}}_\rightarrow}\leq \|X\|_{\text{{\bf LWCC}}_\rightarrow}$, since adding the ancillary system $A'B'$ in a separable state $\sigma_{A'B'}$ and subsequently applying $\tau$ is clearly an $\text{{\bf LWCC}}_\rightarrow$ protocol (which is why this latter inequality is in fact an equality). The above reasoning can be summarised in the following chain of inequalities, totally analogous to~\eqref{telep 1}-\eqref{telep 5}:
\begin{align*}
\|X_{AB}\|_{\text{LOCC}_\rightarrow}\, &=\, \|X_{AB}\|_{\text{\textbf{LWCC}}_\rightarrow} \\[0.5ex]
&=\, \left\| X_{AB} \otimes \left( \frac1n \Phi +\frac{n-1}{n}\, \sigma \right)_{A'B'} \right\|_{\text{\textbf{LWCC}}_\rightarrow} \\[0.5ex]
&\geq\, \left\| \tau\left(X_{AB} \otimes \left( \frac1n \Phi +\frac{n-1}{n}\, \sigma \right)_{A'B'}\right) \right\|_{\text{\textbf{LWCC}}_\rightarrow} \\[0.5ex]
&=\, \left\| \frac1n\, X_{B'B} +\frac{n-1}{n}\, \tau(X_{AB}\otimes \sigma_{A'B'}) \right\|_{\text{\textbf{LWCC}}_\rightarrow} \\
&\geq\, \frac1n\, \|X_{B'B}\|_{\text{\textbf{LWCC}}_\rightarrow} - \frac{n-1}{n}\, \|\tau(X_{AB}\otimes \sigma_{A'B'}) \|_{\text{\textbf{LWCC}}_\rightarrow} \\[0.5ex]
&\geq\, \frac1n\, \|X_{AB}\|_{W} - \frac{n-1}{n}\, \|X_{AB} \|_{\text{\textbf{LWCC}}_\rightarrow} \\[0.5ex]
&=\, \frac1n\, \|X_{AB}\|_{W} - \frac{n-1}{n}\, \|X_{AB} \|_{\text{LOCC}_\rightarrow}\, .
\end{align*}
In conclusion, we find
\bq
\|X_{AB}\|_{\text{LOCC}_{\rightarrow}}\, =\, \|X_{AB}\|_{\text{LWCC}_{\rightarrow}}\, \geq\,  \frac{1}{2n-1}\, \left\| X_{AB} \right\|_W\, , \label{bound telep W}
\eq
which implies $R_{W}(\text{LOCC}_{\rightarrow})\, \leq\, 2n-1$. Once more, to derive a lower bound on $R_{W}(\text{SEP})$ we use Werner states~\cite{Werner, Werner symmetry}. With the same notation as in the proof of Theorem~\ref{dh QM}, it can be shown that
\bq
\left\| \frac{n+1}{2n-1}\, \rho_{S}\, -\, \frac{n-1}{2n-1}\, \rho_{A} \right\|_{W}\, =\, 2\, ,\qquad \left\| \frac{n+1}{2n-1}\, \rho_{S}\, -\, \frac{n-1}{2n-1}\, \rho_{A} \right\|_{\text{SEP}}\, =\, \frac{2}{2n-1}\, , \label{bound Werner W}
\eq
enforcing the complementary bound $R_{W}(\text{SEP})\, \geq\, 2n-1$. The proof of~\eqref{bound Werner W} is provided in Appendix~\ref{app Werner}.
\end{proof}

\end{ex}

\vspace{2ex}
\begin{ex} \label{ex sph}
The state space of $\text{QM}_{2}$ is well-known to be identifiable with a $3$-dimensional ball (Bloch ball). This observation is the starting point for defining a hypothetical class of physical models whose state space is a Euclidean ball of arbitrary dimension. These GPTs have been considered recently in connections to attempts of reconstructing quantum mechanics starting from few physically motivated axioms~\cite{quantum 5 axioms, quantum 4 1/2 axioms, quantum info unit}. This is to be expected in light of a famous classification theorem by Koecher and Vinberg~\cite{Koecher hom cones, Vinberg hom cones}, implying among other things that spherical models are one of the few classes of models enjoying a distinctive property of quantum mechanics called strong self-duality~\cite{telep in GPT}.

In the following, given $x\in \mathds{R}^{d}$ we call $x_{i}$ its $i$th entry ($i=0,\ldots,d-1$), while $\widebar{x}$ will denote the column vector obtained from $x$ by eliminating the zeroth component. The GPT corresponding to this {\it spherical model} has the form
\begin{equation}
\text{Sph}_{d} \coloneqq \left( \mathds{R}^{d},\, C_d,\, u \right) .
\label{spherical}
\end{equation}
Here, $C_{d}$ is the `ice cream cone'
\bq
C_{d}\, \coloneqq\, \left\{\, x \in \mathds{R}^{d}:\ |\widebar{x}|_{2}\, \leq\, x_{0}\, \right\} , \label{ice cream}
\eq
where $|\cdot|_{2}$ is the standard Euclidean norm in $\mathds{R}^{d-1}$, and $u$ is defined as $u(y)\coloneqq y_{0}$ for all $y\in \mathds{R}^{d}$. In light of the canonical identification $\left( \mathds{R}^{d} \right)^{*}\simeq \mathds{R}^{d}$, we find it convenient to adopt from now on a column notation for the dual as well as for the primal space. Within this convention, we shall write unambiguously $u=(1,0,\ldots,0)^{T}$.
%where for the sake of simplicity we can stress the difference between the same column vector in the dual or in the primal space by denoting the latter with $u_{*}$.
This simplification of the notation is going to pay off because just like quantum mechanics, also the spherical model is self-dual, i.e. $C_{d}=C_{d}^{*}$.
As for the base norm in $\text{Sph}_{d}$, it can be shown that $\|x\|=\max\{|x_{0}|,|\widebar{x}|_{2}\}$.

Since our primary interest is in the exploration of the data hiding properties, according to Proposition~\ref{min is optimal} we construct a bipartite system $AB$ by joining two spherical models $A=\text{Sph}_{d_{A}}$ and $B=\text{Sph}_{d_{B}}$ via the minimal tensor product, i.e. taking $AB=\text{Sph}_{d_{A}}\tmin\text{Sph}_{d_{B}}$. This latter assumption will be made throughout the rest of this example. Tensors belonging to the bipartite vector space $\mathds{R}^{d_{A}}\otimes\mathds{R}^{d_{B}}$ (or to its dual) can be thought of as $d_{A}\times d_{B}$ matrices $X\in\mathds{R}^{d_{A}\times d_{B}}$. We shall find useful to denote by $\widebar{X}$ the $(d_{A}-1)\times (d_{B}-1)$ submatrix of $X\in\mathds{R}^{d_{A}\times d_{B}}$ which is obtained by cutting off the zeroth components $X_{i0},\, X_{0j}$ of the latter. Complementarily, given $M\in\mathds{R}^{(d_{A}-1)\times (d_{B}-1)}$ we call $\hat{M}$ the $d_{A}\times d_{B}$ `lifted' matrix whose entries are
\bqq
\hat{M}_{ij}\, =\, \left\{ \begin{array}{cl} 0 & \text{ if $i=0$ or $j=0$,} \\[0.5ex] M_{ij} & \text{ if $i,j\geq 1$.} \end{array}  \right.
\eqq

According to Proposition~\ref{dh ratio}, the data hiding ratio against separable measurements can be computed once we know the expressions for both the separability norm and the base norm induced by the minimal tensor product. Instead of treating the general case, we show how to compute these norms for a restricted yet large class of matrices, that is, those having zero entries in the zeroth row and column.

\vspace{2ex}
\begin{lemma} \label{norms sph}
Consider a bipartite system $AB=\text{\emph{Sph}}_{d_{A}}\tminit\text{\emph{Sph}}_{d_{B}}$. Then for all $M\in\mathds{R}^{(d_{A}-1)\times (d_{B}-1)}$ we have
\bq
\big\| \hat{M} \big\| = \|M\|_1\, ,\qquad \big\| \hat{M}\big\|_{\text{\emph{SEP}}} = \|M\|_\infty\, , \label{norms sph eq}
\eq
where it is understood that the base norm $\| \cdot \|$ is induced by the minimal tensor product, and $\|M\|_{1},\, \|M\|_{\infty}$ denote the trace and operator norm of $M$, i.e. the sum and the largest of its singular values, respectively.
\end{lemma}

\begin{proof}
Consider an arbitrary dual tensor $E\in\mathds{R}^{d_{A}\times d_{B}}$. We claim that if $E$ is separable then necessarily $\big\|\widebar{E}\,\big\|_{1}\leq E_{00}$, and that this condition is also sufficient when all cross terms $E_{i0}, E_{0j}$ ($i,j\geq 1$) vanish. Let us start by proving the necessity of the above inequality. If $E=\sum_{k}x^{(k)} \big(y^{(k)}\big)^{T}$ with $\big|\widebar{x}^{(k)}\big|_{2}\leq x^{(k)}_{0},\, \big|\widebar{y}^{(k)}\big|_{2}\leq y^{(k)}_{0}$ for all $k$, then
\bq
\big\|\widebar{E}\,\big\|_{1}\, =\, \bigg\|\sum_{k} \widebar{x}^{(k)}\big(\widebar{y}^{(k)}\big)^{T}\bigg\|_1\, \leq\,  \sum_{k} \Big\| \widebar{x}^{(k)} \big(\widebar{y}^{(k)}\big)^T\Big\|_{1}\, =\, \sum_{k} \big|\widebar{x}^{(k)}\big|_{2} \big|\widebar{y}^{(k)}\big|_{2}\, \leq\, \sum_{k} x^{(k)}_{0} y^{(k)}_{0}\, =\, E_{00}\, .
\label{norms sph proof 1}
\eq
Now, let us turn to the sufficiency claim. Suppose that $E_{ij}=0$ whenever $i=0,\, j\geq 1$ or $i\geq 1,\, j=0$, and that $\big\|\widebar{E}\,\big\|_{1}\leq E_{00}$. Then by the singular value decomposition theorem there are vectors $v^{(k)}\in\mathds{R}^{d_{A}-1}$ and $w^{(k)}\in\mathds{R}^{d_{B}-1}$ such that $\widebar{E}=\sum_{k} v^{(k)} \big(w^{(k)}\big)^T$ and $\big\|\widebar{E}\,\big\|_{1}=\sum_{k} \big|v^{(k)}\big|_{2} \big|w^{(k)}\big|_{2}$. Define vectors $x^{(k)}_\pm\coloneqq \big|v^{(k)}\big|_{2} u_{A} \pm \hat{v}^{(k)}\in \mathds{R}^{d_{A}}$, i.e.
\bq
\left(x_{\pm}^{(k)}\right)_{i}\, =\, \left\{ \begin{array}{cl} \big|v^{(k)}\big|_{2} & \text{ if $i=0$,} \\[0.5ex] \pm\, v^{(k)}_i & \text{ if $i\geq 1$,} \end{array} \right.
\label{norms sph proof 2}
\eq
and analogously for $y_{\pm}^{(k)}\coloneqq \big|w^{(k)}\big|_{2} u_{B} \pm \hat{w}^{(k)} \in \mathds{R}^{d_{B}}$. Observe that the definition~\eqref{ice cream} tells us that $x_\pm^{(k)} \in C_{d_A}$ and $y_\pm^{(k)} \in C_{d_B}$. Defining $U\coloneqq u_{AB}=u_A\otimes u_B=u_A u_B^T$, it is easy to see that
\bq
E\, =\, \left(E_{00} -\big\|\widebar{E}\,\big\|_1 \right) U\, +\, \frac12\, \sum_{k} \left( x^{(k)}_{+} \big(y^{(k)}_{+}\big)^{T} +  x^{(k)}_{-} \big(y^{(k)}_{-}\big)^{T} \right) ,
\label{norms sph proof 3}
\eq
which shows that $E$ is separable, as claimed.

Now, let us compute the separability norm. On the one hand, the above necessary condition for separability of effects shows that
\begin{align}
\big\|\hat{M}\big\|_\text{SEP}\, &=\, \max_{E,\,U-E\,\in\, C_{d_A}^*\tminfoot C_{d_B}^*} \left( \big|E\big(\hat{M}\big)\big| + \big|(U-E)\big(\hat{M}\big)\big| \right) \, =\, \max_{U \pm F\,\in\, C_{d_A}^*\tminfoot C_{d_B}^*} F\big(\hat{M}\big) \label{norms sph proof 4}\\
&\leq\, \max_{\left\|\bar{F}\right\|_1\, \leq\, 1\pm F_{00}} F\big(\hat{M}\big)\, =\, \max_{\left\|\bar{F}\right\|_1 +|F_{00}|\, \leq\, 1} \Tr \Big[M \widebar{F}^{\,T}\Big]\, =\, \|M\|_\infty\, ,
\label{norms sph proof 5}
\end{align}
where we employed~\eqref{d norm altern} to find a more compact expression for the separability norm, and we exploited the fact that trace norm and operator norm are dual to each other. On the other hand, the fact that $\|N\|_1\leq 1$ is sufficient to guarantee the separability of $U \pm \hat{N}$ leads us, again via~\eqref{d norm altern}, to the complementary bound
\bqq
\big\| \hat{M} \big\|_\text{SEP}\, \geq\, \max_{\|N\|_1\leq 1} \hat{N}(\hat{M})\, =\, \max_{\|N\|_1\leq 1} \Tr\left[ N^T M\right]\, =\, \|M\|_\infty\, .
%\label{norms sph proof 6}
\eqq

Our final task is the calculation of the base norm induced by the minimal tensor product. Thanks to the formula~\eqref{dual base eq}, we can write
\bqq
\big\|\hat{M}\big\|\, =\, \min \Big\{ U(X_+)+U(X_-):\ X_\pm\in C_{d_A}\tmin C_{d_B},\ \hat{M}=X_+ - X_- \Big\}\, .
%\label{norms sph proof 7}
\eqq
Since $X_\pm\in C_{d_A}\tmin C_{d_B}$ implies $U(X_\pm)\geq \big\|\widebar{X}_\pm\big\|_1$ and $\hat{M}=X_+ - X_-$ implies $M=\widebar{X}_+ - \widebar{X}_-$, we see that
\bqq
\big\|\hat{M}\big\|\, \geq\, \min \left\{ \big\|\widebar{X}_+\big\|_1 + \big\|\widebar{X}_-\big\|_1:\ M=\widebar{X}_+ - \widebar{X}_- \right\}\, \geq\, \|M\|_1\, ,
%\label{norms sph proof 8}
\eqq
where the last lower bound follows from the triangle inequality. On the other hand, we construct an ansatz for $X_\pm$ achieving the above lower bound. From the singular value decomposition theorem, it is immediately seen that for all real matrices $M$ there exists a decomposition $M=M_+ - M_-$ such that $\|M_+\|_1=\|M_-\|_1=\frac{\|M\|_1}{2}$, and consequently $\|M\|_1=\|M_+\|_1+\|M_-\|_1$. Then, consider $X_\pm=\frac{\|M\|_1}{2}\, U + \hat{M}_\pm$, so that $X_+-X_-=\hat{M}$. Since there are no cross terms $(X_\pm)_{i0}$ or $(X_\pm)_{0j}$, the condition $(X_\pm)_{00}\geq \big\|\widebar{X}_\pm\big\|_1$ (satisfied by construction) is sufficient to ensure the separability of $X_\pm$, hence this is a valid ansatz. We find
\bqq
\big\|\hat{M}\big\|\, \leq\, U(X_+) + U(X_-)\, =\, \|M_+\|_1+\|M_-\|_1\, =\, \|M\|_1\, ,
%\label{norms sph proof 9}
\eqq
and we are done.
\end{proof}

\vspace{1ex}
\begin{cor} \label{dh sph}
In the bipartite GPT $AB=\text{\emph{Sph}}_{d_{A}}\tminit\text{\emph{Sph}}_{d_{B}}$, the data hiding ratio against separable measurements can be lower bounded as
\bq
R_{\text{\emph{Sph}}}(\text{\emph{SEP}})\, \geq\, \min\{d_{A},d_{B}\} - 1\, .
\eq
In particular, the ultimate data hiding ratios against locally constrained sets of measurements obey
\bq
R_{\text{\emph{LO}}}(d_{A},d_{B}) \geq R_{\text{\emph{LOCC}}_{\rightarrow}}(d_{A},d_{B}) \geq R_{\text{\emph{LOCC}}}(d_A,d_B) \geq R_{\text{\emph{SEP}}}(d_{A},d_{B}) \geq \min\{d_{A},d_{B}\} - 1\, . \label{univ dh ratio lower}
\eq
\end{cor}

\begin{proof}
From Lemma~\ref{norms sph} we know that $R_\text{Sph}(\text{SEP})$ can be lower bounded by the maximal ratio between trace norm and operator norm of matrices in $\mathds{R}^{(d_{A}-1)\times (d_{B}-1)}$, which is well known to be $\min\{d_A-1,d_B-1\}=\min\{d_A,d_B\}-1$. Since we can provide an example of bipartite GPT with local dimensions $d_{A},d_{B}$ for which the data hiding ratio against separable measurements is no smaller than $\min\{d_A,d_B\}-1$, this constitutes a lower bound on the ultimate data hiding ratio $R_{\text{SEP}}(d_{A},d_{B})$.
\end{proof}

\vspace{2ex}
The problem of finding a complementary upper bound for the data hiding ratios in the spherical model is solved by the general result expressed in Theorem~\ref{thm univ} below, which implies that $R_{\text{Sph}}(\text{LO})\leq \min\{d_A,d_B\}$. This yields the almost tight, two-sided bound
\bq
\min\{d_{A},d_{B}\} - 1 \leq R_{\text{Sph}}(\text{SEP}) \leq R_{\text{Sph}}(\text{LOCC}) \leq R_{\text{Sph}}(\text{LOCC}_\rightarrow) \leq R_{\text{Sph}}(\text{LO}) \leq \min\{d_A,d_B\} ,
\label{double bound sph}
\eq
which fully determines the scaling of all the data hiding ratios against locally constrained measurements up to an additive constant.

\vspace{1ex}
\begin{rem}
Corollary~\ref{dh sph} shows that quantum mechanics, even when it is modified to encompass the minimal tensor product composition rule according to Proposition~\ref{min is optimal}, is not optimal from the point of view of data hiding, in the sense that there exist GPTs with the same local dimensions but with a higher data hiding ratio against all locally constrained sets of measurements.
\end{rem}

\end{ex}

\section{Ultimate bound on data hiding effectiveness} \label{sec univ}

Throughout this section, we will prove our main result on ultimate data hiding ratios against locally constrained sets of measurement. Namely, we will show that the lower bound we found in Corollary~\ref{dh sph} by analysing the spherical model is optimal up to an additive constant. As discussed at the end of Section~\ref{sec dh GPTs}, this answers the central question of our investigation, that is, the determination of the general constraints that data hiding is forced to obey in any GPTs.

Before coming to our main theorem, we need some mathematical preliminary on tensor norms. We remind the reader~\cite{DEFANT, RYAN} that given two finite-dimensional Banach spaces $V_{A},V_{B}$ (whose norms will be equally denoted by $\|\cdot\|$ for simplicity), there are two notable ways in which one can construct a norm on the tensor product $V_{A}\otimes V_{B}$. The first of these construction yields the so-called {\it injective norm}, which can be expressed as
\bq
\|X\|_{\varepsilon}\, \coloneqq\, \max\left\{ (f\otimes g)(X):\ \, f\in V_{A}^{*},\ \|f\|_{*}\leq 1,\ g\in V_{B}^{*},\ \|g\|_{*}\leq 1 \right\} .
\label{inj}
\eq
The second norm we are interested in goes under the name of {\it projective norm}, and is defined as
\bq
\|X\|_{\pi}\, \coloneqq\, \inf\left\{ \sum_{i=1}^{n} \|x_{i}\|\,\|y_{i}\| :\ \, n\in\mathds{N},\ X=\sum_{i=1}^{n} x_{i}\otimes y_{i} \right\} .
\label{proj}
\eq
These two norms are dual to each other in the following sense. Thinking of $V_{A}^{*}, V_{B}^{*}$ as Banach spaces equipped with the dual norms $\|\cdot\|_{*}$, we can construct the associated injective and projective tensor norms on $V_{A}^{*} \otimes V_{B}^{*}$, denoted by $\|\cdot\|_{*\varepsilon}$ and $\|\cdot\|_{*\pi}$, respectively. One could wonder, how these norms compare to the dual to~\eqref{inj} and~\eqref{proj}, denoted by $\|\cdot\|_{\varepsilon*}, \|\cdot\|_{\pi*}$, respectively. As it turns out, one has
\bq
\|\cdot\|_{*\varepsilon}\, =\, \|\cdot\|_{\pi*}\, ,\qquad \|\cdot\|_{*\pi}\, =\, \|\cdot\|_{\varepsilon*}\, .
\label{dual inj proj}
\eq

\vspace{1ex}
\begin{note}
The notation is intended to help the reader via simple graphic rules. For instance, the symbol $\|\cdot\|_{\varepsilon*}$ stands for `first construct the injective norm, then take the dual', and conversely $\|\cdot\|_{*\pi}$ means `first take the dual norms, then construct the projective norm out of them'. Then, the above identities can be recovered easily by remembering that taking a $*$ from the external to the internal position (or vice versa) causes an exchange $\varepsilon\leftrightarrow \pi$.
\end{note}

\vspace{2ex}
Concerning the comparison between injective and projective norms, the inequality $\|X\|_{\varepsilon}\, \leq\, \|X\|_{\pi}$ is easily seen to hold for all $X\in V_{A}\otimes V_{B}$. 
To see why this is the case, take a tensor $X\in V_A\otimes V_B$, and consider a decomposition $X=\sum_{i=1}^n x_i\otimes y_i$ as in~\eqref{proj}. For any two functionals $f\in V_A^*$ and $g\in V_B^*$ such that $\|f\|_*, \|g\|_*\leq 1$, one has $|f(x_i)|\leq \|x_i\|$ and $|g(y_i)|\leq \|y_i\|$, so that
\bqq
|(f\otimes g)(X)|\, =\, \bigg| \sum_{i=1}^n f(x_i) g(y_i) \bigg|\, \leq\, \sum_{i=1}^n |f(x_i)|\, |g(y_i) |\, \leq\, \sum_{i=1}^n \|x_i\|\, \|y_i\|\, \leq \, \|X\|_\pi\, ,
\eqq
where in the last step we employed~\eqref{proj}. Maximising over $f,g$ as in~\eqref{inj} yields the claim.
We will write symbolically
\bq
\|\cdot\|_{\varepsilon}\, \leq\, \|\cdot\|_{\pi}\, .
\label{inj=<proj}
\eq
As it turns out, for simple tensors $x\otimes y$ both norms coincide with the product of the local norms, i.e. $\|x\otimes y\|_{\varepsilon}=\|x\otimes y\|_{\pi}=\|x\|\, \|y\|$ (`reasonable cross norms').

In this work we are mainly interested in understanding the asymptotic behaviour of certain quantities as the dimension of the systems goes to infinity. This study fits in the so called local theory of Banach spaces, which is concerned in particular with the quantitative analysis of $d$-dimensional normed spaces (as $d\rightarrow \infty$). The investigation of tensor norms and their relations with operator ideals and factorizing operators is a crucial part of those studies~\cite{DEFANT}. Thus, it is not surprising that tensor norms play an important role in our approach. In this work, we will only care about the relation between projective and injective tensor norms when they are defined on finite dimensional Banach spaces. On the one hand, we are interested in the study of upper and lower bounds for the quotient $\|\cdot\|_{\pi}/ \|\cdot\|_{\varepsilon}$ for some particular spaces of relevance in our study. Let us mention in passing that the problem of comparing the projective and the injective tensor norms in the setting of infinite dimensional Banach leads to very deep results in the theory~\cite{Pisier, PisierII}. On the other hand, we ask ourselves what happens when we consider general finite dimensional Banach spaces with fixed dimensions. That is to say, we want to find the smallest $\mu(d_{A},d_{B})$ such that $\|\cdot\|_{\pi}\leq \mu\, \|\cdot\|_{\varepsilon}$ holds true for all Banach spaces $V_{A},V_{B}$ of dimensions $d_{A},d_{B}$, respectively. At this stage it is not even obvious that such a quantity will be finite. The answer to this question is provided by the following result, which could be known to experts in the topic, although we did not find any explicit reference.

\vspace{2ex}
\begin{prop} \label{prop proj=<ninj}
For all pairs of finite dimensional Banach spaces $V_{A},V_{B}$ with dimensions $d_{A}=\dim V_{A}$, $d_{B}=\dim V_{B}$, we have
\bq
\|\cdot\|_{\pi}\, \leq\, \min\{d_{A},d_{B}\}\, \|\cdot\|_{\varepsilon}\, .
\label{proj=<ninj}
\eq
Furthermore, the constant on the right-hand side of the above inequality is the best possible.
\end{prop}

\begin{proof}
We start by recalling Auerbach's lemma~\cite[Vol I, Sec. 1.c.3]{LT}, which states that any finite-dimensional Banach space admits a basis $\{v_{i}\}_{i}$, whose associated dual basis we denote by $\{v_{i}^{*}\}_{i}$ (so that $v_{i}^{*}(v_{j})=\delta_{ij}$), such that
\bqq
\|v_{i}\| \leq 1\, ,\quad \|v_{i}^{*}\|_{*}\leq 1\qquad \forall\ i\, .
\eqq
Suppose without loss of generality that $d_{A}\leq d_{B}$. Expand any tensor $X\in V_{A}\otimes V_{B}$ in the local Auerbach basis for $V_{A}$, that is, $X=\sum_{j=1}^{d_{A}} v_{j}\otimes y_{j}$. Choose $d_{A}$ functionals $g_{i}\in V_{B}^{*}$ such that $\|g_{i}\|_{*}\leq 1$ and $g_{i}(y_{i})=\|y_{i}\|$. Since also $\|v_{i}^{*}\|_{*}\leq 1$, using~\eqref{inj} we can lower bound the injective norm as
\bqq
\|X\|_{\varepsilon}\, \geq\, (v_{i}^{*}\otimes g_{i})(X)\, =\, \sum_{j=1}^{d_{A}} v_{i}^{*}(v_{j}) g_{i}(y_{j})\, =\, \sum_{j=1}^{d_{A}} \delta_{ij} g_{i}(y_{j})\, =\, g_{i}(y_{i})\, =\, \|y_{i}\|\, .
\eqq
Then, the definition of projective norm as given in~\eqref{proj} tells us that
\bqq
\|X\|_{\pi}\, \leq\, \sum_{i=1}^{d_{A}} \|v_{i}\|\, \|y_{i}\|\, \leq\, \sum_{i=1}^{d_{A}} \|v_{i}\|\, \|X\|_{\varepsilon}\, =\, d_{A}\, \|X\|_{\varepsilon}\, ,
\eqq
where we employed also the other defining property of the Auerbach basis, i.e. $\|v_{i}\|\leq 1$ for all $i$.

To see that the constant $\min\{d_{A}, d_{B}\}$ is optimal for all $d_{A},d_{B}$, we could appeal to Lemma~\ref{norms sph} together with the results of Subsection~\ref{subsec centr}. However, we resort here to an independent argument.

Consider two Euclidean spaces $V_{A},V_{B}$, whose norms we denote by $|\cdot|_{2}$. We can identify $V_{A}\otimes V_{B}$ with the set of $d_{A}\times d_{B}$ real matrices, denoted by $\mathds{R}^{d_{A}\times d_{B}}$. With this convention, for given $X\in V_{A}\otimes V_{B}$, $f\in V_{A}^{*}\simeq \mathds{R}^{d_{A}}$, and $g\in V_{B}^{*}\simeq \mathds{R}^{d_{B}}$, we can write $(f\otimes g)(X)=f^{T}Xg$. Then, remembering that Euclidean norms are self-dual, we obtain
\bqq
\|X\|_{\varepsilon}\, =\, \max \left\{ f^{T}Xg:\ |f|_{2*},|g|_{2*}\leq 1 \right\}\, =\, \max \left\{ f^{T}Xg:\ |f|_{2},|g|_{2}\leq 1 \right\}\, =\, \|X\|_{\infty}\, .
\eqq
The fact that $|\cdot|_{2*}=|\cdot|_{2}$ also implies that $\|\cdot\|_{*\varepsilon}=\|\cdot\|_{\varepsilon}$. This together with~\eqref{dual inj proj} shows that
\bqq
\|\cdot\|_{\pi}\,=\,\|\cdot\|_{*\varepsilon*}\,=\,\|\cdot\|_{\varepsilon*}\,=\,\|\cdot\|_{\infty *}\,=\,\|\cdot\|_{1}\, .
\eqq
Since $\|X\|_{1}=\min\{d_{A},d_{B}\}\, \|X\|_{\infty}$ whenever all the singular values of $X$ coincide, we can conclude.
\end{proof}

\vspace{2ex}

Note that the previous proof still applies if one of the spaces has infinite dimension. In fact, most of the results of this paper can be extended to that case straightforwardly. However, the case of two infinite-dimensional Banach spaces is completely different, and its study would require the use of more sophisticated Banach space techniques.

In order to apply Proposition~\ref{prop proj=<ninj} to our problem, we need to relate the distinguishability norms on a composite system to the injective and projective norm constructed out of the local base norms. This translation is the subject of our next result.

\vspace{2ex}
\begin{prop} \label{d norms inj proj}
Let $A=(V_{A},C_{A},u_{A})$ and $B=(V_{B},C_{B},u_{B})$ be two GPTs. The local base norms turn $V_{A},V_{B}$ into Banach spaces, and we can construct injective and projective tensor norms on $V_{A}\otimes V_{B}$, denoted simply by $\|\cdot\|_{\varepsilon}$ and $\|\cdot\|_{\pi}$. Considering the composite $A\tminit B$ given by the minimal tensor product as in~\eqref{min max theories}, we can also define on $V_{A}\otimes V_{B}$: (i) the base norm $\|\cdot\|_{A \tminitfoot B}$ associated with $A\tminit B$; and (ii) the local distinguishability norm $\|\cdot\|_{\text{\emph{LO}}}$. Then the following holds true:
\bq
\|\cdot\|_{\pi}\, =\, \|\cdot\|_{A \tminitfoot B}\, ,\qquad \|\cdot\|_{\varepsilon}\, \leq\, \|\cdot\|_{\text{\emph{LO}}}\, . \label{d norms inj proj eq}
\eq
\end{prop}

\begin{proof}
We start by showing the first relation in~\eqref{d norms inj proj eq}. Consider $X\in V_{A}\otimes V_{B}$, and decompose it as $X=\sum_{i=1}^{n} x_i\otimes y_i$ in such a way that $\|X\|_{\pi} =\sum_{i=1}^{n} \|x_i\|\,\|y_i\|$ (see the definition~\eqref{proj}). According to Lemma~\ref{dual base}, we can construct $x_i^{\pm}\in C_{A}$ and $y_i^{\pm}\in C_{B}$ such that $x_i=x_i^+ - x_i^-$, $\|x_i\|=u_A\big(x_i^{+}\big)+u_A\big(x_i^{-}\big)$ and analogously for $y_i^\pm$. Then,
\bqq
X\, =\, \sum_{i=1}^{n} \left(x_i^{+}\otimes y_i^{+} + x_i^{-} \otimes y_i^{-}\right)\, -\, \sum_{i=1}^{n} \left(x_i^{+}\otimes y_i^{-} + x_i^{-}\otimes y_i^{+}\right)\, ,
\eqq
and since 
\bqq
\sum_{i=1}^{n} \left(x_i^{+}\otimes y_i^{+} + x_i^{-} \otimes y_i^{-}\right),\  \sum_{i=1}^{n} \left(x_i^{+}\otimes y_i^{-} + x_i^{-} \otimes y_i^{+}\right)\ \in\ C_{A}\tmin C_{B}\, ,
\eqq
the dual formula~\eqref{dual base eq} yields
\begin{align*}
\|X\|_{A \tminfoot B}\, &\leq\, (u_{A}\otimes u_{B})\bigg( \sum_{i=1}^{n} \left(x_i^{+}\otimes y_i^{+} + x_i^{-} \otimes y_i^{-}\right)\bigg)\, +\, (u_{A}\otimes u_{B})\bigg( \sum_{i=1}^{n} \left(x_i^{+}\otimes y_i^{-} + x_i^{-} \otimes y_i^{+}\right)\bigg) \\
&=\, \sum_{i=1}^{n} \left( u_A\big(x_i^{+}\big)+u_A\big(x_i^{-}\big) \right) \left( u_B\big(y_i^{+}\big)+u_B\big(y_i^{-}\big)\right)\, =\, \sum_{i=1}^{n} \big\|x_i\big\|\, \big\|y_i\big\|\, =\, \|X\|_{\pi}\, .
\end{align*}
To show the converse, notice first that $\|X\|_{\pi}=(u_{A}\otimes u_{B})(X)$ holds true for all separable $X\in C_{A} \tmin C_{B}$. In fact, writing $X=\sum_{i} a_{i}\otimes b_{i}$ for $a_{i},b_{i}$ positive, we see that on the one hand $\|X\|_{\pi}\leq \sum_{i} \|a_{i}\|\, \|b_{i}\|= \sum_{i} u_{A}(a_{i}) u_{B}(b_{i}) = (u_{A}\otimes u_{B})(X)$, while on the other hand $(u_{A}\otimes u_{B})(X)\leq \|X\|_{\varepsilon}\leq\|X\|_{\pi}$ by inequality~\eqref{inj=<proj}. Now, for a generic $X\in V_{A}\otimes V_{B}$ apply once again Lemma~\ref{dual base} to construct $X_{\pm}\in C_{A}\tmin C_{B}$ such that $X=X_{+}-X_{-}$ and
\bqq
\|X\|_{A\tminfoot B}=(u_{A}\otimes u_{B})(X_{+}) + (u_{A}\otimes u_{B})(X_{-})\, .
\eqq
By using the above observation and the triangle inequality, we find
\bqq
\|X\|_{\pi} - (u_{A}\otimes u_{B})(X_{-})\, =\, \|X\|_{\pi} - \|X_{-}\|_{\pi}\, \leq\, \|X+X_{-}\|_{\pi}\, =\, \|X_{+}\|_{\pi}\, =\,  (u_{A}\otimes u_{B})(X_{+})\, ,
\eqq
from which
\bqq
\|X\|_{\pi}\, \leq\, (u_{A}\otimes u_{B})(X_{-}) + (u_{A}\otimes u_{B})(X_{+})\, =\, \|X\|_{A \tminfoot B}\, .
\eqq
This completes the proof of the equality $\|\cdot\|_{\pi}\, =\, \|\cdot\|_{A \tminfoot B}$. Now, we move on to the second relation in~\eqref{d norms inj proj eq}. To find a lower bound on the separability norm, we will employ the expression~\eqref{d norm altern} with $\mathcal{M}=\text{LO}$. For arbitrary $X\in V_{A}\otimes V_{B}$ and $f\in V_{A}^{*},\, g\in V_{B}^{*}$ such that $\|f\|_{*}, \|g\|_{*}\leq 1$, we show that $f\otimes g$ is a valid test functional to be plugged into~\eqref{d norm altern} since $\mu\coloneqq \left( \frac12 (u_{A}\otimes u_{B} + f\otimes g),\, \frac12 (u_{A}\otimes u_{B} - f\otimes g)\right) \in \text{LO}$. In fact, notice that
\begin{align}
\frac12 \left( u_{A}\otimes u_{B} + f\otimes g \right)\, &=\, \frac{u_{A}+f}{2}\otimes \frac{u_{B}+g}{2}\, +\, \frac{u_{A}-f}{2}\otimes \frac{u_{B}-g}{2}\, , \label{fg local 1} \\
\frac12 \left( u_{A}\otimes u_{B} - f\otimes g \right)\, &=\, \frac{u_{A}+f}{2}\otimes \frac{u_{B}-g}{2}\, +\, \frac{u_{A}-f}{2}\otimes \frac{u_{B}+g}{2}\, . \label{fg local 2}
\end{align}
Thanks to the fact that the unit balls of the dual to the local base norms have the form $B_{\|\cdot\|_{*}}=[-u,u]$ (Definition~\ref{def base}), we know that $\left( \frac12 (u_{A}+ f),\, \frac12(u_{A}-f) \right)$ is a valid measurement on $A$, and analogously $\left( \frac12 (u_{B}+ g),\, \frac12(u_{B}-g) \right)$ is a measurement on $B$. Using~\eqref{LO}, it is easy to see that $\mu$ is indeed obtainable from a product measurement via a coarse graining procedure as defined in~\eqref{coarse}. Then, equation~\eqref{d norm altern} yields $\|X\|_{\text{LO}}\geq (f\otimes g)(X)$, which becomes in turn $\|X\|_{\text{LO}}\geq \|X\|_{\varepsilon}$ once we maximise over the functionals $f,g$ satisfying $\|f\|_{*},\|g\|_{*}\leq 1$.
\end{proof}

\vspace{1ex}
\begin{cor} \label{cor d norms}
Any global base norm $\|\cdot\|$ in a bipartite GPT constructed according to~\eqref{CAB bound} must obey $\|x\otimes y\|=\|x\otimes y\|_{\text{\emph{LO}}}=\|x\|\, \|y\|$ for all $x\in V_A$ and $y\in V_B$, where $\|\cdot\|$ stands for (both) the local base norms. 
\end{cor}

\begin{proof}
On the one hand, since~\eqref{CAB bound} holds, the global base norms must be upper bounded by the one associated with the minimal tensor product, which coincides with $\|\cdot\|_\pi$ by Proposition~\ref{d norms inj proj}. On the other hand, the second identity in~\eqref{d norms inj proj eq} ensures that $\|\cdot\|_{\text{LO}}\geq \|\cdot\|_\varepsilon$. We already saw how injective and projective norm coincide on simple tensors. Then, putting all together we obtain
\bqq
\|x\|\, \|y\|\, =\, \| x \otimes y\|_\varepsilon\, \leq\, \|x\otimes y\|_{\text{LO}}\, \leq\, \|x\otimes y\|\, \leq\, \|x\otimes y\|_{A \tminfoot B}\, =\, \|x\otimes y\|_\pi\, =\, \|x\| \, \|y\|\, ,
\eqq
concluding the proof.
\end{proof}

\vspace{2ex}
We are finally ready to prove one of our main results, that is, the quasi-optimality of the lower bound~\eqref{univ dh ratio lower}.

\vspace{2ex}
\begin{thm}[Upper bound on ultimate data hiding ratios] \label{thm univ} 
Let $A,B$ be two GPTs, and let their composite $AB$ obey~\eqref{CAB bound}. Then the data hiding ratios against locally constrained sets of measurements satisfy the bound
\bq
R(\text{\emph{SEP}})\, \leq\, R(\text{\emph{LOCC}})\, \leq\, R(\text{\emph{LOCC}}_\rightarrow)\, \leq\, R(\text{\emph{LO}})\, \leq\, \max_{0\neq X\in V_A\otimes V_B} \frac{\| X \|_\pi}{\| X \|_\varepsilon}\, . \label{thm univ eq1}
\eq
In particular, the corresponding ultimate data hiding ratios, as given in Definition~\ref{univ dh ratio}, are upper bounded as follows:
\bq
R_{\text{\emph{SEP}}}(d_{A},d_{B})\, \leq\, R_{\text{\emph{LOCC}}}(d_A,d_B)\, \leq\, R_{\text{\emph{LOCC}}_{\rightarrow}}(d_{A},d_{B})\, \leq\, R_{\text{\emph{LO}}}(d_{A},d_{B})\, \leq\, \min\{d_A,d_B\}\, . \label{thm univ eq2}
\eq
\end{thm}

\begin{proof}
Since the inequalities~\eqref{chain dh} hold, we have to upper bound only the data hiding ratio against local operations. Using~\eqref{dh ratio} and Proposition~\ref{min is optimal}, we know that for a fixed pair of GPTs $A=(V_{A},C_{A},u_{A})$ and $B=(V_{B},C_{B},u_{B})$, this latter ratio satisfies $R(\text{LO})\leq \max_{X\neq 0} \|X\|_{A \tminfoot B} \big/ \|X\|_{\text{LO}}$. Now, Proposition~\ref{d norms inj proj} states that $\|X\|_{A \tminfoot B}=\|X\|_{\pi}$ and $\|X\|_{\text{LO}}^{-1}\leq \|X\|_{\varepsilon}^{-1}$, from which we deduce~\eqref{thm univ eq1}. Using also Proposition~\ref{prop proj=<ninj}, we see that the right-hand side of~\eqref{thm univ eq1} is upper bounded by $\min\{d_A,d_B\}$ for all GPTs of fixed local dimensions $d_A,d_B$, proving also~\eqref{thm univ eq2}.
\end{proof}

\vspace{1ex}
\begin{rem}
Specialising Theorem~\ref{thm univ} to the modified version of quantum mechanics that we called $W$-theory in Example~\ref{ex QM} yields an improvement of~\cite[Lemma 20]{Brandao area law}, in the form of the inequality
\bq
\|X\|_W\, \leq\, n^2 \|X\|_\varepsilon\, ,
\label{improvement Brandao}
\eq
valid for all Hermitian operators $X$ on $\mathds{C}^{n_A}\otimes\mathds{C}^{n_B}$, where $n=\min\{n_A, n_B\}$, as usual. Observe that the right-hand side of~\eqref{improvement Brandao} coincides with the right-hand side of~\cite[Eq. (C1)]{Brandao area law}, while the left-hand side of~\eqref{improvement Brandao} is larger than the corresponding left-hand side of~\cite[Eq. (C1)]{Brandao area law}, as follows from~\eqref{Cmin largest norm},~\eqref{W norm}, and from the fact that the trace norm $\|\cdot\|_1$ is the base norm corresponding to the standard quantum mechanical composition rule.

Let us stress that neither of the two bounds $R_{QM}(\text{LO})\leq c\, \sqrt{n_A n_B}$~\cite{VV dh} and $R_{QM}(\text{LO})\leq \min\{n_A^2, n_B^2\}$~\cite{Brandao area law} is tight. If $n_A$ and $n_B$ are very different from each other the latter bound will be more effective, while if they are of the same order the former will be preferable. In conclusion, as we already highlighted when discussing data hiding in quantum mechanics, determining the optimal scaling of $R_{QM}(\text{LO})$ remains an interesting open problem. At the same time, this example shows how consequences drawn from general results like Theorem~\ref{thm univ} can shed some light even on well-studied problems.
\end{rem}

\section{Data hiding in special classes of GPTs} \label{sec special}

Until now, we have been mainly interested in investigating the strongest cases of data hiding, in order to study the ultimate, intrinsic bounds characterising this non-classical phenomenon. For how well-motivated this inclination to universality can be, throughout this section we want to take a different approach and look into particular classes of models, chosen as suitable generalisations of Examples~\ref{ex QM} and~\ref{ex sph}.

Subsection~\ref{subsec centr} is concerned with all the GPTs whose state space, just like in the spherical model, is centrally symmetric. It turns out that in this case all the data hiding ratios against locally restricted sets of measurements are equal up to an additive constant, and moreover can be computed as a maximal ratio between a projective and an injective norm. On the one hand, this is reminiscent of the proof of Theorem~\ref{thm univ}, where we upper bounded $R(\text{LO})$ with the maximal projective-injective ratio induced by the local base norms. However, while in that case the only piece of information we could extract was the existence of the upper bound, here computing such a quantity yields the exact data hiding ratio up to additive constants (Theorem~\ref{centr ratio}). On the other hand, we see thanks to Example~\ref{ex sph} how this reduction can simplify the solution of the model. In that case, in fact, we were able to translate the data hiding problem into the computation of the maximal ratio between two matrix norms, and the answer to this latter question was a natural consequence of basic theorems in matrix analysis. We show how this approach can be pushed further, by using the machinery we develop to solve another `natural' model (Example~\ref{cubic}).

Throughout Subsection~\ref{subsec Werner}, we look into those GPTs whose vector space is endowed with a representation of a compact group that maps states to states but is otherwise irreducible on the section $\{x\in V:\, u(x)=0\}$. Although this assumption is quite strong, we show how it encompasses the main physically relevant examples of GPTs, like classical probability theory and quantum theory. With the group integral at hand, we are able to generalise the construction of Werner states and to use them for estimating data hiding in terms of simple geometrical parameters of the model (Theorem~\ref{dh Werner}).

\subsection{Centrally symmetric models} \label{subsec centr}

The solution of the spherical model we gave in Example~\ref{ex sph} through Lemma~\ref{norms sph} and Corollary~\ref{dh sph} was based on two remarkable facts: on the one hand, we could achieve the data hiding ratio (up to an additive constant) by employing only tensors of the form $\hat{M}$ for $M\in \mathds{R}^{(d_{A}-1)\times (d_{B}-1)}$, and on the other hand we were able to find simple expressions for base and separability norm of tensors of this simplified form. Here, we want to take the chance to generalise these intuitions a bit further, to encompass any model whose state space is centrally symmetric. Let us start with the following definition.

\vspace{2ex}
\begin{Def}[Centrally symmetric models] \label{centr}
A GPT of the form $(\mathds{R}^{d}, C, u)$ such that $u=(1,0,\ldots,0)^T$ is said to be \emph{centrally symmetric} if there exists a norm $|\cdot|$ on $\mathds{R}^{d-1}$ such that $C=\{(x_{0},\widebar{x})\in \mathds{R}\oplus\mathds{R}^{d-1}:\ x_{0}\geq |\widebar{x}|\}$.
\end{Def}

\vspace{2ex}
As can be easily verified, the dual cone to $C=\{(x_{0},\widebar{x})\in \mathds{R}\oplus\mathds{R}^{d-1}:\ x_{0}\geq |\widebar{x}|\}$ shares the same structure, being given by
\bqq
C ^*\, =\, \{(y_{0},\widebar{y})\in \mathds{R}\oplus\mathds{R}^{d-1}:\ y_{0}\geq |\widebar{y}|_*\}\, ,
\eqq
where $|\cdot|_*$ is the dual to the norm $|\cdot|$. The base norm of a centrally symmetric GPT is given by $\left\|(x_0,\widebar{x})\right\|=\max\{|x_0|, |\widebar{x}|\}$. Another peculiarity of centrally symmetric models is the existence of a privileged state, denoted by $u_{*}\coloneqq (1,0,\ldots,0)^T$. The vector space can then be written as 
\bq
\mathds{R}^{d}\, =\, \mathds{R} u_{*}\oplus \mathds{R}^{d-1}\, .
\label{V decomp}
\eq
We will keep denoting by $x_0, \widebar{x}$ the two components of $x\in \mathds{R}^d$ according to the above decomposition. Along the same lines, for $v\in \mathds{R}^{d-1}$ we will write $\hat{v}\coloneqq 0\oplus v$. A remarkable feature of centrally symmetric models is the existence of a simple linear map $\Lambda:\mathds{R}^{d}\rightarrow \mathds{R}^{d}$, given by
\bq
\Lambda \,\coloneqq\, 1 \oplus (-I_{d-1})
\label{Lambda}
\eq
according to the decomposition~\eqref{V decomp} and with $I_{d-1}$ denoting the identity on $\mathds{R}^{d-1}$, that is also an order isomorphism (i.e. such that $\Lambda(C)=C$). Consequently, for all $x\in \mathds{R}^d$ one has $\|x\|=\|\Lambda(x)\|$ (also obvious from the explicit formula for the base norm given above). Now, let us turn to bipartite systems. From~\eqref{V decomp} we infer the natural decomposition
\bq
V_{AB}\, =\, \mathds{R}^{d_{A}}\otimes \mathds{R}^{d_{B}}\, =\, \left(\mathds{R}\, u_{A*}\otimes u_{B*}\right) \oplus \left(u_{A*}\otimes \mathds{R}^{d_{B}-1}\right) \oplus \left(\mathds{R}^{d_{A}-1}\otimes u_{B*}\right) \oplus \left(\mathds{R}^{d_{A}-1}\otimes  \mathds{R}^{d_{B}-1}\right) . \label{VAB decomp}
\eq
As a side remark, observe that the maps $\Lambda\otimes I,\, I\otimes \Lambda,\, \Lambda\otimes \Lambda$ preserve separability of states and (up to an irrelevant transposition) separability of effects. Therefore, the norm $\|\cdot\|_{\text{SEP}}$ is left invariant by any of these maps. As is easy to verify, the same is true for all the four locally constrained distinguishability norms, with the possible exception of $\|\cdot\|_{\text{LOCC}}$.

Concerning the analysis of the data hiding properties, what we did in the case of the spherical model was to consider only the latter of the above addends, and to a posteriori justify this restriction by employing Theorem~\ref{thm univ} to show that the obtained result is tight up to an additive constant. As it turns out, this procedure can be always followed for centrally symmetric GPTs without loss of generality. This is the content of our first result. As throughout Example~\ref{ex sph}, for a tensor $X\in V_{AB}$ we denote by $\widebar{X}$ its component pertaining to the fourth addend of~\eqref{VAB decomp}.

\vspace{2ex}
\begin{prop} \label{prop centr R r}
Let $AB$ be a bipartite system formed by two centrally symmetric GPTs joined with any rule that respects~\eqref{CAB bound}. Then for all $X\in V_{AB}$ we have
\bq
\frac{\|X\|}{\|X\|_{\text{\emph{LO}}}}\, \leq\, \frac{\big\|\widebar{X}\big\|}{\big\|\widebar{X}\big\|_{\text{\emph{LO}}}} + 2\, . \label{prop centr R r eq}
\eq
\end{prop}

\begin{proof}
Consider $X=s\, u_*\otimes u_* + a\otimes u_* + u_*\otimes b + \widebar{X}\in V_{AB}$, where we omitted the subscripts $A,B$ for the sake of simplicity. Applying the triangle inequality to the global base norm, we see that $\|X\|\leq \| (s\,u_*+a)\otimes u_*\| + \|u_*\otimes b\| + \big\|\widebar{X}\big\|$. Now, Corollary~\ref{cor d norms} guarantees that $\| (s\,u_*+a)\otimes u_*\| = \| (s\,u_*+a)\otimes u_*\|_{\text{LO}}=\|s\,u_*+a\|$ and $\|u_*\otimes b\|=\|u_*\otimes b\|_{\text{LO}}=\|b\|$. Moreover, it is also clear by discarding the system $B$ and performing an arbitrary operation on $A$, that $\|X\|_{\text{LO}}\geq \|s\, u_*+a\|$. Proceeding in an analogous fashion with exchanged subsystems, remembering that $\Lambda$ defined in~\eqref{Lambda} leaves the base norm invariant, and exploiting the triangle inequality, yields
\bqq
\|X\|_{\text{LO}}\, \geq\, \|s\, u_* + b\|\, =\, \|\Lambda (s\, u_*+b)\|\, =\, \|s\,u_* - b\|\, \geq\, \left\|\, \frac12 (s\, u_* + b) - \frac12 (s\, u_* - b) \, \right\|\, =\, \|b\|\, .
\eqq
Finally, using the readily verified invariance of $\|\cdot\|_{\text{LO}}$ under any of the maps $\Lambda\otimes I,\, I\otimes \Lambda,\, \Lambda\otimes \Lambda$, we obtain
\bqq
\|X\|_{\text{LO}}\, \geq\,  \left\| \,\frac14 X - \frac14 (\Lambda\otimes I)(X) - \frac14 (I\otimes \Lambda)(X) + \frac14 (\Lambda\otimes\Lambda)(X)\, \right\|_{\text{LO}}\, =\, \big\|\widebar{X}\big\|_{\text{LO}}\, .
\eqq
Putting all together, we find
\begin{align*}
\frac{\|X\|}{\|X\|_{\text{LO}}}\, &\leq\, \frac{\| s\,u_*+a\| + \|b\| + \big\|\widebar{X}\big\|}{\|X\|_{\text{LO}}}\, =\, \frac{\| s\,u_*+a\|}{\|X\|_{\text{LO}}}\, +\, \frac{\|b\|}{\|X\|_{\text{LO}}}\, +\, \frac{\big\|\widebar{X}\big\|}{\|X\|_{\text{LO}}}\, \leq\, 2 \, +\, \frac{\big\|\widebar{X}\big\|}{\big\|\widebar{X}\big\|_{\text{LO}}}
\end{align*}
\end{proof}

\vspace{2ex}
What Proposition~\ref{prop centr R r} is telling us is that up to an additive constant we can restrict the search for data hiding against local operations to projected tensors of the form $\widebar{X}$. From now on, we denote by $\widebar{R}(\mathcal{M})$ the {\it restricted data hiding ratio} against a set of measurements $\mathcal{M}$ that is obtained by considering only those tensors. With this notation, Proposition~\ref{prop centr R r} can be cast into the form of the inequality $R(\text{LO})\leq \widebar{R}(\text{LO})+2$. In order to carry out the analysis of restricted ratios, we need an analogue of Lemma~\ref{norms sph}.

\vspace{2ex}
\begin{prop} \label{norms centr}
Let $|\cdot|_\varepsilon,|\cdot|_\pi$ be the injective and projective norm constructed on $\mathds{R}^{(d_A-1)\times (d_B-1)}\simeq \mathds{R}^{d_A-1}\otimes \mathds{R}^{d_B-1}$ out of the local norms $|\cdot|$ on $\mathds{R}^{d_A-1}, \, \mathds{R}^{d_B-1}$. Then for $M\in\mathds{R}^{(d_A-1)\times (d_B-1)}$, whose embedding into $V_{AB}$ according to~\eqref{VAB decomp} we denote by $\hat{M}\coloneqq 0\oplus 0\oplus 0\oplus M$, the locally constrained distinguishability norms and the global base norm induced by the minimal tensor product are respectively given by 
\begin{align}
&\big\|\hat{M}\big\|_{\text{\emph{LO}}}\, =\, \big\|\hat{M}\big\|_{\text{\emph{LOCC}}_\rightarrow}\, =\, \big\|\hat{M}\big\|_{\text{\emph{LOCC}}}\, =\, \big\|\hat{M}\big\|_{\text{\emph{SEP}}}\, =\, |M|_\varepsilon\, , \label{norms centr eq1} \\
&\big\|\hat{M}\big\|_{A \tminitfoot B}\, =\, |M|_\pi\, .
\end{align}
\end{prop}

\begin{proof}
The argument follows the guidelines of the proof of Proposition~\ref{norms sph}, so we omit some of the details. First of all, we notice that if a dual tensor $E\in\mathds{R}^{d_{A}\times d_{B}}$ is separable then necessarily $\big|\widebar{E}\big|_{*\pi} \leq E_{00}$, and this condition is also sufficient when all cross terms $E_{i0}, E_{0j}$ ($i,j\geq 1$) vanish. The proof of this claim follows exactly the steps of~\eqref{norms sph proof 1},~\eqref{norms sph proof 2} and~\eqref{norms sph proof 3}, with the dual norm $|\cdot|_*$ displacing the (self-dual) Euclidean norm, and the associated projective norm $|\cdot|_{*\pi}$ displacing the trace norm.

To compute the separability norm, on the one hand as in~\eqref{norms sph proof 4} and~\eqref{norms sph proof 5} we find
\bqq
\big\|\hat{M}\big\|_{\text{SEP}}\, \leq\, \max_{|\bar{F}|_{*\pi} + |F_{00}|\, \leq\, 1} \Tr \Big[ \widebar{F}^{\,T} M\Big]\, =\, |M|_{*\pi*}=|M|_\varepsilon \, ,
\eqq
where in the last step we used~\eqref{dual inj proj}. On the other hand, the complementary inequality $\big\|\hat{M}\big\|_{\text{LO}}\geq |M|_\varepsilon$ is a simple consequence of the second relation in~\eqref{d norms inj proj eq}. In fact, since for all $v\in\mathds{R}^{d_A-1}$ the identity $|v|_*=\|\hat{v}\|$ holds true, we obtain
\bqq
\big\|\hat{M}\big\|_{\text{LO}}\, \geq\, \big\|\hat{M}\big\|_\varepsilon\, =\, \max_{\|f\|_*,\|g\|_*\leq 1} (f\otimes g)\big(\hat{M}\big)\, \geq\, \max_{|v|_*,|w|_*\leq 1} v^T M w\, =\, |M|_\varepsilon\, .
\eqq
Putting all together, we obtain
\bqq
|M|_\varepsilon\, \leq\, \big\|\hat{M}\big\|_{\text{LO}}\, \leq\, \big\|\hat{M}\big\|_{\text{LOCC}_\rightarrow}\, \leq\, \big\|\hat{M}\big\|_{\text{LOCC}} \leq\, \big\|\hat{M}\big\|_{\text{SEP}}\, \leq\, |M|_\varepsilon\, ,
\eqq
from which~\eqref{norms centr eq1} follows.
\end{proof}

\vspace{1ex}
\begin{thm} \label{centr ratio}
For a fixed pair of centrally symmetric GPTs whose composite is formed with the minimal tensor product rule according to Proposition~\ref{min is optimal}, all the four restricted data hiding ratios against locally constrained measurements coincide, and can be computed as
\bq
\widebar{R}\, \coloneqq\, \widebar{R}(\text{\emph{LO}})\, =\, \widebar{R}(\text{\emph{LOCC}}_\rightarrow)\, =\, \widebar{R}(\text{\emph{LOCC}})\, =\, \widebar{R}(\text{\emph{SEP}})\, =\, \max_{0\,\neq\, M\,\in\, \mathds{R}^{(d_A-1)\times (d_B-1)}} \frac{\,|M|_\pi}{\,|M|_\varepsilon}\, . \label{centr ratio eq1}
\eq
Moreover, the `true' data hiding ratios are equal to the restricted ones up to an additive constant. Namely,
\bq
\widebar{R}\, \leq\, R(\text{\emph{SEP}})\, \leq\, R(\text{\emph{LOCC}})\, \leq\, R(\text{\emph{LOCC}}_\rightarrow)\, \leq\, R(\text{\emph{LO}})\, \leq\, \widebar{R} + 2\, . \label{centr ratio eq2}
\eq
\end{thm}

\begin{proof}
First of all,~\eqref{centr ratio eq1} follows directly from Proposition~\ref{norms centr}. To show~\eqref{centr ratio eq2}, note that, by definition, restricted data hiding ratios are smaller than the true ones. In particular, $R(\text{SEP})\geq \widebar{R}(\text{SEP})=\widebar{R}$. Using~\eqref{chain dh} and the result of Proposition~\ref{prop centr R r} in the form of the inequality $R(\text{LO})\leq \widebar{R}(\text{LO})+2$, we obtain immediately~\eqref{centr ratio eq2}.
\end{proof}

\vspace{1ex}
\begin{rem}
Sometimes the quantity $\widebar{R}$ appearing in~\eqref{centr ratio eq1} can be better computed with the help of a simple trick. Namely, remember that for all norms $|\cdot|^{(1)},|\cdot|^{(2)}$ on any space $V$ we have $\max_{0\neq x\in V} \frac{|x|^{(1)}}{|x|^{(2)}} = \max_{0\neq f\in V^*} \frac{|f|^{(2)}_*}{|f|^{(1)}_*}$. Thanks to this identity and to the duality relation~\eqref{dual inj proj}, we see that $\widebar{R}$ can alternatively be expressed as
\bq
\widebar{R}\, =\, \max_{0\,\neq\, M\,\in\, \mathds{R}^{(d_A-1)\times (d_B-1)}} \frac{\,|M|_{*\varepsilon}}{\,|M|_{*\pi}}\, . \label{Rbar dual}
\eq
\end{rem}

\vspace{2ex}
Theorem~\ref{centr ratio} suggests a precise recipe for computing data hiding ratios of centrally symmetric models. It is in fact enough to analyse the behaviours of injective and projective norms, compute $\widebar{R}$ defined in~\eqref{centr ratio eq1}, and finally use~\eqref{centr ratio eq2} to determine all the ratios against locally constrained measurements up to a universal additive constant. In fact, this is basically what we did in Example~\ref{ex sph}, since trace norm and operator norm are exactly the injective and projective norms constructed out of local Euclidean spaces, as we saw in the proof of Proposition~\ref{prop proj=<ninj}. In the remaining part of this subsection, we are going to show how to apply all this machinery to solve another explicit example.

\vspace{2ex}
\begin{ex}[Cubic model] \label{cubic} 

A generalised bit, or {\it gbit} for short~\cite{Barrett original}, is perhaps the simplest example of a non-classical system. As a GPT, we can define it as $G_{2} \coloneqq \left( \mathds{R}^{3},\, C_{\square},\, u \right)$, where $C_{\square} \coloneqq \{ (x,y,z)\in\mathds{R}^{3}:\ |x|+|y|\leq z \}$ and $u(x,y,z)\coloneqq z$. As is easily seen, the state space of $G_{2}$ is a two-dimensional square. Thus, a straightforward generalisation of the above concept takes as state space an $n$-dimensional hypercube. For simplicity, in the following we refer to $G_n$ as a {\it cubic model}, a GPT with dimension $\dim G_n=n+1$.

The main reason for introducing these toy models is that they provide a useful tool to understand non-local correlations. In fact, the celebrated Popescu-Rohrlich (PR) box~\cite{PR boxes} can be seen as a maximal tensor product of the form $G_{2}\tmax G_{2}$. Here, we are mainly interested in discussing the data hiding properties of GPTs of the form $G_n\tmin G_m$, for all positive integers $n,m$. Since the $n$-dimensional hypercube is centrally symmetric, this is a perfect playground for testing the machinery we developed throughout this subsection. Notice that in the case of $G_n$ the norm $|\cdot|$ on $\mathds{R}^n$ appearing in Definition~\ref{centr} is given by the $\infty$-norm $|v|_\infty\coloneqq \max_{1\, \leq\, i\, \leq\, n} |v_i|$, with dual $|w|_{\infty*}=|w|_1=\sum_{i=1}^n |w_i|$. It is worth noticing that the extreme points of the unit ball $B_{|\cdot|_1}$ coincide with the elements of the standard basis of $\mathds{R}^n$ up to a sign.

According to Theorem~\ref{centr ratio}, the first step in solving the GPT $G_n\tmin G_m$ is the determination of the restricted data hiding ratio $\widebar{R}=\max_{0\neq M\in \mathds{R}^{n\times m}} \frac{|M|_\pi}{|M|_\varepsilon}$, where $|\cdot|_\varepsilon,|\cdot|_\pi$ denote respectively the injective and projective norm induced on $\mathds{R}^{n\times m}$ by the local $\infty$-norms. We already observed how this can be dually rephrased as~\eqref{Rbar dual}. In our case, this is helpful because the injective and projective dual norms $|\cdot|_{*\varepsilon}, |\cdot|_{*\pi}$ can be easily computed as follows.
\begin{align}
|M|_{*\varepsilon}\, &=\, \max_{|v|_\infty,|w|_\infty\leq 1} v^T M w\, =\, \max_{s\in \{\pm 1\}^n,\, t\in \{\pm 1\}^m} s^T M t\, \coloneqq\, \|M\|_{\infty\rightarrow 1}, \\
|M|_{*\pi}\, &=\, |M|_{\varepsilon*}\, =\, |M|_{\infty *}\, =\, |M|_1\, ,
\end{align}
where in the second line we used
\bqq
|N|_{\varepsilon}\, =\, \max_{|v|_1, |w|_1\leq 1} v^T N w\, =\, \max_{i,j} |N_{ij}|\, =\, |N|_\infty\, ,
\eqq
so that $|M|_{\varepsilon*}=|M|_1= \sum_{i,j} |M_{ij}|$, where we extended the definitions of $\infty$-norm and $1$-norm to matrices in an obvious way.

Now, we have to find the largest ratio $|M|_1/ \|M\|_{\infty\rightarrow 1}$ that is achievable by $M\in \mathds{R}^{n\times m}$. The argument comprises two parts: first, we have to exhibit an explicit example $M_0$ displaying a high ratio, and secondly, we have to show that this is optimal up to a constant factor.

Let us start with the first task. From now on, we will assume without loss of generality $n\leq m$. Suppose that there exists a matrix $H\in \{\pm 1\}^{n\times m}$ made of signs such that $HH^T=m\mathds{1}_n$ (such matrices are often called {\it partial Hadamard matrices}). Then simple considerations very close in spirit to Lindsey's lemma~\cite{Erdos74,cut norm H} (see also~\cite{A&S,Alon sparse}) show that for all $s\in \{\pm 1\}^n$ and $t\in \{\pm 1\}^m$ we have
\bqq
s^T H t\, \leq\, \sum_{j} \bigg| \sum_i s_i H_{ij}\bigg|\, \leq\, \sqrt{m}\, \sqrt{\sum_j \bigg| \sum_i s_i H_{ij}\bigg|^2}\, =\, \sqrt{m}\, \sqrt{\sum_{ijk} s_i s_k H_{ij} H_{kj}}\, =\, \sqrt{m}\, \sqrt{\sum_{ik} s_i s_k \, n\delta_{ik}}\, =\, n\sqrt{m}\, .
\eqq
Maximising over sign vectors $s,t$ yields $\|H\|_{\infty\rightarrow 1}\leq n\sqrt{m}$. Since $|H|_1=nm$, we obtain the lower bound $|H|_1 \big/ \|H\|_{\infty\rightarrow 1}\geq \sqrt{n}$. Now, the existence of a matrix $H$ with the properties required for the above construction is not guaranteed for arbitrary integers $n,m$. In fact, this is never the case when $n=m>2$ and $n$ is not a multiple of $4$, while when $n=m=4k$ it is the content of the so-called {\it Hadamard conjecture}. However, observe that Hadamard matrices are elementarily guaranteed to exist for $n=m=2^k$, as the explicit example $\left( \begin{smallmatrix} 1 & 1 \\ -1 & 1 \end{smallmatrix} \right)^{\otimes k}$ shows. Since for all $n$ there is $n'$ power of $2$ such that $\frac{n}{2}\leq n'\leq n$, we get $|H|_1 \big/ \|H\|_{\infty\rightarrow 1}\geq \sqrt{n'}\geq \sqrt{\frac{n}{2}}$. This shows that
\bq
\max_{0\, \neq\, M\, \in\, \mathds{R}^{n\times m}} \frac{|M|_{*\pi}}{|M|_{*\varepsilon}}\, =\, \max_{0\, \neq\, M\, \in\, \mathds{R}^{n\times m}} \frac{|M|_1}{\|M\|_{\infty\rightarrow 1}}\, \geq\, \sqrt{\frac{n}{2}}\, .
\label{lower Hadamard}
\eq
We note in passing that the crude estimate~\eqref{lower Hadamard} can be improved to $\sqrt{n}$ when $n\leq \frac{m}{2}$ by the same tricks. Further asymptotic refinements can be obtained with the help of sophisticated results such as~\cite{half Hadamard}, but this is beyond the scope of the present paper.

Now, we have to show that the above lower bound is optimal, i.e. that the scaling of the two sides of~\eqref{lower Hadamard} are the same. In order to do so, we resort to a celebrated inequality known as Khintchine inequality~\cite{bound B,best constant 1,best constant 2}, which states that for all $c\in \mathds{R}^n$, if $s_1,\ldots, s_n$ is an i.i.d. sequence with Rademacher distribution, then
\bq
\mathds{E}_s \bigg|\sum_i s_i c_i \bigg|\, \geq\, \frac{|c|_2}{\sqrt{2}}\, ,
\label{K ineq}
\eq
where $|c|_2^2=\sum_i c_i^2$. Thanks to this result, we know that whenever $M\in\mathds{R}^{n\times m}$ we have~\cite[Corollary 2.4]{Alon sparse}
\begin{align*}
\|M\|_{\infty\rightarrow 1}\, &=\, \max_{s\in \{\pm 1\}^n,\, t\in \{\pm 1\}^m} s^T M t\, =\, \max_{s\in \{\pm 1\}^n} \sum_j \bigg| \sum_i s_i M_{ij} \bigg| \\[0.5ex]
&\geq\ \mathds{E}_s \sum_j \bigg|\sum_i s_i M_{ij}\bigg|\ =\ \sum_j \mathds{E}_s \bigg|\sum_i s_i M_{ij}\bigg|\ \geq\ \frac{1}{\sqrt{2}}\, \sum_j \sqrt{\sum_i M_{ij}^2} \\
&\geq\ \frac{1}{\sqrt{2n}}\,\sum_{ij} |M_{ij}|\ =\ \frac{1}{\sqrt{2n}}\,|M|_1\, .
\end{align*}
This is nothing but the complementary upper bound to~\eqref{lower Hadamard}, and reads
\bq
\max_{0\, \neq\, M\, \in\, \mathds{R}^{n\times m}} \frac{|M|_{*\pi}}{|M|_{*\varepsilon}}\, =\, \max_{0\, \neq\, M\, \in\, \mathds{R}^{n\times m}} \frac{|M|_1}{\|M\|_{\infty\rightarrow 1}}\, \leq\, \sqrt{2n}\, .
\label{upper Hadamard}
\eq
The two inequalities~\eqref{lower Hadamard} and~\eqref{upper Hadamard} together solve the problem of data hiding in the cubic model. We summarise this solution as follows.

\vspace{2ex}
\begin{prop} \label{dh cubic}
Let the composite of two cubic models $A=G_n$, $B=G_m$ be given by the minimal tensor product, $AB=G_n\tminit G_m$. Then the restricted data hiding ratio $\widebar{R}_{G_n\tminitfoot G_m}$ defined via Proposition~\ref{norms centr} satisfies
\bq
\sqrt{\frac{\min\{n,m\}}{2}}\, \leq\, \widebar{R}_{G_n\tminitfoot G_m}\, \leq\, \sqrt{2\, \min\{n,m\}}\, .
\eq
Consequently, all the data hiding ratios against locally constrained measurements scale as 
\bq
R_{G_n \tminitfoot G_m}(\mathcal{M})\, =\, \Theta\left(\sqrt{\min\{n,m\}}\right)\, =\, \Theta\left(\sqrt{\min\{d_A,d_B\}}\right) , \qquad \text{for $\mathcal{M}=\text{\emph{LO}},\, \text{\emph{LOCC}}_\rightarrow,\, \text{\emph{LOCC}},\, \text{\emph{SEP}}$.}
\eq
\end{prop}

\end{ex}

\subsection{Werner construction for symmetric models} \label{subsec Werner}

A remarkable feature possessed by all the examples of GPTs we have seen so far, including the two that are undoubtedly the most physically relevant, i.e. classical probability theory (Example~\ref{ex class}) and quantum theory (Example~\ref{ex QM}), is the existence of a wide group of symmetries, i.e. linear transformations sending states to states. In the classical case, these are just permutations of the entries of the probability vector, while in the quantum case a symmetry is any conjugation by a unitary matrix. As for the spherical and cubic models (Examples~\ref{ex sph} and~\ref{cubic}), there are the natural actions of $SO(d-1)$ and of signed permutations on the last $d-1$ components of the state vector. We will show how these symmetries can be exploited in order to define a relevant class of bipartite states in a theory made of two copies of the same symmetric GPT. Let us start with the following definition.

\vspace{2ex}
\begin{Def}[Completely symmetric models] \label{compl symm}
A GPT $(V,C,u)$ is said to be \emph{completely symmetric} if there is a compact group $G$ and a representation $\zeta: G\rightarrow \mathcal{L}(V)$ that: (i) is such that $\zeta_g$ sends normalised states to normalised states, for all $g\in G$; and (ii) is irreducible and nontrivial on the preserved subspace $V_0\coloneqq \{x\in V:\, u(x)=0\}$.
\end{Def}

We will see at the end of this Subsection that all the examples of GPTs we have considered so far (namely, classical and quantum theory, and spherical and cubic model) are in fact completely symmetric. However, for now we are interested in keeping the reasoning as abstract and general as possible.
Thus, consider a completely symmetric GPT as in Definition~\ref{compl symm}. We start by listing some elementary consequences of the existence of a group symmetry. Let us remind the reader that if $\zeta:G\rightarrow \mathcal{L}(V)$ is a representation, then the dual space $V^*$ is naturally endowed with the dual representation $\zeta^*:G\rightarrow \mathcal{L}(V^*)$ given by $\zeta^*_g=\zeta_{g^{-1}}^T$. With this notation, the fact that $V_0$ is preserved by the action of $G$ can be written as $\zeta_g^*u=u$ for all $g\in G$. In other words, $\mathds{R}u$ is a trivial component of $V^*$ under the action of $G$ through $\zeta^*$. Since it is well-known that if $V$ is real and $G$ is compact then $\zeta$ and $\zeta^*$ are isomorphic, there must necessarily exist also a $G$-invariant vector $u_*\in V$. We deduce from the fact that $V_0$ is irreducible and nontrivial that $u_*\notin V_0$, so that we are free to rescale it in such a way that $u(u_*)=1$. Note that this fixes the decomposition of $V$ into $G$-irreducible representations as
\bq
V\, =\, \mathds{R}u_*\oplus V_0\, , \label{decomp V}
\eq
where the two pieces are non-isomorphic. Now, we proceed to show that $u_*$ is in fact a state, as implied by property (i) in Definition~\ref{compl symm}. In order to do so, we use the Haar integral on $G$, denoted by $\int_G dg$ and whose existence is guaranteed by the compactness of $G$. For every state $\omega\in\Omega$, we have that $u_*=\int_G dg\, \zeta_g \omega$, because the right-hand side is $G$-invariant and normalised and therefore must coincide with $u_*$. This expression for $u_*$ as a positive combination of states $\zeta_g \omega$ reveals that $u_*$ is itself a state.

As is easy to see by referring to~\eqref{decomp V}, the dual vector space $V^*$ decomposes as
\bq
V^*\, =\, \mathds{R}u\oplus V_0^* \label{decomp V*}
\eq
under the action of $G$, where the first addend is a trivial representation and the second one is $G$-isomorphic to $V_0$. Therefore, any $G$-isomorphism $\chi:V^*\rightarrow V$ must map $u$ into a multiple of $u_*$ and $V_0^*$ into $V_0$.
If we think of $\chi$ as a tensor in $V\otimes V$, we can write $\chi\,=\,\alpha\, u_*\!\otimes u_* + \beta\, \mathcal{E} = \alpha U_* +\beta \mathcal{E}$, where we used the shorthand $U_*\coloneqq u_*\otimes u_*$, and $\mathcal{E} \in V_0\otimes V_0$ is (canonically identified with) a fixed $G$-isomorphism $V_0^*\rightarrow V_0$. Depending on the representation we choose, we can alternatively write $(\zeta_g\otimes \zeta_g)(\mathcal{E})=\mathcal{E}$ for all $g\in G$ or $\mathcal{E} \zeta_g^*=\zeta_g \mathcal{E} $ for all $g\in G$. From this latter expression we see that $\mathcal{E}_*\coloneqq (\mathcal{E}^{-1})^T:V\rightarrow V^*$ is also a $G$-isomorphism. As for the tensor $\mathcal{E}_*\in V_0^*\otimes V_0^*$, we will rephrase this as $(\zeta_g^*\otimes \zeta_g^*)(\mathcal{E}_*)=\mathcal{E}_*$ for all $g\in G$.

It is perhaps convenient to think of $\mathcal{E}$ and $\mathcal{E}_*$ also as scalar products $\braket{\cdot,\cdot}_\mathcal{E}$, $\braket{\cdot,\cdot}_{\mathcal{E}_*}$ on $V_0^*$ and $V_0$, respectively. This can be done via the definitions $\braket{f,g}_\mathcal{E} \coloneqq (f\otimes g)(\mathcal{E})$ and $\braket{v,w}_{\mathcal{E}_*}\coloneqq \mathcal{E}_*(v\otimes w)$. It is well known that it is possible to choose $\mathcal{E}$ such that both these scalar product are positive definite. We will always make this assumption throughout the rest of this Section. If $\{f_i\}_{i=1}^{d-1}$ is an orthonormal basis for $\braket{\cdot,\cdot}_\mathcal{E}$, we will have
\bq
\mathcal{E}\, =\, \sum_{i=1}^{d-1} f_i^*\otimes f_i^*\, \in\, V_0\otimes V_0\, ,\qquad \mathcal{E}_*\,=\, \sum_{i=1}^{d-1} f_i\otimes f_i\,\in\, V_0^*\otimes V_0^*\, ,
\label{varepsilon expl}
\eq
with $\{f_i^*\}_i$ being the dual basis to $\{f_i\}_i$. As a simple consequence, we see that $\mathcal{E}_*(\mathcal{E})=d-1$, and the two norms induced by the above scalar products are dual to each other. Moreover, observe that the completion $\{u,f_1,\ldots,f_d\}$ is an orthonormal basis for a global $G$-invariant scalar product. We are now ready to give the following definition.

\vspace{2ex}
\begin{Def}[Werner states] \label{Werner def}
For a composite $AB$ made of two copies $A,B$ of the same completely symmetric GPT $(V,C,u)$ and such that~\eqref{CAB bound} is obeyed, a \emph{Werner state} is a normalised state in $V\otimes V$ that corresponds to a $G$-isomorphism $V^*\rightarrow V$. 
\end{Def}

\vspace{2ex}
With the language developed throughout the above discussion, we can express a generic Werner state as
\bq
\chi_\varphi\, =\, U_* + \varphi\, \mathcal{E}\, ,
\label{Werner}
\eq
where $\varphi$ is a real parameter whose range depends on specific features of the model. For any bipartite cone $C_{AB}$ satisfying~\eqref{CAB bound}, let us define
\bq
k_\pm\, \coloneqq\, \max\left\{ k:\, \chi_{\pm k}\in C_{AB} \right\}\, ,
\label{k+-}
\eq
so that the allowed range of $\varphi$ in~\eqref{Werner} will be $[-k_-,k_+]$. Besides the obvious observation that $k_\pm\geq 0$, there seems to be nothing we can say a priori about these parameters. In order to proceed further, we need a little lemma.

\vspace{2ex}
\begin{lemma} \label{lemma G projector}
Let $(V,C,u)$ be a completely symmetric GPT with Werner states given by~\eqref{Werner}. If $\mathcal{E}, \mathcal{E}_*$ are given by~\eqref{varepsilon expl}, the identity
\bq
\int_G dg\ \zeta_g\otimes \zeta_g\, =\, U_* U\, +\, \frac{1}{d-1}\ \mathcal{E}\, \mathcal{E}_*
\label{G projector}
\eq
between linear operators on $V\otimes V$ holds true. Here, $U\coloneqq u\otimes u$, $U_*\coloneqq u_*\otimes u_*$, and $u_*u:V\rightarrow V$ acts as $(u_* u)(x)=u(x)\, u_*$ for all $x\in V$.
\end{lemma}

\begin{proof}
For an arbitrary $X\in V\otimes V$, the properties of the Haar measure guarantee that $\int_G dg\, \zeta_g\otimes \zeta_g (X)$ is a $G$-invariant tensor, and thus can be expressed as $\alpha\, U_* + \beta\, \mathcal{E}$ for some $\alpha,\beta\in \mathds{R}$. Applying $U$ on both sides we see that
\bqq
\alpha\, =\, \int_G dg\, U\left(\zeta_g\otimes \zeta_g (X)\right)\, =\, \int_G dg \left(\zeta_g^* \otimes \zeta_g^*\, U\right)(X)\, =\, \int_G dg\, U(X)\, =\, U(X)\, .
\eqq
Moreover, it is easy to verify that since $\mathcal{E}_*(\mathcal{E})=d-1$ and $\left(\zeta_g^*\otimes \zeta_g^*\right)(\mathcal{E}_*)=\mathcal{E}_*$ for all $g\in G$, the analogous equality $(d-1)\beta=\mathcal{E}_*(X)$ holds, thus completing the proof of~\eqref{G projector}.
\end{proof}

\vspace{2ex}
With Lemma~\ref{lemma G projector} in our hands, we can say a bit more about the constants $k_\pm$ introduced in~\eqref{k+-}. Namely, we can determine lower bounds that correspond to the {\it separability region} for the family of Werner states~\eqref{Werner}. As expected, these lower bounds will depend only on the local structure of the model, not on the particular choice of the composite cone, which will instead affect the values of $k_\pm$.

\vspace{2ex}
\begin{prop} \label{prop sep Werner}
For a completely symmetric GPT $(V,C,u)$ of dimension $d$ and a scalar product $\braket{\cdot,\cdot}_{\mathcal{E}_*}$ on the corresponding $V_0$, introduce the quantities
\bq
m_+ \coloneqq\, \max\left\{\braket{v,v}_{\mathcal{E}_*}:\ v\in V_0,\, u_*+v\in C\right\}, \qquad m_- \coloneqq\, \max\left\{ \braket{v,w}_{\mathcal{E}_*}:\ v,w\in V_0,\, u_*+v, u_*-w\in C\right\}\, .
\label{m+-}
\eq
Then the constants $k_\pm$ defined through~\eqref{k+-} satisfy $k_\pm\geq \frac{m_\pm}{d-1}$, and a Werner state $\chi_\varphi$ as given in~\eqref{Werner} is separable if and only if $-\frac{m_-}{d-1}\leq \varphi\leq \frac{m_+}{d-1}$.
\end{prop}

\begin{proof}
Take $v\in V_0$ such that $u_*+v\in C$ and $\braket{v,v}_{\mathcal{E}_*}=m_+$. Compute
\begin{align*}
\int_G dg\ (\zeta_g\otimes \zeta_g)\left((u_*+v)\otimes (u_*+v)\right)\, &= \left( U_* U\, +\, \frac{1}{d-1}\ \mathcal{E}\, \mathcal{E}_* \right)\left((u_*+v)\otimes (u_*+v)\right) \\
&=\, U_*\, +\, \frac{\mathcal{E}_*(v\otimes v)}{d-1}\ \mathcal{E} =\, U_*\, +\, \frac{\braket{v,v}_{\mathcal{E}_*}}{d-1}\ \mathcal{E}\, =\, U_*\, +\, \frac{m_+}{d-1}\ \mathcal{E}\, .
\end{align*}
From the leftmost side we see that this is an allowed separable state of the bipartite system. Looking at the rightmost side and comparing it to~\eqref{Werner}, we see that $k_+\geq \frac{m_+}{d-1}$ and that $\chi_{m_+/(d-1)}\in C_A\tmin C_B$. An analogous construction shows that we can find $v,w\in V_0$ such that $u_*+v, u_*-w\in C$ and
\bqq
\int_G dg\ \zeta_g\otimes \zeta_g \left((u_*+v)\otimes (u_*-w)\right)\, =\, U_*\, -\, \frac{m_-}{d-1}\ \mathcal{E}\, \in\, C_A\tmin C_B\, ,
\eqq
so that $k_-\geq \frac{m_-}{d-1}$ and $\chi_{-k_-}\in C_A\tmin C_B$.

In order to show that $\chi_\varphi\in C_A\tmin C_B$ if and only if $\varphi\in \left[- \frac{m_-}{d-1},\, \frac{m_+}{d-1}\right]$, we start by proving that $\omega\in C_A\tmin C_B$ implies $-m_-\leq \mathcal{E}_*(\omega)\leq m_+$. In fact, if $\omega=\sum_i p_i\, (u_*+v_i)\otimes (u_*+w_i)$ we obtain $\mathcal{E}_*(\omega)\, =\, \sum_i p_i\, \braket{v_i,w_i}_{\mathcal{E}_*}$, and the claim follows from the inequalities
\bqq
\braket{v_i,w_i}_{\mathcal{E}_*}\, =\, - \braket{v_i,-w_i}_{\mathcal{E}_*}\, \geq\, -m_-\, ,
\eqq
valid because $u_*+v_i,u_*-(-w_i)\in C$ and thus $\braket{v_i,-w_i}_{\mathcal{E}_*}\leq k_-$, and 
\bqq
\braket{v_i,w_i}_{\mathcal{E}_*}\, \leq\, \sqrt{\braket{v_i,v_i}_{\mathcal{E}_*} \braket{w_i,w_i}_{\mathcal{E}_*}}\, \leq\, m_+\, .
\eqq
Applying this to a Werner state $\chi_\varphi$ we obtain that a necessary condition for separability is $-m_-\leq \mathcal{E}_*(\chi_\varphi)=(d-1)\varphi\leq m_+$, concluding the proof.
\end{proof}

\vspace{2ex}
In order to study the data hiding properties of the family of Werner states, we need to understand the separability conditions at the dual level. Luckily enough, there is no need to repeat any calculation. This is because if $(V,C,u)$ is completely symmetric then $(V^*,C^*,u_*)$ is itself a completely symmetric GPT. Therefore, we start by giving the following `dual' definitions:
\bq
k_\pm^*\, \coloneqq\, \max\left\{ k:\, U\pm k \mathcal{E}_*\in C^*_{AB} \right\}\, ,
\label{k+-*}
\eq
\vspace{-4.5ex}
\bq
m_+^* \coloneqq\, \max\left\{\braket{f,f}_{\mathcal{E}}:\ f\in V_0^*,\, u+f\in C^*\right\}, \qquad m_-^* \coloneqq\, \max\left\{ \braket{f,g}_{\mathcal{E}}:\ f,g\in V_0^*,\, u+f, u-g\in C^*\right\}\, .
\label{m+-*}
\eq
Exactly as in Proposition~\ref{prop sep Werner}, we obtain
\bq
k_\pm^*\, \geq\, \frac{m_\pm^*}{d-1}\, .
\label{k+-* lower bound}
\eq
Moreover, one can prove the following.

\vspace{2ex}
\begin{lemma}
In a completely symmetric GPT $(V,C,u)$, the constants $k_\pm,k_\pm^*,m_\pm,m_\pm^*$ defined by~\eqref{k+-},~\eqref{k+-*},~\eqref{m+-},~\eqref{m+-*}, respectively, satisfy $m_-\leq m_+$, $m_-^*\leq m_+^*$, and moreover
\bq
k_\pm^*\,k_\mp\, =\, \frac{1}{d-1}\, ,\qquad 1\,\leq\, m_\pm m_\mp^*\, \leq\, d-1\, .
\label{k m dual}
\eq
\end{lemma}

\begin{proof}
The inequalities $m_-\leq m_+$ and $m_-^*\leq m_+^*$ follow trivially from the definition by applying once Cauchy-Schwartz inequality. Thus, let us proceed to show the relation between $k_\pm$ and $k_\pm^*$. By definition of dual cone, $U\pm k \mathcal{E}_*\in C^*_{AB}$ if and only if $(U+ k \mathcal{E}_*)(\omega)\geq 0$ for all $\omega \in C_{AB}$, which can be assumed to be normalised without loss of generality. However, because of the group symmetry it is enough to test Werner states. In fact, since $(\zeta_g^*\otimes \zeta_g^*)(U+ k \mathcal{E}_*)=U+ k \mathcal{E}_*$ for all $g\in G$, we obtain
\bqq
(U+ k \mathcal{E}_*)(\omega)\, =\, \int_G dg\, \left((\zeta_g^*\otimes \zeta_g^*)(U+ k \mathcal{E}_*)\right) (\omega)\, =\, \int_G dg\ (U+ k \mathcal{E}_*) \left((\zeta_g\otimes \zeta_g) (\omega) \right)\, =\, (U+ k \mathcal{E}_*) (\chi_\varphi)\, =\, 1 + (d-1) k \varphi\, ,
\eqq
where $(d-1)\varphi=\mathcal{E}_*(\omega)$. The rightmost side of the above equation is positive for all $\varphi\in [-k_-,k_+]$ if and only if $-\frac{1}{(d-1) k_+}\leq k\leq \frac{1}{(d-1) k_-}$.

Now, let us devote our attention to the bounds on the $m$ quantities in~\eqref{k m dual}. Clearly, it is enough to show that $1\leq m_+ m_-^*\leq d-1$, up to exchanging primal and dual GPT. Consider $f\in V_0^*$ such that $\braket{f,f}_{\mathcal{E}}=1$. Then it is easy to see that $u-\frac{f}{\sqrt{m_+}}\in C^*$, since for all states $\omega\in \Omega$, which can be conveniently parametrised as $\omega=u_*+v$ with $v\in V_0$ satisfying $\braket{v,v}_{\mathcal{E}_*}\leq m_+$, one has
\bqq
\Big(u \pm \frac{f}{\sqrt{m_+}}\Big)(\omega)\, =\, 1 \pm \frac{1}{\sqrt{m_+}}\, f(v)\, \geq\, 1 - \frac{1}{\sqrt{m_+}}\, \sqrt{\braket{f,f}_{\mathcal{E}} \braket{v,v}_{\mathcal{E}_*}}\, \geq\, 1 - \sqrt{\braket{f,f}_{\mathcal{E}}}\, =\, 0\, .
\eqq
Since $u\, \pm \frac{f}{\sqrt{m_+}}\in C^*$, from the definition~\eqref{m+-*} we infer that $m_-^*\geq \braket{\frac{f}{\sqrt{m_+}},\frac{f}{\sqrt{m_+}}}_{\mathcal{E}}=\frac{1}{m_+}$. In order to prove the upper bound $m_+m_-^*\leq d-1$, let us write
\bqq
m_+\, \leq\, (d-1) k_+\, =\, \frac{1}{k_-^*}\, \leq\, \frac{d-1}{m_-^*}\, ,
\eqq
where we employed in order: the results of Proposition~\ref{prop sep Werner}, the first relation in~\eqref{k m dual} and finally~\eqref{k+-* lower bound}.
\end{proof}

\vspace{2ex}
With the tools we have developed so far, we are ready to discuss quantitatively the existence of allowed and separable measurements displaying Werner symmetry.

\vspace{2ex}
\begin{prop} \label{prop dual Werner}
In a completely symmetric GPT $(V,C,u)$, for $\alpha,\beta\in\mathds{R}$ the functionals $\left( \alpha U + \beta \mathcal{E}_*,\, (1-\alpha) U - \beta \mathcal{E}_* \right)$ form an allowed measurement if and only if
\bq
- k_-^*\alpha\, \leq\, \beta\, \leq\, k_+\alpha\, ,\qquad -k_+^*(1-\alpha)\,\leq\, \beta\, \leq\, k_-^*(1-\alpha)\, .
\label{dual Werner}
\eq
In the $(\alpha, \beta)$-plane, these two conditions identify a parallelogram with vertices
\bq
(0,0),\, (1,0),\ \left(\frac{k_-^*}{k_+^* + k_-^*},\ \frac{k_+^* k_-^*}{k_+^* + k_-^*} \right),\ \left(\frac{k_+^*}{k_+^* + k_-^*},\, -\frac{k_+^* k_-^*}{k_+^* + k_-^*} \right) .
\eq
The separability conditions for the measurement under examination can be deduced from~\eqref{dual Werner} by making the substitutions $k_\pm^*\mapsto \frac{m_\pm^*}{d-1}$.
\end{prop}

\begin{proof}
We omit the details, since the proof consists in a systematic application of the definitions~\eqref{k+-*}, together with the dual conditions to those already given in Proposition~\ref{prop sep Werner}.
\end{proof}

\vspace{2ex}
Equipped with Proposition~\ref{prop dual Werner}, we are ready to explore the data hiding properties of Werner states. Since we did not make any general claim about anything but separability, we will compute the highest data hiding efficiency against separable measurements that is obtainable by using Werner states in Definition~\ref{dh}.

\vspace{2ex}
\begin{thm}[Data hiding with Werner states] \label{dh Werner}
For a composite system made of two copies of the same completely symmetric GPT, the highest data hiding efficiency against separable measurements that is achievable with only Werner states is given by
\bq
R_{\text{\emph{Werner}}}(\text{\emph{SEP}})\, =\, \max_{0\neq (a,b)\in\mathds{R}^2}\, \frac{\| a U_* + b \mathcal{E} \|}{\,\| a U_* + b \mathcal{E} \|_{\text{\emph{SEP}}}}\ =\ 1\, +\, \frac{2}{k_+ + k_-}\, \max\left\{\frac{1}{m_+^*} - k_-,\, \frac{1}{m_-^*} - k_+ \right\} .
\label{dh Werner eq}
\eq
\end{thm}

\begin{proof}
Let us start by computing the base norm $\|a U_*+ b \mathcal{E}\|$, for $a,b\in \mathds{R}$. Thanks to Proposition~\ref{prop dual Werner}, we have to test just one nontrivial measurement, namely
\bqq
\left( \frac{k_-^*}{k_+^* + k_-^*}\, U + \frac{k_+^* k_-^*}{k_+^* + k_-^*}\, \mathcal{E},\ \frac{k_+^*}{k_+^* + k_-^*}\, U - \frac{k_+^* k_-^*}{k_+^* + k_-^*}\, \mathcal{E} \right) .
\eqq
Thus, using~\eqref{d norm} we find
\begin{align}
\| a U_*+ b \mathcal{E} \|\, &=\, \left| \left( \frac{k_-^*}{k_+^* + k_-^*}\, U + \frac{k_+^* k_-^*}{k_+^* + k_-^*}\, \mathcal{E}\right) \left( a U_* + b\mathcal{E}\right)\right| \, +\, \left|\left( \frac{k_+^*}{k_+^* + k_-^*}\, U - \frac{k_+^* k_-^*}{k_+^* + k_-^*}\, \mathcal{E}\right) \left( a U_* + b \mathcal{E}\right)\right|\, = \label{base Werner 1} \\[1ex]
&=\, \left| \frac{k_-^*}{k_+^* + k_-^*}\, a + \frac{k_+^* k_-^*}{k_+^* + k_-^*}\, (d-1) b\right|\, +\, \left| \frac{k_+^*}{k_+^* + k_-^*}\, a - \frac{k_+^* k_-^*}{k_+^* + k_-^*}\, (d-1) b\right|\, = \label{base Werner 2}\\[1ex]
&=\, \max\left\{ |a|,\ \frac{| 2(d-1)\,b\, k_+^* k_-^* + a (k_-^* - k_+^*) |}{k_+^* + k_-^*} \right\} , \label{base Werner 3}
\end{align}
where we used the elementary formula $|x+y|+|z-y|=\max\{ |x+z|,\, |2y+x-z| \}$ in the last step. As suggested by Proposition~\ref{prop dual Werner}, we can obtain the separability norm by replacing everywhere $k_\pm^*$ with $\frac{m_\pm^*}{d-1}$:
\bq
\| a U_*+ b \mathcal{E} \|_{\text{SEP}}\, =\, \max\left\{ |a|,\ \frac{| 2\,b\, m_+^* m_-^* + a (m_-^* - m_+^*) |}{m_+^* + m_-^*} \right\}\, . \label{sep Werner}
\eq
Observe that plugging~\eqref{k m dual} into~\eqref{base Werner 3} yields the somehow more handy expression
\bq
\| a U_*+ b \mathcal{E} \|\, =\, \max\left\{ |a|,\ \frac{| 2\,b\, + a (k_- - k_+) |}{k_+ + k_-} \right\} .
\label{base Werner}
\eq
Now, computing the ratio $\max_{0\neq (a,b)\in\mathds{R}^2}\, \frac{\| a U_* + b \mathcal{E} \|}{\,\| a U_* + b \mathcal{E} \|_{\text{SEP}}}$ amounts to minimising~\eqref{sep Werner} for a fixed value of~\eqref{base Werner}. It is very easy to see that such a minimum is always achieved for pairs $(a,b)$ such that the two expressions in the maximum appearing in~\eqref{sep Werner} coincide. Up to scalar multiples, there are exactly two such pairs, namely those satisfying $b=\pm \frac{a}{m_\mp^*}$. By substituting these values into~\eqref{base Werner} and dividing by~\eqref{sep Werner}, we find
\bq
R_{\text{Werner}}(\text{SEP})\, =\, \max_{0\neq (a,b)\in\mathds{R}^2}\, \frac{\| a U_* + b \mathcal{E} \|}{\,\| a U_* + b \mathcal{E} \|_{\text{SEP}}}\, =\, \max\left\{ 1,\ \frac{\big| \frac{2}{m_-^*}\, + k_- - k_+ \big|}{k_+ + k_-},\ \frac{\big| -\frac{2}{m_+^*}\, + k_- - k_+ \big|}{k_+ + k_-} \right\} . \label{dh Werner 1}
\eq
Now, using~\eqref{k m dual} it is easy to see that 
\bqq
\frac{2}{m_\pm^*} \pm (k_+-k_-)\, \geq\, \frac{2}{(d-1) k_\pm^*} \pm (k_+-k_-) \, =\, 2 k_\mp \pm (k_+-k_-)\, =\, k_+ + k_-\, .
\eqq
Thanks to the above inequalities, we can further simplify~\eqref{dh Werner 1} to
\begin{align*}
R_{\text{Werner}}(\text{SEP})\, &=\, \max\left\{ \frac{\big| \frac{2}{m_-^*}\, + k_- - k_+ \big|}{k_+ + k_-},\ \frac{\big| -\frac{2}{m_+^*}\, + k_- - k_+ \big|}{k_+ + k_-} \right\}\, =\\
&=\, \frac{1}{k_+ + k_-}\, \max\left\{ \frac{2}{m_+^*} + (k_+-k_-),\ \frac{2}{m_-^*} - (k_+-k_-) \right\}\, =\\
&=\, 1\, +\, \frac{2}{k_+ + k_-}\, \max \left\{ \frac{1}{m_+^*} - k_-,\ \frac{1}{m_-^*} - k_+ \right\} ,
\end{align*}
finally proving~\eqref{dh Werner eq}.
\end{proof}

\vspace{2ex}
The above result shows how the maximal data hiding ratio obtainable by employing only Werner states depends on just few geometric parameters characterising the model under examination. The usefulness of this theorem rests on its immediate applicability to several natural classes of highly symmetric GPTs. Since computing the relevant parameters is often a simple and intuitive task, as we shall see in a moment,~\eqref{dh Werner eq} gives a quick lower bound on all the data hiding ratios against locally constrained sets of measurements.
To demonstrate the power of Theorem~\ref{dh Werner}, we apply it to the symmetric models that we have examined so far: classical probability theory, quantum mechanics, spherical model, and cubic model.

\begin{itemize}

\item {\it Classical probability theory} (Example~\ref{ex class}). By looking at the definition~\eqref{classical}, it is easy to see that the relevant group here is the symmetric group $S_d$. We have $u=(1,\ldots,1)^T=d u_*$, and the action on the subspace $V_0=\{x\in\mathds{R}^d:\ \sum_i x_i=0\}$ is well-known to be irreducible. This can be either proved with elementary tools or shown in one line via character theory and Burnside's formula. In fact, denote with $\eta:S_d\rightarrow \mathds{R}^{d\times d}$ the standard representation of $S_d$ on $\mathds{R}^d$. Since we can exhibit an explicit one-dimensional irrep (vectors of constant entries), it is enough to show that $\frac{1}{d!} \sum_{\pi\in S_d} \left( \Tr \eta(\pi)\right)^2 = \langle \chi_\eta, \chi_\eta\rangle=2$. Observing that $\Tr \eta(\pi)$ is the number of elements fixed by $\pi$ and applying Burnside's formula to the natural action of $S_d$ on $\{1,\ldots, d\}^{\times 2}$, which has exactly $2$ orbits, we find $\frac{1}{d!} \sum_{\pi\in S_d} \left( \Tr \eta(\pi)\right)^2=2$.

Going back to the analysis of classical probability theory as a completely symmetric model, and exploiting the identification between elements in the tensor product and $d\times d$ real matrices we can also write $U=uu^T= d^2 U_*$ and $\mathcal{E}=\mathds{1}-\frac{1}{d}uu^T=\mathcal{E}_*$. Furthermore, observe that since the upper and lower bound in~\eqref{CAB bound} coincide, the bipartite cone is composed of all the entrywise positive matrices, collectively denoted by $\mathds{R}^{d\times d}_+$. Computing the relevant quantities is now elementary:
\begin{align*}
m_+^* &=\, \max\left\{v^Tv:\ u+v\in\mathds{R}^d_+,\ \sum\nolimits_i v_i=0 \right\}\, =\, d(d-1)\, =\, d^2 m_+\, ,\\
m_-^* &=\, \max\left\{v^Tw:\ u+v,\, u-w\in\mathds{R}^d_+,\ \sum\nolimits_i v_i=\sum\nolimits_j w_j = 0 \right\}\, =\, d\, =\, d^2 m_-\, ,\\
k_+ &=\, \max\left\{ k:\ u_*u_*^T+ k \left(\mathds{1} -d u_*u_*^T\right)\in \mathds{R}^{d\times d}_+\right\}\, =\, \frac{1}{d}\, ,\\
k_- &=\, \max\left\{ k:\ u_*u_*^T- k \left(\mathds{1} -d u_*u_*^T\right)\in \mathds{R}^{d\times d}_+\right\}\, =\, \frac{1}{d(d-1)}\, ,
\end{align*}
Applying~\eqref{dh Werner eq}, we find $R_{\text{Werner}}^{\text{Cl}}(\text{SEP})=1$, which is expected since the collapse of the hierarchy~\eqref{CAB bound} when either of the two cones is simplicial forbids the existence of data hiding altogether.

\item {\it Quantum mechanics} (Example~\ref{ex QM}). We refer to~\eqref{quantum} for the notation. The symmetry group in quantum mechanics is $U(n)$, the group of unitary $n\times n$ matrices. It is easy to see that we can choose $u=\mathds{1}=n\, u_*$ and that the orthogonal complement to this trivial representation (that is, the space of traceless hermitian matrices) is irreducible. In fact, this follows already from the same result for the classical case, since permutation matrices are also unitaries.

For a bipartite system we see that $U=\mathds{1}=n^2\, U_*$ and $\mathcal{E}=F-\frac{\mathds{1}}{n}=\mathcal{E}_*$, where $F\ket{\alpha\beta}=\ket{\beta\alpha}$ denotes the flip operator, as usual. If the composite is assembled according to the standard quantum mechanical rule~\eqref{cone bipartite quantum}, we find
\begin{align*}
m_+^* &=\, \max\left\{\Tr X^2:\ \mathds{1}+X\geq 0,\ \Tr X=0 \right\}\, =\, n(n-1)\, =\, n^2 m_+\, ,\\
m_-^* &=\, \max\left\{\Tr XY:\ \mathds{1}+X,\, \mathds{1}-Y\geq 0,\ \Tr X=\Tr Y=0 \right\}\, =\, n\, =\, n^2 m_-\, ,\\
k_+ &=\, \max\left\{ k:\ \mathds{1}/ n^2+ k \left( F-\mathds{1}/n \right)\geq 0 \right\}\, =\, \frac{1}{n(n+1)}\, ,\\
k_- &=\, \max\left\{ k:\ \mathds{1}/ n^2 - k \left( F-\mathds{1}/n \right)\geq 0 \right\}\, =\, \frac{1}{n(n-1)}\, ,
\end{align*}
from which it follows easily $R_{\text{Werner}}^{\text{QM}}(\text{SEP})=n$. According to Theorem~\ref{dh QM}, this is even the optimal data hiding ratio against all separable measurements.

If the composite is formed with the minimal tensor product rule (what we called `$W$-theory' in Section~\ref{sec ex}), we write instead
\bqq
k_+ =\, \frac{m_+}{n^2-1}\, =\, \frac{1}{n(n+1)}\, ,\qquad k_- =\, \frac{m_-}{n^2-1}\, =\, \frac{1}{n(n^2-1)}\, ,
\eqq
and thus $R_{\text{Werner}}^{\text{W}}(\text{SEP})=2n-1$, in accordance with Proposition~\ref{dh W}.

\item {\it Spherical model} (Example~\ref{ex sph}). Using the same notation as in~\eqref{spherical}, we see that in this case the relevant group is $O(d-1)$, the set of $(d-1)\times (d-1)$ orthogonal matrices acting on the last $d-1$ components of $\mathds{R}^d$. We can choose $u=(1,0,\ldots,0)^T=u_*$, and again the orthogonal complement is irreducible as follows from the same result for classical probability theory. For a bipartite system $U=uu^T=U_*$ and $\mathcal{E}=\hat{\mathds{1}}_{d-1}=\mathcal{E}_*$. If the composite is formed with the minimal tensor product rule, then it is very easy to verify that
\bqq
m_\pm^*=\, \max\{ v^Tv:\ u+v\in C_d\}\, =\, 1\, =\, m_\pm\, , \qquad k_\pm\, =\, \frac{m_\pm}{d-1}\, =\, \frac{1}{d-1}\, ,
\eqq
so that $R_{\text{Werner}}^{\text{Sph}}(\text{SEP})=d-1$, which coincides with the lower bound for the data hiding ratio against separable measurements given in Corollary~\ref{dh sph}.

\item {\it Cubic model} (Example~\ref{cubic}). We follow the notation previously established. The symmetry group of the cubic model in $d$ dimensions is simply the symmetry group of the $(d-1)$-dimensional hypercube. As for the spherical model, we have $u=(1,0,\ldots,0)^T=u_*$, $U=uu^T=U_*$ and $u^\perp$ is irreducible. Furthermore, we can choose $\mathcal{E}=\hat{\mathds{1}}_{d-1}=\mathcal{E}_*$. However, even for a minimal tensor product composite we have
\begin{align*}
m_\pm^*&=\, \max\{ v^Tv:\ v\in\mathds{R}^{d-1},\ |v|_1\leq 1 \}\, =\, 1\, ,\\
m_\pm &=\, \max\{ v^Tv:\ v\in\mathds{R}^{d-1},\ |v|_\infty\leq 1 \}\, =\, d-1\, ,\\
k_\pm &=\, \frac{m_\pm}{d-1}\, =\, 1\, ,
\end{align*}
and therefore $R_{\text{Werner}}^{\text{G}}(\text{SEP})=1$. This example shows how it might be the case that Werner states do not display any data hiding property, despite the fact that there is global data hiding in the cubic model, as shown by Proposition~\ref{dh cubic}.

\end{itemize}

\section{Conclusions}

We have presented a general theory of data hiding in bipartite systems composed of GPTs (general probabilistic theories), a framework comprising quantum states and probability distributions, and generally operational theories with states and measurements. It significantly extends ideas and results from quantum mechanics; in particular, we were able to determine the maximum so-called data hiding ratio in terms of the state space dimension, by making a connection between this problem and Grothendieck's tensor norms. This maximum is essentially attained for GPTs over spherical cones.

Inspired by the prominent role played by Werner states in quantum mechanical data hiding, we investigated Werner-like states in classes of theories with symmetric state spaces. By using these states as ansatzes, we could find a general lower bound on the data hiding ratio that depends on few geometrically meaningful parameters.

For quantum mechanics on finite-dimensional spaces we proved a new upper bound on the data hiding ratio against LOCC operations by exploiting the celebrated quantum teleportation protocol. We saw how this improves previous results and definitively settles the problem of determining the optimal data hiding ratio against LOCC for fixed local dimensions. However, the same problem for the smaller set of local operations remains open, as we showed that none of the bounds that we provide or that are available in the literature is generally optimal.

Perhaps surprisingly, data hiding ratios in quantum mechanics are not as large as the maximal conceivable ones, being of the order of the square root of the latter, which are exhibited by spherical cones. Thus, if one were to summarise the results of our research in one single sentence, one could say that Nature is non-classical, but not as non-classical as it could have been.
Although this conclusion has been found by many other authors in relation to different investigations into non-locality phenomena, it appears to us so just, that we offer it here, as the essence, so to speak, of the whole story.

Lastly, while here we have focused on the maximum data hiding ratios among $R(\text{SEP})$ for various local GPTs, the opposite end is very interesting, too, from a foundational point of view. Namely, it is clear that to have $R(\text{SEP})>1$, it is necessary that both cones $C_A$ and $C_B$ are non-classical in the sense that they can not be simplicial, for the reason that otherwise min- and max-tensor product of the cones would coincide. It is however not known that this is necessary, and one might conjecture that $R(\text{SEP})=1$ if and only if one of the cones $C_A$ or $C_B$ is simplicial. A general result along these lines is due to Namioka and Phelps~\cite{NP}, which says that for $C_B$ the cone of the so-called gbit, $R(\text{SEP})=1$ iff $C_A$ is simplicial. If the more general conjecture were true, it would make a good case for $R(\text{SEP})$ as a measure of non-classicality of a GPT. For a thorough discussion of this and related problems, see \cite[Chapter~II]{lamiatesi}.

\vspace{2ex}
{\it Acknowledgements.} We thank Guillaume Aubrun for many insightful discussions on functional analysis, in particular for having suggested the use of Auerbach's lemma in the proof of Proposition~\ref{prop proj=<ninj}. LL is also grateful to Howard Barnum and Matthias Christandl for convincing him that defining LOCC in arbitrary GPTs is possible.
LL and AW are indebted to John Calsamiglia and C\'ecilia Lancien for enlightening observations on data hiding and certain aspects of convex analysis, respectively. Finally, we thank the anonymous referees for useful comments on the first version of this paper.
We acknowledge financial support from the European Research Council (AdG IRQUAT No. 267386), the Spanish MINECO (Project no. FIS2013-40627-P, no. FIS2016-86681-P, and no. MTM2014-54240-P), the Generalitat de Catalunya (CIRIT Project no. 2014 SGR 966), and the Comunidad de Madrid (QUITEMAD+ Project S2013/ICE-2801). CP was partially supported by the ``Ram\'on y Cajal program" (RYC-2012-10449).

\appendix

\section{Equiprobable data hiding pairs} \label{app equi}

In this appendix we look into the consequences of restricting the definition of data hiding to equiprobable pairs of states. Let us start by stating the corresponding modified version of Definition~\ref{dh}.

\vspace{2ex}
\begin{Def} \label{equi dh}
Let $(V,C,u)$ be a GPT. For an informationally complete set of measurements $\mathcal{M} \subseteq \mathbf{M}$, we say that there is \emph{balanced data hiding against $\mathcal{M}$ with efficiency $\widetilde{R}\geq 1$} if there are two normalised states $\rho,\sigma\in\Omega$ such that the probability of error defined in~\eqref{pr error} satisfies
\begin{equation}
P_{e}^{\mathbf{M}}(\rho,\sigma; 1/2) = 0\, ,\qquad P_{e}^{\mathcal{M}}(\rho,\sigma; 1/2) = \frac12 \left( 1-\frac{1}{\widetilde{R}}\right) .
\label{equi dh eq}
\end{equation}
The highest balanced data hiding efficiency against $\mathcal{M}$ is called \emph{balanced data hiding ratio against $\mathcal{M}$} and will be denoted by $\widetilde{R}(\mathcal{M})$.
\end{Def}

\vspace{1ex}
\begin{rem}
Clearly, standard data hiding ratios are always larger or equal to the corresponding balanced versions, i.e. $R(\mathcal{M})\geq \widetilde{R}(\mathcal{M})$.
\end{rem}

Throughout the rest of this appendix, we will consider (balanced) data hiding against a fixed set of measurements $\mathcal{M}$ that we do not specify further. 
There are at least two reasons why the above Definition~\ref{equi dh} does not add much to the theory developed in the main text, at least from a fundamental perspective. The first such reason is that it turns out that highly efficient standard data hiding pairs of states are automatically almost equiprobable, as the following result establishes.

\vspace{2ex}
\begin{lemma}
Let $\rho,\sigma$ be states that obey~\eqref{dh eq} for some a priori probability $p\neq 1/2$. Then the standard data hiding efficiency $R$ satisfies 
\begin{equation}
R\, \leq\, \frac{1}{|2p-1|}\, .
\label{if high R then p almost 1/2}
\end{equation}
\end{lemma}

\begin{proof}
An obvious upper bound to the probability of error $P_e^\mathcal{M}(\rho,\sigma ; p)$ can be obtained by considering the strategy of guessing according to the maximum likelihood rule.
Since such strategy fails with probability $\min\{p,1-p\}$, we get
\bqq
P_e^\mathcal{M}(\rho,\sigma; p)\, \leq\, \min\{p,1-p\}\, =\, \frac{1-|2p-1|}{2}\, .
\eqq
Using the definition of efficiency~\eqref{dh eq}, one finds immediately
\bqq
R\, =\, \frac{1}{1-2P_e^{\mathcal{M}}(\rho,\sigma ;p)}\, \leq\, \frac{1}{|2p-1|}\, ,
\eqq
as claimed.
\end{proof}

\vspace{2ex}
While the above result establishes the `approximate equiprobability' of highly efficient data hiding pairs, balanced data hiding ratios still stand as a different concept, as in their definition we required the \emph{exact} equality $p=1/2$ from the start. 
In fact, it could indeed happen that $\widetilde{R}(\mathcal{M})<R(\mathcal{M})$ in some cases. However, we will show in a moment that anyway $R$ and $\widetilde{R}$ have to be of the same order. To this purpose, we need a modified version of Proposition~\ref{dh ratio}.

\vspace{2ex}
\begin{lemma} \label{equi dh ratio}
For an informationally complete set of measurements $\mathcal{M}\subseteq\mathbf{M}$ in an arbitrary GPT, the balanced data hiding ratio $\widetilde{R}(\mathcal{M})$ can be computed as
\begin{equation}
\widetilde{R}(\mathcal{M})\, =\, \max \left\{ \frac{\|y\|}{\|y\|_{\mathcal{M}}}:\ y\in V,\ y\neq 0,\ u(y)=0 \right\}\, .
\label{equi dh ratio eq}
\end{equation}
\end{lemma}

\begin{proof}
The argument follows closely that we gave in the proof of Proposition~\ref{dh ratio}. The only difference is that now one needs to characterise the set $\widetilde{K}$ of vectors $y\in V$ that can be expressed as $y=\frac12 \rho- \frac12 \sigma$ for some normalised states $\rho,\sigma\in \Omega$. It is not difficult to realise that $\widetilde{K} = \left\{y\in V:\ \|y\|\leq 1,\, u(y)=0 \right\}$, ultimately yielding~\eqref{equi dh ratio eq}.

The above characterisation can be justified as follows. On the one hand, if $y=\frac12\rho-\frac12 \sigma$ for $\rho,\sigma\in\Omega$ then $\|y\|\leq \frac12 \|\rho\|+\frac12 \|\sigma\| = \frac12+\frac12=1$, and moreover $u(y)=\frac12 (1-1)=0$. On the other hand, applying Lemma~\ref{dual base} as usual we can write any $y\in V$ such that $\|y\| = 1$ as a difference $y=y_+-y_-$, where $y_\pm \geq 0$ and $u(y_+)+u(y_-)=1$. Requiring also $u(y)=0$ leads us to the conditions $u(y_+)=u(y_-)=\frac12$, so that $y_+=\frac12 \rho$ and $y_-=\frac12 \sigma$ for some normalised states $\rho,\sigma\in\Omega$. This shows that all vectors $y\in V$ such that $\|y\| = 1$ and $u(y)=0$ belong to $\widetilde{K}$. Since the latter set is convex and contains $0$, we deduce that indeed $\left\{y\in V:\ \|y\|\leq 1,\, u(y)=0 \right\}\subseteq \widetilde{K}$, which is what was left to show.
\end{proof}

\vspace{2ex}
We are now ready to prove the following result, which concludes our analysis.

\vspace{2ex}
\begin{prop} \label{equi vs standard}
The balanced and standard data hiding ratios, as given by Definition~\ref{dh} and~\ref{equi dh}, respectively, satisfy the inequalities
\bq
\widetilde{R}(\mathcal{M})\, \leq\, R(\mathcal{M})\, \leq\, 2 \widetilde{R}(\mathcal{M}) + 1\, .
\label{equi vs standard eq}
\eq
\end{prop}

\begin{proof}
We already saw that balanced data hiding ratios are always upper bounded by their standard versions, so let us focus on the complementary bound. We start by picking a vector $x\in V$ that saturates the maximum in~\eqref{dh ratio eq}, i.e. that is such that $\|x\|_\mathcal{M}=1$ and $\|x\|=R$. Choose a normalised state $\omega\in \Omega$, and construct $y\coloneqq x-u(x) \omega$, so that $u(y)=0$. Observe that $|u(x)|\leq \|x\|_\mathcal{M}=1$. To estimate the norms $\|y\|$ and $\|y\|_{\mathcal{M}}$, one can apply repeatedly the triangle inequality:
\begin{align*}
\|y\|\, &\geq\, \|x\| - |u(x)|\, \geq\, R(\mathcal{M})-1\, ,\\
\|y\|_\mathcal{M}\, &\leq\, \|x\|_\mathcal{M} + |u(x)|\, \leq\, 2\, .
\end{align*}
Therefore, we can use Lemma~\ref{equi dh ratio} to write
\bqq
\widetilde{R}(\mathcal{M})\, \geq\, \frac{\|y\|}{\|y\|_\mathcal{M}}\, \geq\, \frac{R(\mathcal{M})-1}{2}\, ,
\eqq
which gives~\eqref{equi vs standard eq} upon elementary algebraic manipulations.
\end{proof}

\section{Quantum data hiding with Werner states} \label{app Werner}

Throughout this appendix, we briefly review the well-known techniques used to find lower bounds on data hiding ratios via Werner states, and apply them to prove~\eqref{bound Werner} and~\eqref{bound Werner W}. More precisely, we will show that those described in~\eqref{bound Werner} and~\eqref{bound Werner W} are the optimal data hiding pairs within the Werner class, defined as the real span of symmetric and antisymmetric projector~\eqref{symm antisymm proj}. Formally, this amounts to show that
\begin{equation}
\max_{(\alpha,\beta)\neq (0,0)} \frac{\|\alpha\rho_{S}+\beta\rho_{A}\|_{1}}{\|\alpha\rho_{S}+\beta\rho_{A}\|_{\text{SEP}}}\, =\, n\, ,\qquad \max_{(\alpha,\beta)\neq (0,0)} \frac{\|\alpha\rho_{S}+\beta\rho_{A}\|_{W}}{\|\alpha\rho_{S}+\beta\rho_{A}\|_{\text{SEP}}}\, =\, 2n-1\, .
\label{app Werner eq1}
\end{equation}
As a side remark, let us notice here that the above separability norms could be easily replaced with $\text{LOCC}$ norms, since the extremal measurements in the Werner class are well-known to be $\text{LOCC}$~\cite[Propositions 1 and 3]{VV dh Chernoff}.

Consider a binary measurement $(E,\mathds{1}-E)$, which we will choose later to be either standard, or a witness, or else an $\text{LOCC}$. Thanks to the fact that to estimate distinguishability norms we are going to compute quantities of the form $\Tr E X$ where $X$ is in the Werner class and thus invariant under the `twirling' operation $\mathcal{T}(\cdot)\coloneqq \int dU\, U\otimes U (\cdot) U^{\dag}\otimes U^{\dag}$, a very standard trick going back to the first paper on data hiding~\cite{dh original 1} allows us to conclude that we can suppose also the measurement elements $E,\mathds{1}-E$ to belong to the Werner class. For $E = a\mathds{1}+b F$, it is well-known that the measurement $(E,\mathds{1}-E)$ is:
\begin{itemize}
\item a standard quantum mechanical measurement iff $0\leq E\leq \mathds{1}$, i.e. iff $0\leq a\pm b\leq 1$;
\item separable iff $E,\mathds{1}-E$ are both separable operators~\eqref{separable}, iff $0\leq a - b\leq 1$ and $0\leq a + nb\leq 1$;
\item a $W$-theory measurement iff $E,\mathds{1}-E$ are both witnesses~\eqref{witnesses}, iff $0\leq a\leq 1$ and $0\leq a+b\leq 1$.
\end{itemize}
While the first two points are well-known results~\cite{Werner symmetry}, the latter deserves a quick comment. An operator $E$ is a witness iff $\braket{\alpha \beta| E|\alpha\beta}\geq 0$ for all states $\ket{\alpha},\ket{\beta}\in\mathds{C}^{n}$. We can thus use the properties of the flip operator to conclude that $E=a\mathds{1}+bF$ is a witness iff $a + pb\geq 0$ for all $p\in [0,1]$, which together with the analogous condition for $\mathds{1}-E$ reproduces the above constraints. Optimising over all measurements in the Werner class we conclude that
\begin{align}
\|\alpha\rho_{S}+\beta\rho_{A}\|_{1}\, &=\, |\alpha|+|\beta|\, , \label{app Werner eq2} \\
\|\alpha\rho_{S}+\beta\rho_{A}\|_{\text{SEP}}\, &=\, \frac{2}{n+1}\, |\alpha| + \left| \frac{n-1}{n+1}\alpha +\beta\right| , \label{app Werner eq3} \\
\|\alpha\rho_{S}+\beta\rho_{A}\|_{W}\, &=\, |\alpha-\beta|+2|\beta|\, . \label{app Werner eq4}
\end{align}
Computing the maxima in~\eqref{app Werner eq1} is now an elementary exercise.

\section{Comparison between $\|\cdot\|_W$ and $\|\cdot\|_2$} \label{app W norm}

Here, we argue that the methods in~\cite{VV dh} can not be applied directly to determine the data hiding ratio in $W$-theory against any locally constrained set of measurements. As detailed at the end of Example~\ref{ex QM}, to this purpose it is enough to show the following.

\vspace{2ex}
\begin{lemma}
For all $\delta>0$, there exists a sequence of operators $X_n$ acting on $\mathds{C}^n\otimes\mathds{C}^n$, such that $\|X_n\|_2=1$ but $\|X_n\|_W\geq \Omega\left(n^{3/2-\delta}\right)$.
\end{lemma}

\begin{proof}
From the results of~\cite{aspects generic entanglement} it is known that a random orthogonal projector $\Pi$ onto a subspace of $\mathds{C}^n\otimes \mathds{C}^n$ of dimension $k\gg n$ will satisfy
\bq
\braket{\alpha\beta | \Pi |\alpha \beta }\, \leq\, \frac{2k}{n^2} \qquad \forall\ \ket{\alpha}, \ket{\beta}\in \mathds{C}^n:\ \braket{\alpha|\alpha}=\braket{\beta|\beta}=1
\label{random proj}
\eq
with high probability as $n$ tends to infinity. Fixing $k=n^{1+2\delta}$ and picking one such $\Pi_n$ for all $n\in\mathds{N}$, we can construct $X_n\coloneqq n^{-1/2-\delta}\, \Pi_n$. By definition, $\|X_n\|_2=1$ for all $n\in\mathds{N}$. Moreover, defining $Y\coloneqq \mathds{1} - \frac{n^2}{k}\, \Pi_n$, by~\eqref{random proj} we have $-1\leq \braket{\alpha\beta| Y | \alpha\beta}\leq 1$ for all normalised $\ket{\alpha},\ket{\beta}$, i.e. $Y$ belongs to the dual unit ball of the  base norm pertaining to the minimal tensor product $\text{QM}_n\tmin\text{QM}_n$. Hence,
\bqq
\|X_n\|_W\, \geq\, \Tr \left[ X_n Y \right]\, =\, \frac{n^2-k}{\sqrt{k}}\, =\, \Omega\left(n^{3/2-\delta}\right)\, ,
\eqq
as claimed.
\end{proof}

\end{document}